\documentclass[10pt]{siamltex}
\usepackage{latexsym, amsmath, amssymb,amsopn,amsfonts,amscd}
\usepackage{epsfig,graphics,color,graphicx,graphpap}
\usepackage{verbatim}
\newcommand{\ZZ}{{\mathbb Z}}
\newcommand{\RR}{{\mathbb R}}

\newcommand{\nn}{{\nonumber}}

\newtheorem{remark}{Remark}[section]

\usepackage{url}       
\newcommand{\D}{\left\langle D_0 \right\rangle}
\newcommand{\DV}{\left\langle D_{V_0} \right\rangle}
\newcommand{\eps}{\epsilon}

\ifx\pdfoutput\undefined
 \DeclareGraphicsExtensions{.eps}
\else
 \ifx\pdfoutput\relax
  \DeclareGraphicsExtensions{.eps}
 \else
  \ifnum\pdfoutput>0
   \DeclareGraphicsExtensions{.pdf}
  \else
   \DeclareGraphicsExtensions{.eps}
  \fi
 \fi
\fi

\title{Scattering, homogenization and interface effects for oscillatory potentials with strong singularities} 

 \author{Vincent Duch\^ene\thanks{\'Equipe EDP, DMA - \'Ecole Normale Supe\'rieure 
 45, rue d'Ulm, 75230 Paris Cedex 05 - France.}
     \and Michael I. Weinstein\thanks{Department of Applied Physics and Applied Mathematics, Columbia University 
 200 S. W. Mudd, 500 W. 120th St., New York City, NY 10027, USA.}}

\begin{document}

\maketitle

\begin{abstract} 
We study one-dimensional scattering for 
a decaying potential with rapid periodic oscillations and strong localized singularities. In particular, we consider
the Schr\"odinger equation
\begin{equation}
 H_\epsilon\ \psi \equiv\ \left(\ -\partial_x^2+V_0(x)+q\left(x,x/\epsilon\right)\ \right)\psi=k^2\psi
\nn\end{equation} 
for $k\in\RR$ and $\epsilon\ll 1$. Here, $q(\cdot, y+1)=q(\cdot,y)$, has mean zero and $|V_0(x)+q(x,\cdot)|\to0$ as $|x|\to\infty$.
The distorted plane waves of $H_\epsilon$ are solutions of the form: $e_{V^\epsilon\pm}(x;k) = e^{\pm ikx}+u^s_\pm(x;k)$, $u^s_\pm$ outgoing as $|x|\to\infty$.
 We derive their $\epsilon$ small asymptotic behavior, from which the asymptotic behavior of scattering quantities such as the transmission coefficient, $t^\epsilon(k)$, follow.
 
 Let $t_0^{hom}(k)$ denote the homogenized transmission coefficient associated with the average potential $V_0$. 
If the potential is smooth, then classical homogenization theory
gives asymptotic expansions of, for example, distorted plane waves, and transmission and reflection coefficients. 
Singularities of $V_0$ or discontinuities of $q_\epsilon$ are ``interfaces'' across which a solution must satisfy 
 interface conditions (continuity or jump conditions). To satisfy these conditions it is necessary to introduce  {\it interface correctors},
 which are highly oscillatory in $\epsilon$.

 Our theory admits potentials which have discontinuities in the microstructure, $q_\epsilon(x)$ as well as strong singularities 
 in the background potential, $V_0(x)$. A consequence of our main results is that $t^\epsilon(k)-t_0^{hom}(k)$, the error in the homogenized transmission coefficient is (i) ${\mathcal O}(\epsilon^2)$ if $q_\epsilon$ is continuous and
  (ii) ${\mathcal O}(\epsilon)$ if $q_\epsilon$ has discontinuities.
Moreover, in the discontinuous case the correctors are highly oscillatory in $\epsilon$, {\it i.e.} $\sim \exp({2\pi i\frac{\nu}{\epsilon}})$, for $\epsilon\ll1$. Thus a first order corrector is not well-defined since $\epsilon^{-1}\left(t^\epsilon(k)-t_0^{hom}(k)\right)$ does not have a limit as $\epsilon\to0$. This expression may have limits which depend on the particular sequence through which $\epsilon$ tends to zero.  The analysis is based on a (pre-conditioned) Lippman-Schwinger equation, introduced in~\cite{GW:05}.
\end{abstract}

\begin{keywords} 
Schr\"odinger operator, transmission coefficient, scattering theory, interface effects, microstructure, homogenization
\end{keywords}

\begin{AMS}
35J10, 35P25, 35B40, 35B27 
\end{AMS}

\pagestyle{myheadings}
\thispagestyle{plain}
\markboth{V. DUCH\^ENE AND M. I. WEINSTEIN}{SCATTERING FOR OSCILLATORY POTENTIALS WITH STRONG SINGULARITIES}

\section{Introduction}\label{sec:introduction}

An important method for computing the effective properties of highly oscillatory media
 is the method of homogenization. The goal of homogenization is to approximate a highly oscillatory medium, described by a differential equation with oscillatory coefficients, by 
 an approximate and homogeneous medium, described by a ``homogenized'' differential equation with constant or slowly varying coefficients. In its regime of validity, the homogenized differential equation (i) predicts effective properties which are approximately those of the heterogeneous medium and (ii) is, by comparison with the full problem, much simpler to study either analytically or by numerical simulation.
 
 While the homogenized limit can often be obtained by a formal multiple scale expansion or by variational methods~\cite{BLP,JKO,Allaire:08,Tartar:09}, these expansions are typically valid in the {\it bulk medium}, away from boundaries, discontinuities or more singular sets of coefficients. 
  Indeed, solutions to elliptic operators with oscillatory coefficients on bounded domains have been shown to require boundary layer correctors, which are sensitive to the 
 manner in which the microstructure meets a boundary~\cite{Santosa-Vogelius,Moskow-Vogelius-1,Allaire-Amar:99, Gerard-VaretMasmoudi:08,Gerard-VaretMasmoudi:10} or interface 
 ~\cite{HPShen:07}. Furthermore, the importance of correctors to homogenization due to interface effects, boundary layers {\it etc.} is explored analytically and computationally, in the context of accurate estimation of scattering resonances in~\cite{GW:03,GW:05}. 
 
 In this article we study the scattering problem for 
 the one-dimensional time-independent Schr\"odinger equation
\begin{equation}
\label{schro}
\left(\ H_\epsilon - k^2\ \right)\ \psi\ \equiv\ \left(\ -\frac{d^2}{d x^2}+V^\epsilon(x)-k^2\ \right)\psi(x)=0.
\end{equation}
The potential, $V^\epsilon(x)=V_0(x)+q\left(x,x/\epsilon\right)$, is the sum of 
 a slowly varying part
 with smooth and singular components, $V_0=V_{reg}+V_{sing}$, 
 and a rapidly oscillatory part, $q_\epsilon(x)=q\left(x,x/\epsilon\right),\ \epsilon\ll1$. $V^\epsilon(x)$ is assumed to decay to zero as $x$ tends to infinity. 
 We also assume $V^\epsilon(x)\ge0$, a simple way to restrict to the case where $H_\epsilon$ has no discrete eigenvalues (bound states) and has only continuous spectrum (extended / radiation states). 
 The wave number, $k$, is fixed and we study the $\epsilon-$ small behavior.
 \medskip
 
 Many physically important scattering properties are not captured by leading order homogenization.  Line-widths and imaginary parts of scattering resonances are key to quantifying the lifetimes of metastable states in quantum systems, or in electro-magnetics, the leakage rates of energy from photonic structures; see~\cite{GW:03,GW:05} and references therein. In~\cite{GW:03,GW:05} it was shown that inclusion of even the first non-trivial correction due to microstructure can yield large improvements in the approximation of such scattering quantities. 
  Since, as we shall see, defects and singularities can be responsible for the {\it dominant} correctors and these contributions are not captured in smooth homogenization setting, we therefore seek a better understanding of 
 homogenization for wave / scattering problems in their presence. 
 %
 %
 In this paper we ask:\medskip
 
 \noindent {\it How are scattering properties, such as transmission and reflection coefficients, $t_\epsilon(k)$ and $r_\epsilon(k)$, influenced by interfaces, defects and singularities?}\medskip

 The heart of the matter is an asymptotic study of the {\it distorted plane waves}, solutions of 
 $(H_\epsilon-k^2)\psi=0$ of the form: 
 \begin{equation}
 e_{V^\epsilon\pm}(x;k) = e^{\pm ikx}+u^s_\pm(x;k),\ \ u^s_\pm\ {\rm outgoing\ as}\ |x|\to\infty,\ \ \text{for $\epsilon$ small.}
 \nn\end{equation}
 Consequences of our analysis include the following:
 \begin{enumerate}
 \item Theorem~\ref{thm:Expand-LS} provides a convergent expansion of the distorted plane waves of $H_Q=-\partial_x^2+V_0 +Q$, which is valid for a large class of perturbing potentials, $Q$, which may be pointwise large, but highly oscillatory (supported at high frequencies although not necessarily periodic). Theorem~\ref{prop:convergence-asymptotic-expansion-LS} is the corresponding expansion for the transmission coefficient $t[k;Q]$. By Proposition~\ref{prop:Tqeps-small} we can apply Theorems~\ref{thm:Expand-LS} and~\ref{prop:convergence-asymptotic-expansion-LS}  to $Q(x)=q_\epsilon(x)=q\left(x,x/\epsilon\right)$, where $q(x,y)$ is $1-$ periodic in $y$, decaying as $|x|\to\infty$, and satisfies Hypotheses~{\bf (V)}. 
 \item Theorem~\ref{thm:Texpansion} implies that:
  \subitem (i) $t^\epsilon(k)-t_0^{hom}(k)={\mathcal O}(\epsilon^2)$ if $q_\epsilon$ is continuous and
  \subitem (ii) $t^\epsilon(k)-t_0^{hom}(k)={\mathcal O}(\epsilon)$ if $q_\epsilon$ has discontinuities.\\
 For $q_\epsilon$ discontinuous {\it interface correctors}, which are highly oscillatory in $\epsilon$, enter the expansion;
 see the discussion in section~\ref{sechomo} concerning failure and restoration of interface conditions at singularities of $V_0$ or discontinuities of $q_\epsilon$. These correctors are related to the asymptotics of boundary layers arising in work 
  on homogenization of divergence form operators on bounded domains~\cite{Santosa-Vogelius,Moskow-Vogelius-1,Allaire-Amar:99,Gerard-VaretMasmoudi:08,Gerard-VaretMasmoudi:10}.
Since these correctors involve $\epsilon$ dependence of the form: $\sim \exp({2\pi i\frac{\nu}{\epsilon}}),\ \ \epsilon\ll1,\ \ 0\ne\nu\in\mathbb{R}$, the expression $\epsilon^{-1}\left(t^\epsilon(k)-t_0^{hom}(k)\right)$ does not have a limit as $\epsilon\to0$,
 and a correction to the value of $t_0^{hom}(k)$ is not well-defined. However, there can be limits which depend on the particular sequences through which $\epsilon$ tends to zero. See the more detailed discussion after the statement of 
Theorem~\ref{thm:Texpansion-LS}.
\end{enumerate}

\noindent {\bf Outline of paper:} In section~\ref{sec:mainresults} we state detailed hypotheses and our main theorems on transmission coefficients, Theorems~\ref{thm:Texpansion} and~\ref{thm:Texpansion-LS}, which depend on our analysis of distorted plane waves (Theorem~\ref{thm:Expand-LS}). We also present the results of numerical simulations designed to illustrate the relationship between regularity of the potential, $V^\epsilon$, and $\epsilon$ small asymptotics of the transmission coefficient, stated in Theorem~\ref{thm:Texpansion}. In section~\ref{sec:background&statement} we present the technical background on one-dimensional scattering theory. In section~\ref{sechomo} we derive, by including {\it interface correctors} to an expansion derived by the classical method of multiple scales, 
 an expansion of the distorted plane waves and of the transmission coefficient valid to all orders in the small parameter 
 $\epsilon$ . Section~\ref{sec:rigorous-theory} contains rigorous proofs of the expansion of the distorted plane waves (Theorem~\ref{thm:Expand-LS}) and transmission coefficients (Theorem~\ref{thm:Expand-LS} and~\ref{prop:convergence-asymptotic-expansion-LS}) with error bounds. The proof is 
based on the reformulation of the scattering problem as a pre-conditioned Lippman-Schwinger equation, an approach introduced in~\cite{GW:05}. Appendix~\ref{sec:numerics} contains a brief discussion of the numerical methods used in the simulations. Appendix~\ref{proof-TR} contains the technical proof of operator bounds which are central to the proofs in section~\ref{sec:rigorous-theory}.

\bigskip

\noindent{\bf Acknowledgements:} The authors wish to thank R.V. Kohn and J. Marzuola for fruitful discussions. VD was supported, in part, by Agence Nationale de la Recherche Grant ANR-08-BLAN-0301-01. MIW
was supported in part by NSF grant DMS-07-07850 and DMS-10-08855. MIW
would also like to acknowledge the hospitality of the Courant
Institute of Mathematical Sciences, where he was on sabbatical during
the preparation of this article. VD would like to thank the Department of Applied Physics and Applied Mathematics (APAM) at Columbia University for its hospitality during the Spring of 2008 when this work was initiated.

\section{Main results and Discussion} \label{sec:mainresults}
 We begin with the key hypotheses. Hypotheses~\textbf{(V)} make precise the decomposition of the potential, $V$, into regular, singular and oscillatory parts. Hypothesis \textbf{(G)} specifies, for the cases of {\it generic} and {\it non-generic} potentials, $V_0$, the admissible values of the wave number, $k$. 
 We then state and discuss our main results concerning the transmission coefficients, in the small $\epsilon$ limit. \bigskip
 
{\bf Hypotheses (V)}
\begin{align}
V^\epsilon(x)\ &\equiv\ V_0(x)\ +\ q_\epsilon(x),\ \ {\rm real-valued} \label{Vdecomp}\\
&\equiv V_{sing}(x)\ +\ V_{reg}(x) +\ q_\epsilon(x),\nn\\
q_\epsilon(x)\ &\equiv\ q\left(x,\frac{x}\epsilon\right), \ \ \ V^\epsilon(x)\ge0,
\label{qeps}\end{align}
where
\begin{enumerate}
\item {\it Singular part of $V^\epsilon$, $V_{sing}$}:
\begin{equation}
V_{sing}(x)\ =\ \sum_{j=0}^{N-1}\ c_j\ \delta (x-x_j),\ \ {\rm where}\ \ 
 c_j, x_j \in \RR,\ \ x_j < x_{j+1}.
 \label{Vsing}
 \end{equation}
 \item {\it Regular part of $V^\epsilon$:} \ \ $V_{reg}\in L^{1,2}(\RR)$ with 
\begin{equation}
\|V\|_{L^{1,2}}\ \equiv\ \int_\RR (1+|s|)^{2}\ |V(s)|\ ds\ <\ \infty. \label{Vreg} 
\end{equation}
\item {\it Rapidly varying part of $V^\epsilon$, $q_\epsilon(x)=q\left(x,\frac{x}\epsilon\right)$:}\ \ The mapping $(x,y)\mapsto q(x,y)$ is
\subitem (a)\ $1-$ periodic, {\it i.e.} for each $x\in\RR,\ \ q(x,y+1)=q(y)$,
\subitem (b)\ mean zero with respect to $y$, {\it i.e.} for each $x\in\RR$, 
\begin{equation} \int_0^1 q(x,y)\ dy=0, \label{mean0}\end{equation}
\subitem(c)\ $q\in pC^3_xL^2_{y,per}$, the set of functions $ q:\RR\times S^1\to\RR$, such that there exists a finite partition of $\RR$
\begin{equation} 
-\infty=a_0<a_1<a_2<\dots<a_M<a_{M+1}=+\infty
\nn\end{equation} with
\begin{align}
& \sum_{j=0}^{M+1} \int_0^1 \| q(\cdot,y)\|_{C^3(a_j,a_{j+1})}^2\ dy\ <\ \infty\label{qnorm}
\end{align}
\item We shall work with the Fourier expansion of $q(x,y)$, written as 
\begin{equation}
q(x,y)\ =\ \sum_{j\ne0}\ q_j(x)\ e^{2\pi i jy},\ \ q_j(x)\equiv \int_0^1 e^{-2\pi i jy}q(x,y)\ dy
\label{q-Fourier}
\end{equation}
and assume
\begin{align}
&\int_\RR \int_0^1 |q(x,y)|^2\ dy\ dx=\ \sum_{|j|\ge1}\int_\RR |q_j(x)|^2<\infty,\\
&\int_0^1 |q(x,y)|^2\ dy\ =\ \sum_{|j|\ge1}|q_j(x)|^2\ \to0,\ \ |x|\to\infty\ .
\label{qL2to0}
\end{align}
\item Proposition~\ref{prop:Tqeps-small}, which is a step in proving Theorem~\ref{thm:Texpansion}, requires more decay at infinity for $q_\epsilon$: there exists $\rho>8$ such that
\begin{align}
&(1+|\cdot|^2)^{\rho/2}q_j\in L^2, \ |j|\ge1,\ \mbox{ and } \ \sum_{|j|\ge1} \left\|(1+|\cdot|^2)^{\rho/2}q_j\right\|_{L^2}<\infty,\\
&\frac{d}{dx}\left((1+|x|^2)^{\rho/2}q_j(x)\right)\in L^2 \mbox{ and } \sup_{|j|\ge1} \left\|\frac{d}{dx}\left((1+|x|^2)^{\rho/2}q_j(x)\right)\right\|_{L^2_x}<\infty.
\label{qsharpL2to0}
\end{align}
\end{enumerate}
\bigskip

{\bf Hypothesis (G)}\ If $V_0$ is generic (see Definition~\ref{def:generic}), then the wave number, $k\in K$, an arbitrary compact subset of $\RR$. If $V_0$ is not generic, then the compact set $K$ must be such that  $0\notin K$.\medskip

\begin{remark} If $V_0$ is not generic (as for example $V_0\equiv0$), then the expansions we present in Theorem~\ref{thm:Texpansion} and Theorem~\ref{thm:Texpansion-LS},  are not uniform in a neighborhood of $k=0$. This will be the subject of a future paper.
\end{remark}
\bigskip

The aim of this article is to understand the scattering properties for this class of potentials. In particular, we are interested in the influence of combined microstructure ($q_\epsilon$) and singularities ($V_{sing}$) on the reflection and transmission coefficients and distorted plane waves (see below). Formal application of classical homogenization theory (see for example~\cite{BLP}) suggests that the leading order (in $\epsilon\to0$) scattering behavior is governed by the averaged (homogenized) operator $-\partial_x^2+V_0(x)$; see~\eqref{Vdecomp}. For example, if $V^\epsilon(x)$ is smooth (in particular, $V_{sing}\equiv0$), then the transmission coefficient satisfies the expansion
\begin{equation}
t^\epsilon(k)\ \ \sim\ \ t^{hom}_0(k)+ \epsilon t^{hom}_1(k)\ + \epsilon^2 t^{hom}_2(k)\ +\ \dots
\label{formal-homog}
\end{equation}
where $t_j^{hom}$ are computed from the formal 2-scale homogenization expansion.
In particular, $t_0^{hom}$ is the transmission coefficient associated with the averaged potential $V_0(x)$,
However, homogenization is a theory valid only in the {\it bulk}, away from boundaries or non-smooth points of coefficients. For our class of potentials, this expansion must be corrected.

Our main result is the small $\epsilon$ characterization of the distorted plane waves presented in Theorem~\ref{thm:Expand-LS}. A key consequence of our analysis is the following: \bigskip
\begin{theorem}\label{thm:Texpansion}
Let $V^\epsilon(x)\ =\ V_0(x)\ +\ q_\epsilon(x)$ with $V_0$ and $q_\epsilon(x)=q(x,x/\epsilon)$ satisfying Hypotheses~\textbf{(V)}, and $k\in K$ a compact subset of $\RR$ satisfying Hypothesis~\textbf{(G)}. Denote by $e_{V_0\pm}(x;k)$ the distorted plane waves associated with the unperturbed operator ${-\partial_x^2+V_0(x)}$; see section~\ref{sec:background&statement}.

Then, there exists $\epsilon_0=\epsilon_0(K)$, such that for $0<\epsilon<\epsilon_0$, the transmission coefficient $t^\epsilon=t^\epsilon(k)$ (see~\eqref{eqn:RT}) associated with $V^\epsilon(x)$ satisfies the following expansion uniformly in $k\in K$:
\begin{equation}\label{eqn:Texpansion}
t^\epsilon(k)\ =\ t^{hom}_0(k)\ +\ \epsilon\ t_1^\epsilon(k)\ +\ \epsilon^2\ \left(\ t^{hom}_2(k)\ +\ t_2^\epsilon(k)\ \right) \ +\ t_{rem}^\epsilon(k),
\end{equation}
where $t^{hom}_0(k)$ denotes the transmission coefficient, associated with the average (homogenized) potential $V_0$
 and 
\begin{align}
&t^\epsilon_1(k) = \frac{1}{4k\pi}\displaystyle\sum_{j=1}^{M} e_{V_0+}(a_j;k)e_{V_0-}(a_j;k) \sum_{|l|\geq 1} \left[q_l\right]_{a_j} \frac{e^{2i\pi l\frac{a_j}{\epsilon}}}{l} , \label{t1eps}\\
&t^{hom}_2(k) = \displaystyle\frac{i}{8k\pi^2}
 \sum_{|j|\geq 1} j^{-2} \int_\RR\ |q_j(z)|^2\ e_{V_0-}(z;k)e_{V_0+}(z;k) dz, \label{t2hom}\\
 &t^\epsilon_2(k) =\ \frac{i}{8k\pi^2} \displaystyle\sum_{j=1}^{M} \sum_{|l|\geq 1} \left[\ \partial_x \left(e_{V_0+}(x;k)\ e_{V_0-}(x;k)q_l(x) \right)\ \right]_{a_j} \frac{e^{2i\pi l\frac{a_j}{\epsilon}}}{l^2},\label{t2eps}\\
 &t_{rem}^\epsilon(k) =\ o(\epsilon^{2+}) \mbox{, more precisely quantified in Proposition~\ref{prop:convergence-asymptotic-expansion}}.
\end{align}

\begin{itemize}
\item[(a)]$t^{hom}_j,\ \ j=0,2\dots$, denote the expansion coefficients for the transmission coefficient obtained from the two-scale (bulk) homogenization expansion, valid for smooth potentials. 
\item[(b)] $t_1^\epsilon$ arises due to discontinuities in $x\mapsto q(x,\cdot)$, and 
\item[(c)] $t_2^\epsilon$ arises due to both the singular part of the potential, $V_{sing}$, and discontinuities in $x\mapsto q(x,\cdot)$ or $x\mapsto \partial_x q(x,\cdot)$.
\end{itemize}
 $t_1^\epsilon$ and $t_2^\epsilon$ are uniformly bounded, for $\epsilon$ small. However each is a sum over rapidly oscillating (as $\epsilon\to0$) terms of the form $\exp({i\frac{\nu}{\epsilon}})$, corresponding to discontinuity points of $q_\epsilon$, respectively points in the support of $V_{sing}$. 
\end{theorem}
\bigskip

Theorem~\ref{thm:Texpansion} is a consequence of the more general Theorem~\ref{thm:Texpansion-LS}, stated below, which follows from the asymptotic study of the convergent expansion of the distorted plane waves, presented in Theorem~\ref{thm:Expand-LS}. The proof of Theorem~\ref{thm:Expand-LS} is based on construction and asymptotic study of the scattering problem via a {\it pre-conditioned} Lippman-Schwinger equation.
 This approach is quite general and applies to the perturbation theory of Schr\"odinger operators of the form 
 \begin{equation}
 H=-\partial_x^2 + V_0(x)\ +\ Q(x),
 \nn\end{equation}
 where $Q$ is small in the sense that $|||Q|||\sim\left\|(I-\Delta)^{-\frac12}Q (I-\Delta)^{-\frac12}\right\|_{L^2\to L^2}$ is small. This formulation was introduced in~\cite{GW:05} to study the perturbation of scattering resonances due to high contrast microstructure perturbations of a potential. 
 If $Q$ is a ``microstructure'', roughly meaning that it is supported at high frequencies, then $|||Q|||$ is small. 
 Here, we apply this method and obtain a convergent expansion of $Q\mapsto e_{V_0+Q}(x,k)$ for fixed $k$ and $|||Q|||$ sufficiently small. The expansion of the transmission coefficient, $Q\mapsto t_{V_0+Q}(k)$, is a direct consequence of:\medskip
 \begin{theorem}\label{thm:Texpansion-LS}
Let $V(x)\ =\ V_0(x)\ +\ Q(x)$ with $V_0$ satisfying Hypotheses \textbf{(V)} and $(1+|x|^2)^{\rho/2}Q\in L^2$, for $\rho>8$. We use the following norm on $Q$, see section~\ref{sec:Formulation-LS-equation}:
 \[
 |||Q|||\ \equiv\ \left\| \D^{-1}\ (1+|x|^2)^{\rho/4}\ Q \ (1+|x|^2)^{\rho/4}\ \D^{-1} \right\|_{L^2 \to L^2}.
 \]
Set $k\in K$ a compact subset of $\RR$ satisfying Hypothesis~\textbf{(G)}, and denote by $e_{V_0\pm}(x;k)$ the distorted plane waves associated with the unperturbed operator $-\partial_x^2+V_0(x)$; see section~\ref{sec:background&statement}. Denote by $t=t(k,Q)=t(k)$ the transmission coefficient (see~\eqref{eqn:RT}) associated with $V(x)$. There exists $\tau_0=\tau_0(K)$ such that for $0<|||Q|||<\tau_0(K)$, we have the following expansion
 which holds uniformly in $k\in K$: 
\begin{equation}\label{eqn:Texpansion-LS}
t(k,Q)\ =\ t^{hom}_0(k)\ +\ t_1[Q]\ +\ t_2[Q,Q] \ +\ t_{rem}(k),
\end{equation}
with $t^{hom}_0(k)$ the transmission coefficient, associated with the average (homogenized) potential $V_0$, and the following:
\begin{align}
&t_1[Q] \ = \ \frac{1}{2ik}\int_{-\infty}^{\infty} Q(\zeta)\ e_{V_0+}(\zeta;k)\ e_{V_0-}(\zeta;k) \ d\zeta,\\
&t_2[Q,Q] \ = \ \frac{1}{2ik}\int_{-\infty}^{\infty}\ Q\ R_{V_0}(k)(Q(\zeta)\ e_{V_0+}(\zeta;k))\ e_{V_0-}(\zeta;k) \ d\zeta,\\
&t_{rem}(k) \ = \mathcal{O}\left(|||Q|||^{2+}\right)\ \mbox{ and more precisely estimated in Theorem~\ref{prop:convergence-asymptotic-expansion-LS}}.
\end{align}
Here, $R_{V_0}(k)$, $e_{V_0+}(x;k)$ and $e_{V_0-}(x;k)$ being defined in section~\ref{sec:background&statement}.
\end{theorem}
\medskip

\begin{remark}\label{rem:cancellations}\ {\bf Symmetry considerations:}\ 
There is a class of potentials, $q_\epsilon$, whose members are discontinuous, and yet the (oscillatory in $\epsilon$) correctors, $t_j^\epsilon,\ j\ge1$ vanish. In subsection~\ref{sec:specificV} we explore families of such structures. 
Indeed, let us apply Theorem~\ref{thm:Texpansion-LS} with $V\equiv V_\epsilon$ satisfies Hypotheses~{\bf (V)} as well as the additional properties:
$V_0$ even and $q_\epsilon$ ``separable'':
 \[\ V_0(x)=V_0(-x),\ \ \ \ q_\epsilon(x) \ = \ q_0(x)\ q_{per}\left(\frac{x}{\epsilon}\right).\]
 One can easily see that $V_0$ is even implies that $e_{V_0+}(\cdot;k) e_{V_0-}(\cdot;k)$ is even. Therefore, if $q_0$ and $q_{per}$ are of opposite parity, then $x\mapsto e_{V_0+}(x;k)\ e_{V_0-}(x;k)\ q_\epsilon(x)$ is odd
 and therefore $t_1[q_\epsilon](k)\equiv0$ for any $\epsilon>0$. It follows that for such potentials, and \emph{even if $q_\eps$ is discontinuous}, the leading order correction to $t^{hom}_0(k)$ is $t_2[Q,Q]$ which is of order $\mathcal{O}(\epsilon^2)$ (see section~\ref{sec:Texpansion}). Moreover, in this special case the second order corrector is well defined:
 \[\lim_{\epsilon\downarrow0} \ \epsilon^{-2} \left(\ t^\epsilon(k) - t^{hom}_0(k)\ \right) = t^{hom}_2(k).\] 
\end{remark}
 \bigskip
 
 The three subplots of figure~\ref{figT-T0} illustrate the results of Theorem~\ref{thm:Texpansion} on the behavior of $t^\epsilon-t_0$ 
 for several contrasting choices of potential $V^\epsilon=V_0+q_\eps$, where $V_0$ is a finite sum of Dirac delta functions, at equally spaced points.\footnote{The precise functions and parameters used to obtain the plots displayed in figures~\ref{figT-T0} and~\ref{figSmoothDirac} are given in Appendix~\ref{sec:numerics}, page~\pageref{sec:numerics}.}
 \begin{itemize}
\item The left panel of figure~\ref{figT-T0} corresponds to the case where $q_\epsilon$ is discontinuous. It shows that
\begin{align}
&t^\epsilon-t^{hom}_0={\mathcal O}(\epsilon),\ \ \epsilon\to0,\ \ {\rm and\ yet }\ \ 
\epsilon^{-1}\left(\ t^\epsilon-t^{hom}_0\ \right) \ \ {\rm does \ not\ have\ a \ limit}.
\nn\end{align}
\item The center panel of figure~\ref{figT-T0} corresponds to the case where $q_\epsilon$ is a smooth function, and $V_0$ is a Dirac delta function. Here, 
\begin{align}
&t^\epsilon-t^{hom}_0={\mathcal O}(\epsilon^2),\ \ \epsilon\to0,\ \ {\rm and\ yet }\ \ 
\epsilon^{-2}\left(\ t^\epsilon-t^{hom}_0\ \right) \ \ {\rm does \ not\ have\ a \ limit}.
\nn\end{align}
\item The right panel of figure~\ref{figT-T0} corresponds to the case where $q_\epsilon$ is a smooth function, and $V_0$ is a {\it smoothed out} Dirac delta function. Here we find
\begin{align}
&t^\epsilon-t^{hom}_0={\mathcal O}(\epsilon^2),\ \ \epsilon\to0,\ \ {\rm and}\ \ 
\lim_{\epsilon\downarrow0} \epsilon^{-2}\left(\ t^\epsilon-t^{hom}_0\ \right)
 \ \ {\rm is\ well-defined}.
\nn\end{align}
\end{itemize}
This phenomenon of indeterminacy of higher order correctors, due to boundary layer effects is discussed, in the context of a Dirichlet spectral problem~\cite{Santosa-Vogelius,Moskow-Vogelius-1}.

\begin{figure}[htp]
\includegraphics[width=0.32\textwidth]{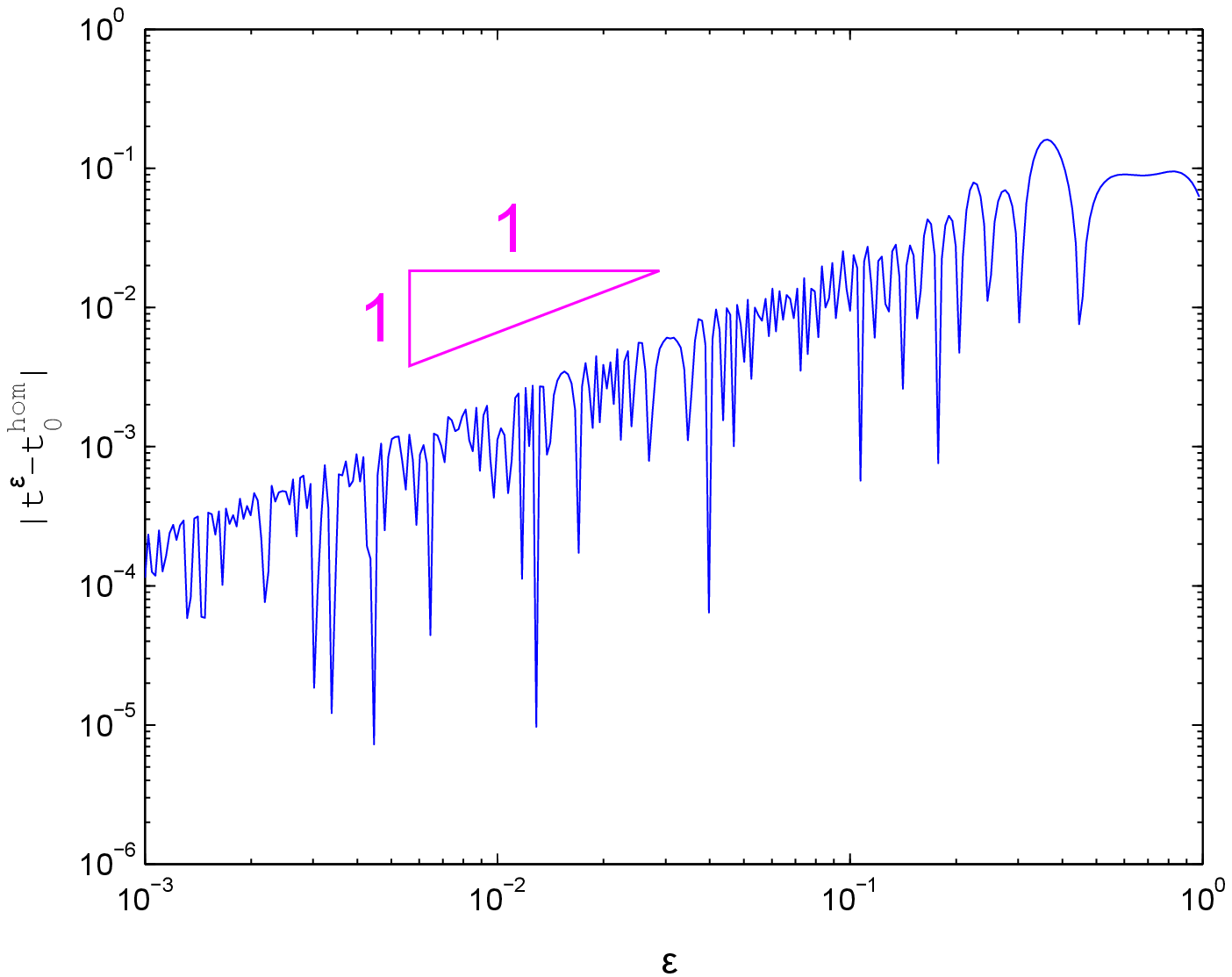} 
\hfill
\includegraphics[width=0.32\textwidth]{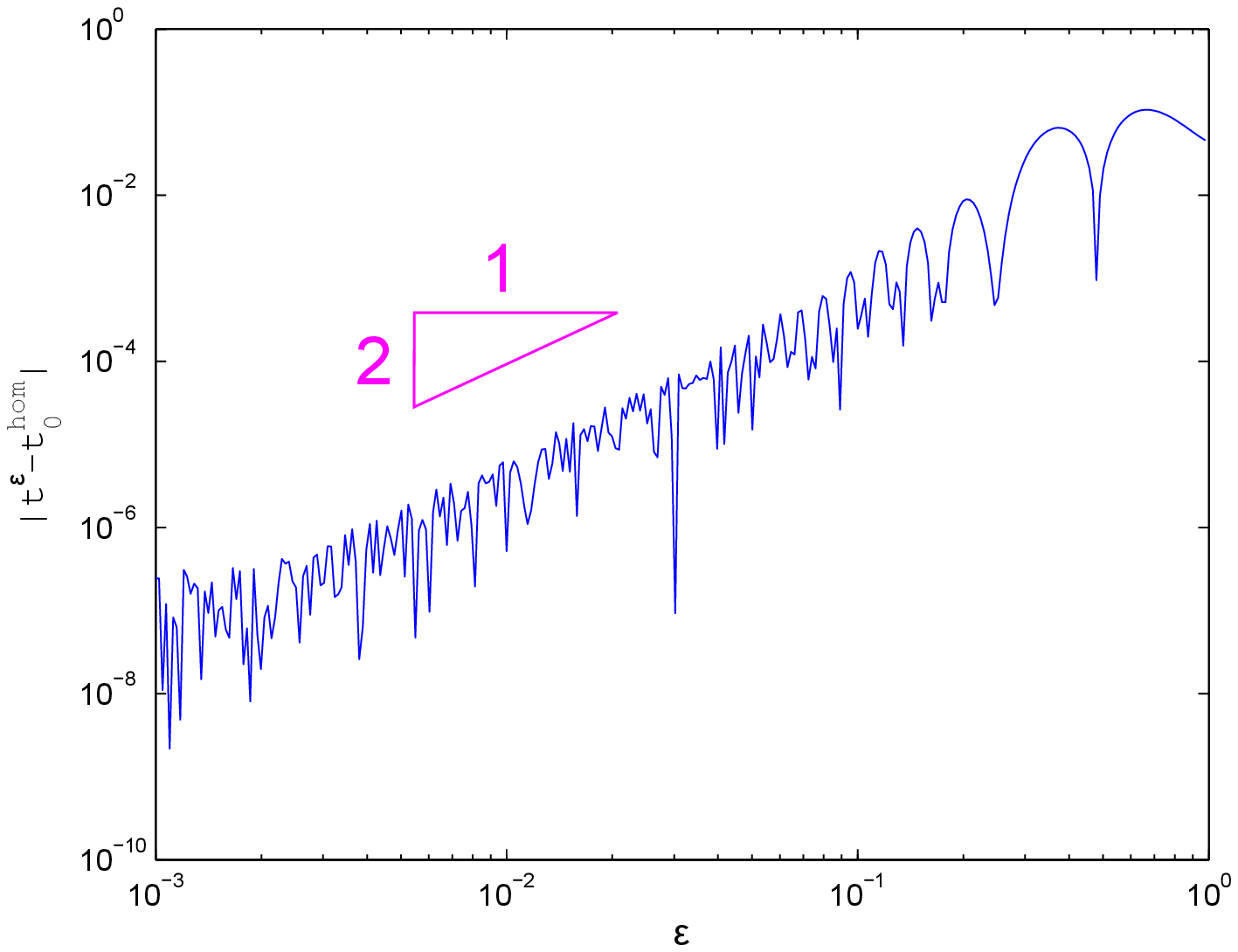}
\hfill
\includegraphics[width=0.32\textwidth]{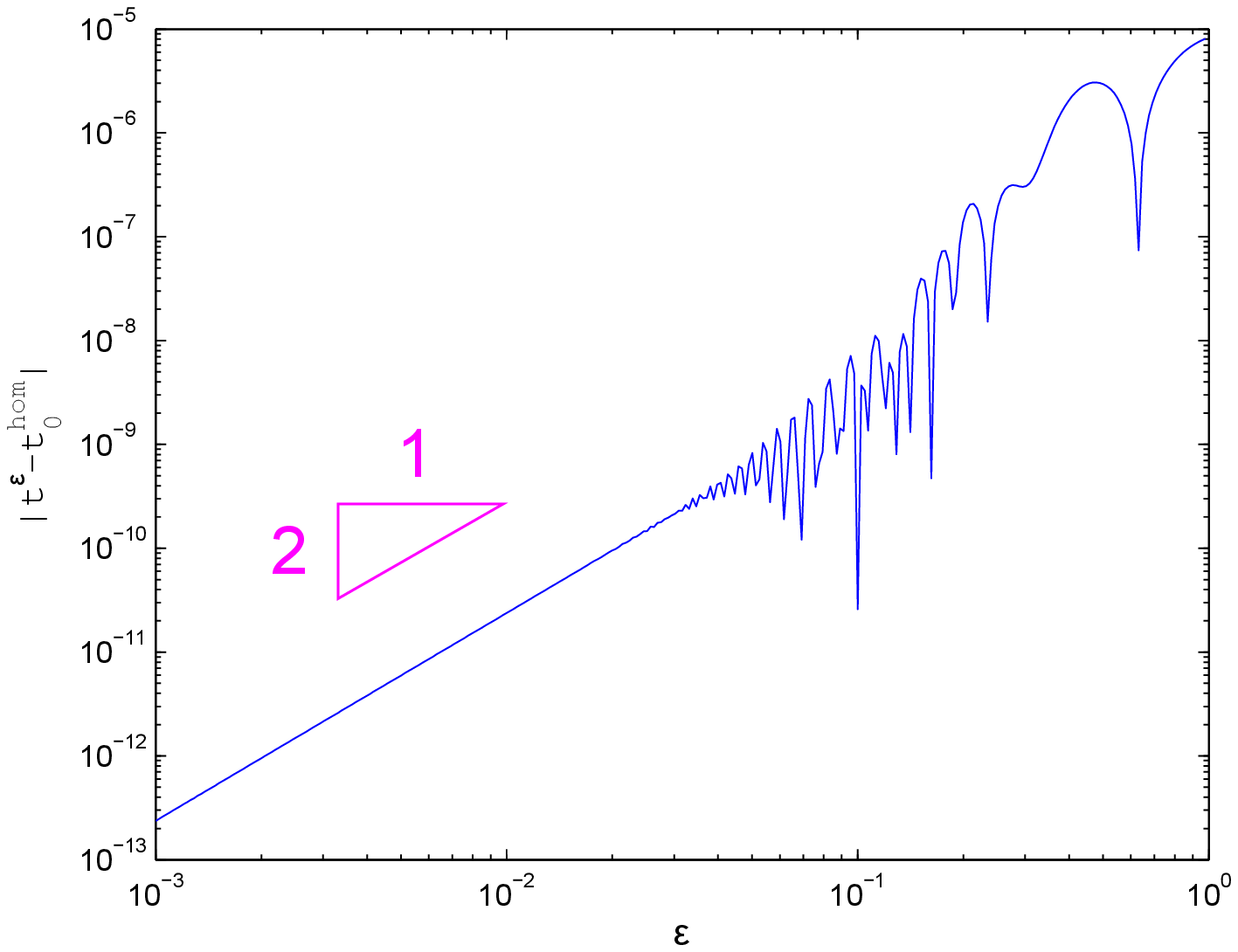}
\caption[Illustration of Theorem~\ref{thm:Texpansion}]{
 Illustration of Theorem~\ref{thm:Texpansion} via plot of $\log\left|t^\epsilon-t_0^{hom}\right|$ versus $\log\epsilon^{-1}$ for the case of $q$ discontinuous and $V_0$ a sum of Dirac delta functions (left panel, average slope $1$), $q$ smooth and $V_0$ a sum of Dirac delta functions (center panel, average slope $2$ ). The right panel (slope $2$) is for the case where $V^\eps=V_0$ is a {\it smooth approximation} of a finite sum of Dirac delta-functions. }\label{figT-T0}
\end{figure}
The transition between the cases of a regular potential and a potential containing singularities is illustrated in Figure~\ref{figSmoothDirac}. The three panels show the behavior of $t^\epsilon-t_0$ with respect to $\epsilon$, where the potential $V^\epsilon\ = \ V_0+q_\epsilon$ satisfies $q_\epsilon$ is smooth and $V_0$ is a sum of smoothed out Dirac delta functions. From right to left, $V_0$ is an improving approximation of Dirac delta functions. 
\begin{figure}[htp]
\includegraphics[width=0.32\textwidth]{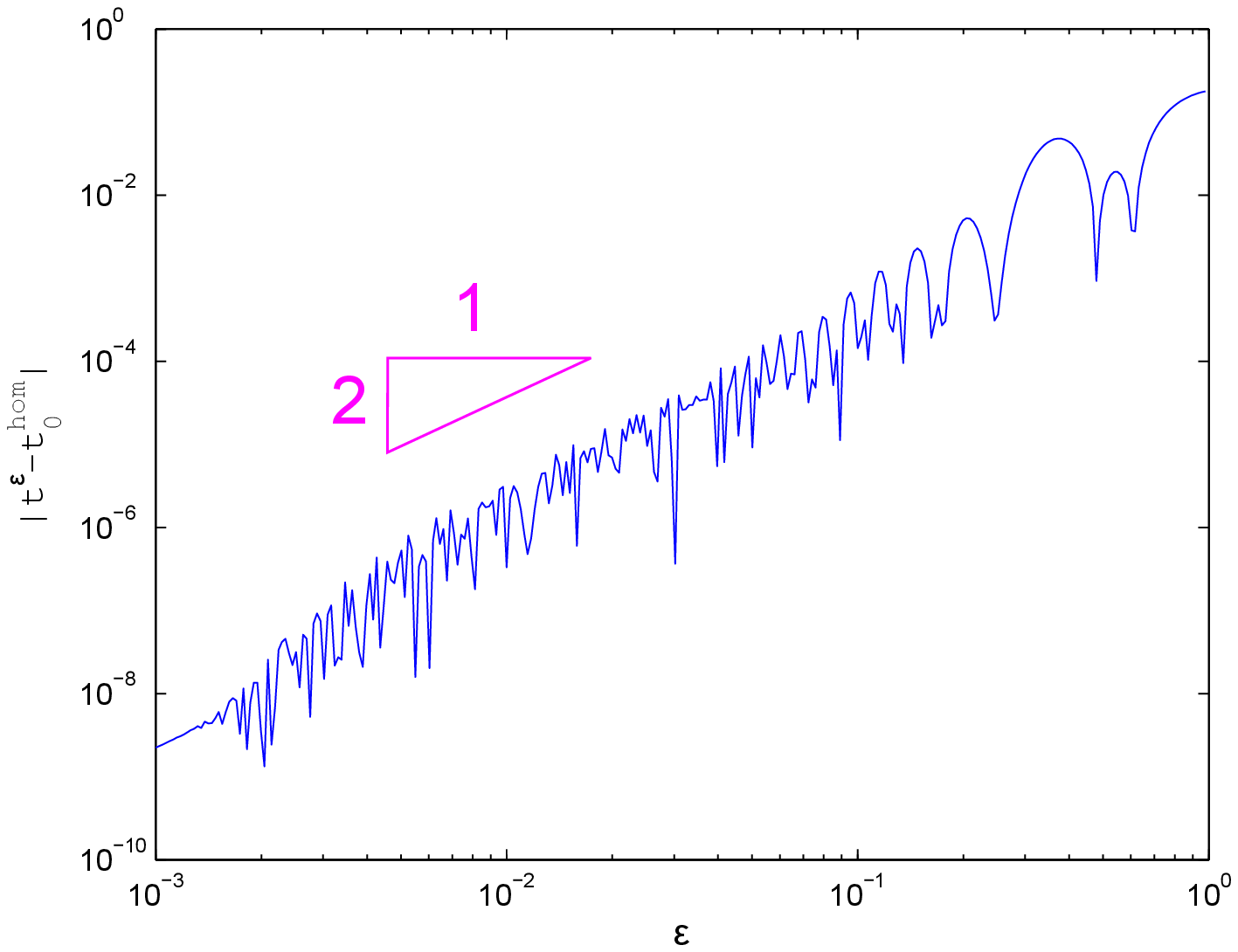}
\hfill
\includegraphics[width=0.32\textwidth]{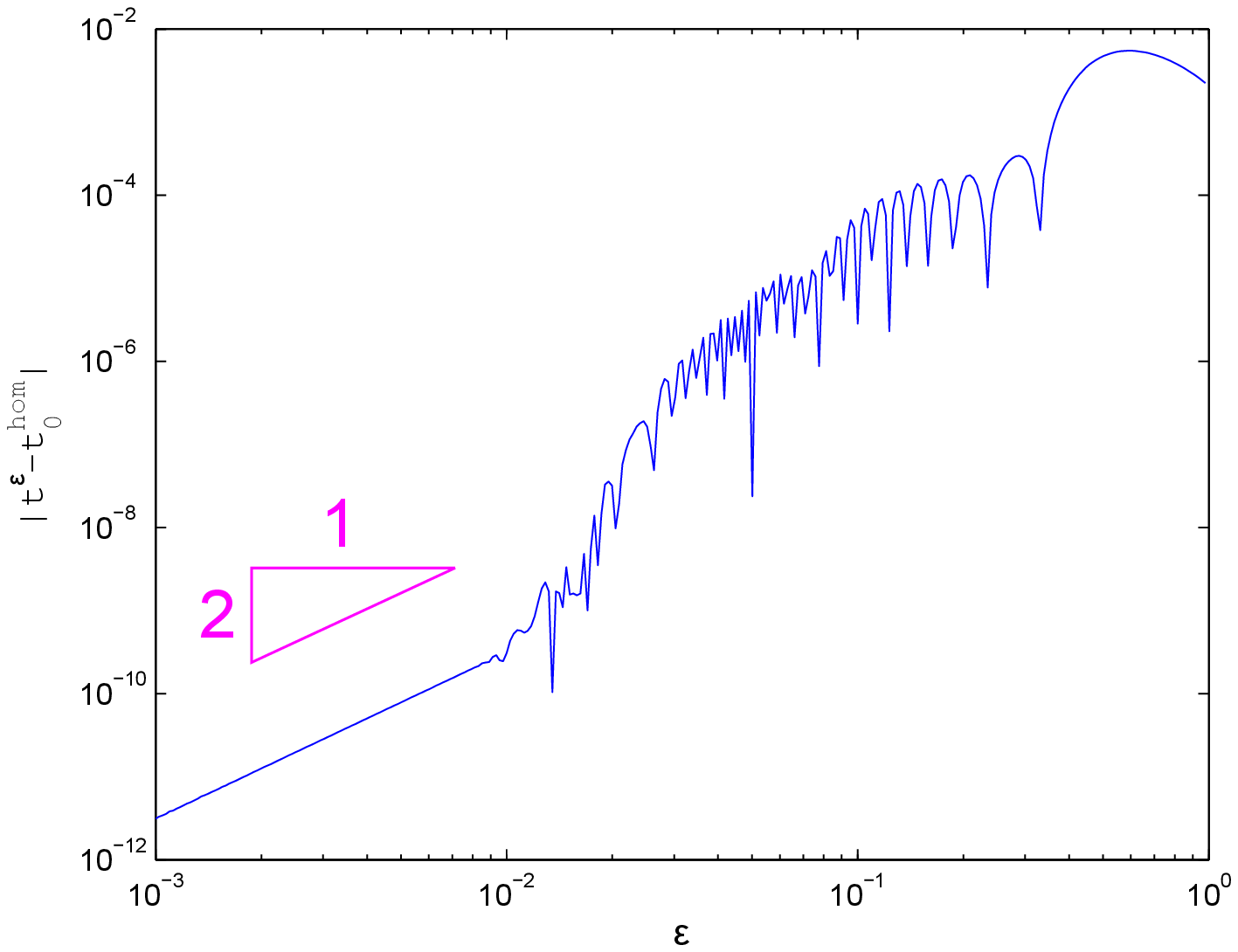}
\hfill
\includegraphics[width=0.32\textwidth]{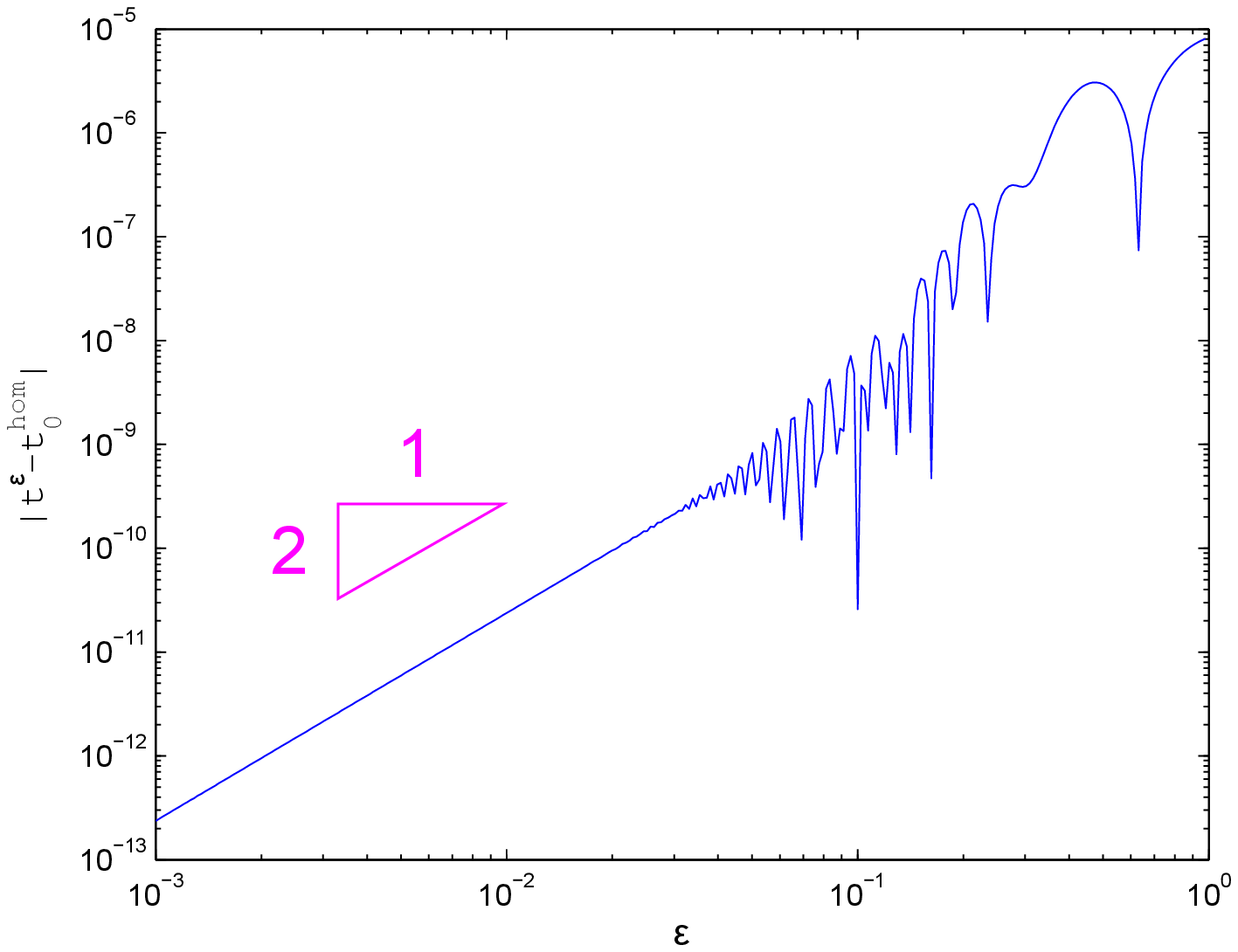} 
\caption[Approximated Dirac delta function]{
 Plot of $\log\left|t^\epsilon-t_0^{hom}\right|$ versus $\log\epsilon^{-1}$ for the case of $q$ smooth and $V_0$ a sum of three approximate Dirac delta functions $  \delta_\rho(x)\equiv \frac{1}{\rho\sqrt{\pi}} \mathrm{e}^{-x^2/\rho^2}$, with $ \rho=0.001, 0.01, 0.1 $. }\label{figSmoothDirac}
\end{figure}

\subsection{Some specific structures}\label{sec:specificV}

We now study in detail two natural and illustrative classes of potentials:
\begin{enumerate}
 \item We first consider a one-parameter family of structures, which are truncations of a smooth potential, where for certain parameter ranges the manner of truncation causes a discontinuity. The latter corresponds to cleaving a periodic structure in a manner not commensurate with the background medium:
 \begin{equation}\label{Vex1}
 V^\epsilon_1(x;\theta) \ = \ \cos\left(\frac{2\pi x}\epsilon + \theta\right) \mathbf{1}_{[-1,1]}(x) .
 \end{equation}
We are obviously in the case related in Remark~\ref{rem:cancellations}, with $V_0\equiv 0$ (so that $e_{V\pm}(x;k)\ = \ e^{\pm ikx}$ and $t_0^{hom}\ =\ 1$). More precisely, it is easy to show that
\begin{align*}
 t^\epsilon_1(k;\theta) &\equiv \frac{1}{4k\pi}\displaystyle\sum_{j=1}^{M} e_{V_0+}(a_j;k)e_{V_0-}(a_j;k) \sum_{|l|\geq 1} \left[q_l\right]_{a_j} \frac{e^{2i\pi l\frac{a_j}{\epsilon}}}{l} \\
 &= \frac{-i}{2k\pi}\ \cos(\theta) \ \sin\left(\frac{2\pi }\epsilon\right).
\end{align*}
In general, $t_1^\epsilon(k)\ne0$ but for $\theta = \frac\pi2 + m\pi,\ m\in\ZZ$, $q_\epsilon(\cdot;\theta)$ is even and therefore
 for all $k\in\RR$ and $\epsilon>0$, we have $t_1[q_\epsilon](k) \ = \ 0.$



\item Our second example is a piecewise constant (discontinuous) structure which is smoothly truncated
 \begin{equation}\label{Vex2}
 V^\epsilon_2(x;\theta) \ \equiv \ h_{per}\left(\frac{x}\epsilon + \theta\right) \ e^{-\frac{x^2}{(x-1)(x+1)}}\ \mathbf{1}_{[-1,1]}(x),
 \end{equation}
 with $h_{per}(y)$ the 1-periodic function such that $h(y) = -1$ for $y\in (-1/2,1/2]$, and $h(y) = 1$ for $y\in (1/2,3/2]$.
 
 Since the slow-varying part of $q_\epsilon(x)$ is smooth, and $V_0$ has no singularity, Theorem~\ref{thm:Texpansion} predicts that
\[t^\epsilon-t^{hom}_0={\mathcal O}(\epsilon^2),\ \ \epsilon\to0,\ \ {\rm and}\ \ 
\lim_{\epsilon\downarrow0} \epsilon^{-2}\left(\ t^\epsilon-t^{hom}_0\ \right) = t^{hom}_2
 \ {\rm is\ well-defined},\]
 even though the function $q_\epsilon(x)$ has internal discontinuities.
 In Figure~\ref{fig:V1V2-eps}, we plot $\log\left|t^\epsilon-t_0^{hom}\right|$ versus $\log\epsilon^{-1}$ for the two potentials $V^\epsilon_1$ and $V^\epsilon_2$, setting $k=1$, and $\theta=0$. 
 
 \begin{figure}[htp]
\hfill
\includegraphics[width=0.45\textwidth]{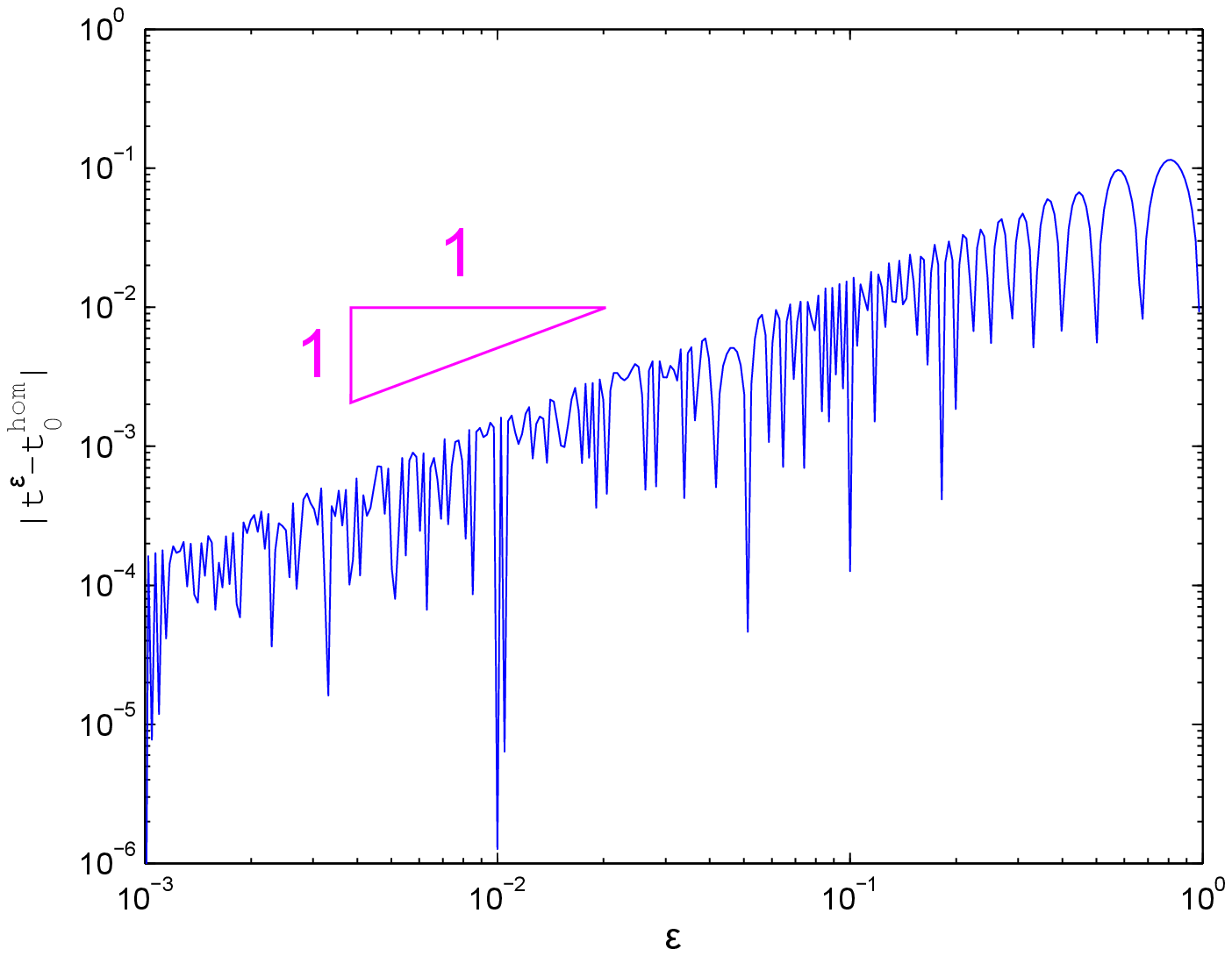}
\hfill
\includegraphics[width=0.45\textwidth]{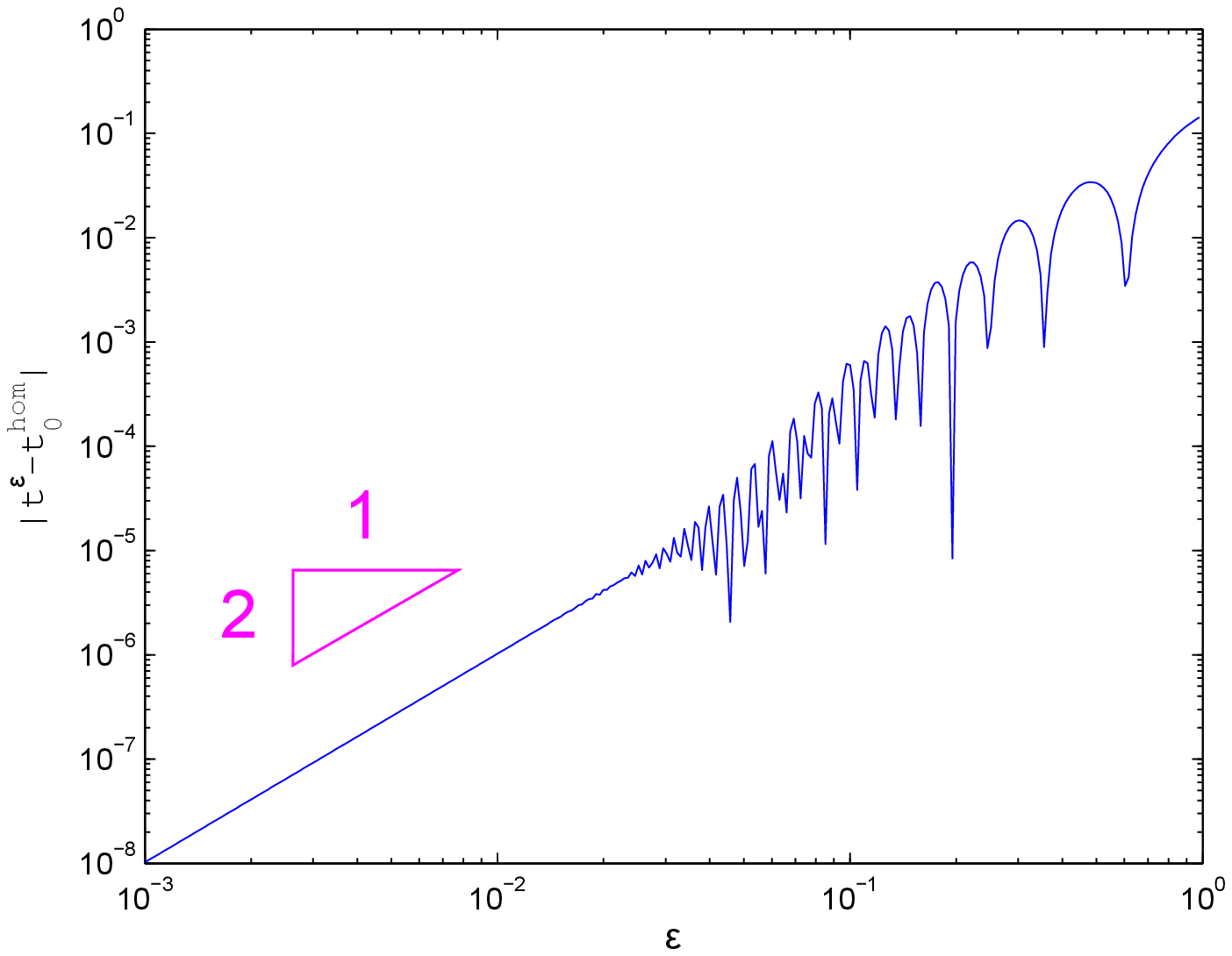}
\hfill
\caption{
 Plot of $\log\left|t^\epsilon-t_0^{hom}\right|$ versus $\log\epsilon^{-1}$ for the case of the potentials $V^\epsilon_1(x;\theta)$ in~\eqref{Vex1} (left panel, slope 1), and $V^\epsilon_2(x;\theta)$ in~\eqref{Vex2} (right panel, slope 2). One has $k=1$, and $\theta=0$.}\label{fig:V1V2-eps}
\end{figure}
 
\end{enumerate}

\section{Background on one-dimensional scattering theory}\label{sec:background&statement}

For simplicity, we consider potentials, $W$, which have no localized eigenstates,
 {\it i.e.} the spectrum of $-\partial_x^2+W(x)$ is continuous.
 We further assume that $W$ has the form
 \begin{align*}
 W&=W_{reg}+W_{sing}, \ \ {\rm with}\\
 & \ W_{reg}\ \in\ L^{1,3/2+}(\RR), \\
 & \ W_{sing}\ =\ \sum_{j=0}^{N-1} c_j \ \delta (x-x_j),\ \ {\rm where}\ \ 
 c_j, x_j \in \RR,\ \ x_j < x_{j+1}.
\end{align*}
 We now introduce an appropriate notion of solution to the Schr\"odinger equation 
 \begin{equation}
 \left(\ H_W\ -\ k^2\ \right)u\ \equiv\ \left(-\frac{d^2}{d x^2}+W(x)-k^2\right) u =0.
 \label{schrod-eqn}
 \end{equation}
Let $[U]_\xi$ denote the jump in $U$ at the point $\xi$, {\it i.e.}
\begin{equation}
[U]_\xi = \lim_{x\to\xi^+}U(x)\ -\ \lim_{x\to\xi^-}U(x).
\label{jumpdef}
\end{equation}
 \begin{definition}\label{def:solution}
 We say that $u$ is a solution of time-independent Schr\"odinger equation~\eqref{schrod-eqn}
 if $u$ is piecewise $C^2$, and satisfies the~\eqref{schrod-eqn} on $\RR\setminus \supp W_{sing}=\RR\setminus\{x_0,\dots,x_{N-1}\}$ as well as the jump conditions
 \begin{equation} \left\{\begin{array}{l}
\left[u\right]_x= 0,\ \ \ x\in\RR, \\ \\
\left[\frac{d}{dx}u\right]_x = 0 \mbox{ if } x\in\RR\setminus \supp W_{sing}, \\ \\
\left[\frac{d}{dx}u\right]_{x_j}= c_j u(x_j) \mbox{ where } x_j\in \supp W_{sing}. 
\end{array}\right. \label{JumpCondition}\end{equation}
\end{definition}

Of special interest are the Jost solutions, defined below.
\begin{definition}\label{def:JostSolutions}
 The {\it Jost solutions} $f_\pm(x;k)\equiv m_\pm(x;k)e^{\pm ikx}$ are the unique solutions of~\eqref{schrod-eqn}, such that
 \[ \lim_{x\to \pm\infty} m_\pm(x;k) \ = \ 1.\]
\end{definition}
This definition is valid, as we see in Appendix~\ref{sec:Jost}. We shall use some smoothness and decay properties of these solutions, that are also postponed to Appendix~\ref{sec:Jost}, for the sake of readability.

\medskip

With the help of the Jost solutions, we are able to define scattering quantities, as the {\em transmission and reflection coefficients}, and the {\em distorted plane waves}.

Since $f_\pm(x;k)$ and $f_\pm(x;-k)$ are solutions of~\eqref{schrod-eqn}, and are independent for $k\neq0$, there exists unique functions $t_{\pm}(k)$ and $r_{\pm}(k)$, such that
\begin{align*}
f_- (x,k) & \ =\ \frac{r_+(k)}{t_+(k)}\ f_+ (x,k) + \frac{1}{t_+(k)}\ f_+ (x,-k), \\
f_+ (x,k) & \ =\ \frac{r_-(k)}{t_-(k)}\ f_- (x,k) + \frac{1}{t_-(k)}\ f_- (x,-k).
\end{align*}
It is then easy to check that $t_+(k)=t_-(k)\equiv t(k)$, and that $t(k)$ and $r_{\pm}(k)$ are continous at $k=0$.
The distorted plane waves $e_{W\pm}(x;k)$ are then defined by:
\begin{definition}\label{def:eTR}
 Given a potential $W(x)$, we define $e_{W\pm} (x;k)$, the {\it distorted plane waves
 associated with $H_W$} by 
 \begin{align}
e_+(x;k) &\ \equiv\ t(k)f_+(x;k) \ \equiv \ t(k)m_+(x;k)e^{ikx} \\
e_-(x;k) &\ \equiv\ t(k)f_-(x;k) \ \equiv \ t(k)m_-(x;k)e^{-ikx}.
\end{align}
\end{definition}

The distorted plane waves $e_{W\pm}(x;k)$ play the role for $H_W$ that the plane waves $e^{\pm ikx}$ play for $H_0=-\partial_x^2$, as we see below.
Let us first introduce the notion of {\it outgoing radiation} as $|x|\to\infty$.
\begin{definition}\label{def:outgoing}
$U(x)$ is said to satisfy an {\it outgoing radiation condition} or to be {\it outgoing} as $|x|\to\infty$ if
\begin{equation}
\left(\ \partial_x\ \mp\ ik\ \right)U\ \to\ 0,\ \ {\rm as}\ x\to\pm\infty.
\nn\end{equation}
\end{definition}
\begin{proposition}\label{prop:outgoing}
Given a potential $W(x)$, $e_{W\pm} (x;k)$, the distorted plane waves $e_{W\pm}(x;k)$ are the unique solutions of~\eqref{schrod-eqn} satisfying
 \begin{equation}
e_{W\pm}(x;k)\ =\ e^{\pm i k x}\ +\ {\rm outgoing}(x).
\label{outgoing}
\end{equation}
\end{proposition}
More precisely, they satisfy the following asymptotic relations~\cite{DT:79}: 
 \begin{equation}
\label{eqn:RT}\left\{\begin{array}{lr}
e_{W+}(x;k)\ -\ \left(e^{ikx}+ r_{+}(k)e^{-ikx}\right)\ \longrightarrow\ 0 &\mbox{ as } x\to -\infty, \\
e_{W+}(x;k)\ -\ t(k)e^{ikx}\ \longrightarrow\ 0 &\mbox{ as } x\to +\infty, \\ 
e_{W-}(x;k)\ -\ t(k)e^{-ikx}\ \longrightarrow\ 0 &\mbox{ as } x\to -\infty, \\
e_{W-}(x;k)\ -\ \left(e^{-ikx}+ r_{-}(k)e^{ikx}\right)\ \longrightarrow\ 0 &\mbox{ as } x\to +\infty. \end{array}\right. 
\end{equation}

A consequence of the relations~\eqref{eqn:RT} is the Wronskian identity:
\begin{align}
Wron\left(e_{W+}(\cdot;k),e_{W-}(\cdot;k)\right)\ &\equiv\ e_{W+}\partial_xe_{W-}-\partial_xe_{W+}e_{W-}
&=\ -2ik\ t(k).
\label{wronskian}
\end{align}
In terms of the Jost solutions:
\begin{align}
Wron\left(f_{+}(\cdot;k),f_{-}(\cdot;k)\right)\ &=\ -\frac{2ik}{t(k)},\ \ k\ne0.
\label{wronskian-jost}
\end{align}

By analyticity in $W$, potentials for which $\left.Wron\left(f_{+}(\cdot;k),f_{-}(\cdot;k)\right)\right|_{k=0}=0$ are isolated in the space of potentials.

\begin{definition}\label{def:generic}
A potential $W$ is said to be \underline{generic} if 
\[ Wron\left(f_+(x;0),f_-(x;0)\right) \ = \ Wron\left(m_+(x;0),m_-(x;0)\right) \ \neq \ 0.\]
Otherwise, the operator $H_W$ is said to have a zero-energy resonance, {\em i.e.} $H_W u=0$ has a non-trivial solution that is bounded both as $x\to\infty$ and as $x\to-\infty$.
\end{definition}

Note that the potential $W(x)\equiv0$ is not generic since $m_+(x;k)\equiv m_-(x;k)\equiv1$. If $W$ is generic, we have~\cite{DT:79,Newton:86,Weder:99}
 \begin{equation}
t(k)\ =\ -\frac{2ik}{Wron\left(f_+(\cdot;0),f_-(x;0)\right)}\ +\ o(k)\ =\ \mathcal{O}(k),\ \ |k|\to0.
\label{t-as-k20}
\end{equation}
In particular, $t(0)=0$ and $r_\pm(0)=-1$.

\medskip

A simple calculation then yields the following expressions for the outgoing Green's function (resolvent kernel) and the outgoing resolvent, $R_W(k),\ k\ne0$:
\begin{equation}
R_W(x,y;k)\ =\ \left\{ \begin{array}{lr}
                 \frac{1}{-2ik\ t(k)}\  e_{W-}(y;k)\ e_{W+}(x;k) , & y<x,\\ \\
                 \frac{1}{-2ik\ t(k)}\ e_{W-}(x;k)\ e_{W+}(y;k) , & y>x,
                  \end{array}
             \right.
 \label{outgoing-resolvent}
 \end{equation}
\begin{equation}\label{eqn:RV0}
R_{W}(k)F(x)\ = \left(-\frac{d^2}{d x^2}+W(x)-k^2\right)^{-1}\ F(x)\ =\  \int_{-\infty}^\infty R_W(x,\zeta;k) F(\zeta) d\zeta.
\end{equation}

\begin{remark}\label{rem:k=0}
 Note that these expressions, originally defined for $k\neq0$, are easily extended to the point $k=0$, for generic potentials. Indeed, one has by Definition~\ref{def:generic}:
 \[\frac{1}{-2ik\ t(k)}\  e_{W-}(y;k)\ e_{W+}(x;k) \ = \ \frac{f_{-}(y;k)\ f_{+}(x;k)}{Wron\left(f_+(\cdot;k),f_-(\cdot;k)\right)}.\]
In the generic case, this expression has a limit when $k\to0$ by~\eqref{wronskian-jost}
 and~\eqref{t-as-k20}. In the following, we work with the distorted plane waves, which sometimes lead to expressions which are only defined for $k\neq0$. By the above considerations, it is easy to check that in the case of a generic potential, these expressions have a well-defined finite limit when $k\to0$.
\end{remark}

\medskip

In particular, we have the following\medskip
\begin{proposition}\label{inhomog-solve}
Let $F\in L^1(\RR)$. Assume $W(x)$ satisfying Hypotheses~\textbf{(V)} and $k\in K$ satisfying Hypothesis~\textbf{(G)}. Then the inhomogeneous equation
\begin{equation}
\left(-\frac{d^2}{d x^2}+W(x)-k^2\right)\ U\ =\ F\label{inhom}
\end{equation}
has the unique outgoing solution $U=R_W(k)F$.
Moreover, $\|U\|_{L^\infty}\le C\ \|F\|_{L^1}$, with a constant, $C(K)$.
\end{proposition}
\begin{proof} Existence follows from the explicit integral representation~\eqref{eqn:RV0}. Note that if $W$ is generic, then $R_W(k)F$ is defined
 for any $k\in\RR$, whereas in the non-generic case, $Wron\left(f_+(x;k),f_-(x;k)\right)\longrightarrow 0\ (k\to0)$ and $f_{\pm}(x,k)$ does not tend to zero as $k\to0$~\cite{DT:79}, so that $R_W(k)F$ has a simple pole at $k=0$.
 \medskip
 
 To prove uniqueness, note that if the difference, $d(x)$, of two solutions is non-zero, then $d(x)$ is a non-trivial solution of the scattering resonance problem, that is ${(H_W-k^2) d=0}$, $d(x)$ outgoing at $|x|\to\infty$ with scattering resonance energy $k^2\in\RR$. However, the scattering resonance energies must satisfy $\Im (k^2)<0$; see, for example,~\cite{Tang-Zworski}. Therefore, $d(x)\equiv0$. This completes the proof.
\end{proof}


\section{Homogenization / Multiple Scale Perturbation Expansion}
\label{sechomo}
\subsection{Multiple scale expansion}
In this section, our goal is to formally obtain the expansion displayed in Theorem~\ref{thm:Texpansion}, using a systematic two-scale / homogenization perturbation scheme. A proof (and derivation by other means) of this expansion is presented in section 
\ref{sec:rigorous-theory}.

We seek a solution of the equation  
\begin{equation}
\left(-\frac{d^2}{d x^2}+V_0(x)+q\left(x,\frac{x}{\epsilon}\right)-k^2\right)e_{V^\epsilon+}(x;k)=0,
\label{eV+eqn}
\end{equation}
 in the form of a two-scale function, $e_{V^\epsilon+}(x;k)=U^\epsilon (x,\frac{x}{\epsilon})$ which satisfies the jump conditions~\eqref{JumpCondition} and the outgoing radiation condition of Definition~\ref{outgoing}. Treating $x$ and $y$ as independent variables, we find that $U^\epsilon(x,y)$ is a solution of
\begin{equation}
\label{homogenization}
\left(-\left(\frac{\partial}{\partial x}+\frac{1}{\epsilon}\frac{\partial}{\partial y}\right)^2+V_0(x)+q(x,y)-k^2\right)U^\epsilon(x,y)=0.
\end{equation}

We then formally expand $U^\epsilon(x,y)$ as
 \begin{equation}
 U^\epsilon(x,y) = \sum_{j=0}^\infty \epsilon^j\ U_j(x,y) ,\label{expansion}
 \end{equation}
and require that 
\begin{align}
& U_j(x,y+1)\ =\ U_j(x,y),\ \ j\ge0,\nn\\
& U_0(x,y)-e^{ikx}, \ \ U_j(x,y)\ \ j\ge1\ \ {\rm outgoing\ as}\ |x|\to\infty,\label{Ujconditions}\\
& \left. U_j(x,y)\right|_{y=x/\epsilon}\ {\rm satisfies\ jump\ conditions}~\eqref{JumpCondition}.
\nn
\end{align}

The problem is solved by substituting 
the expansion~\eqref{expansion} into~\eqref{homogenization} and imposing the equation, jump conditions and radiation condition at each order in $\epsilon$. 
The differential equation becomes
\begin{equation}
 \left(-\left(\frac{\partial}{\partial x}+\frac{1}{\epsilon}\frac{\partial}{\partial y}\right)^2+V_0(x)+q(x,y)-k^2\right)\ U^\epsilon\left(x,\frac{x}{\epsilon}\right)=\sum_{j=-2}^\infty\ \epsilon^j\ r_j\ =\ 0,
\end{equation}
implying the following hierarchy of equations at each order in $\epsilon$ 
\begin{subequations}
\begin{align}
\label{H1} {\mathcal{O}(\epsilon^{-2})}\ &\qquad r_{-2}=-\partial^2_y U_0\ =\ 0 , \\
 \label{H2} \mathcal{O}(\epsilon^{-1})&\qquad r_{-1}=-\partial^2_y U_1 - 2\partial_y \partial_x U_0\ =\ 0 ,\\
\label{H3}
 \mathcal{O}(\epsilon^{0})&\qquad\ \ r_{0}= -\partial^2_y U_2 - 2\partial_y \partial_x U_1 - \partial_x^2U_0+(V_0+q)U_0-k^2U_0 = 0,\\
\label{H4} \mathcal{O}(\epsilon^{1})&\qquad\ \ r_{1}= -\partial^2_y U_3 - 2\partial_y \partial_x U_2 - \partial^2_x U_1 +(V_0+q)U_1-k^2 U_1\ =\ 0 ,\\
\label{H5} 
\mathcal{O}(\epsilon^{2})&\qquad\ \ r_{2}= -\partial^2_y U_4 - 2\partial_y \partial_x U_3 - \partial^2_x U_2 +(V_0+q)U_2-k^2 U_2\ =\ 0 ,\\
\label{H6} 
\mathcal{O}(\epsilon^{3})&\qquad\ \ r_{3}= -\partial^2_y U_5 - 2\partial_y \partial_x U_4 - \partial^2_x U_3 +(V_0+q)U_3-k^2 U_3\ =\ 0 ,\\ 
 \dots \ \ \ \dots \nn\\
\label{Hj} 
\mathcal{O}(\epsilon^{j})&\qquad\ \ r_{j} = r_j\left[U_{j+2},U_{j+1},U_j\right]\ =\ 0 \ . 
\end{align}
\end{subequations}
For example, to construct an approximate solution of~\eqref{homogenization} satisfying~\eqref{Ujconditions} up to the order 3, we solve simultaneously the equations $r_j=0$ for $j=-2,\cdots,3$. This will determine the functions $U_0$, $U_1$, $U_2$ and $U_3$ which make $U^\epsilon$ an approximate solution through order $\mathcal{O}(\epsilon^3)$. Since $e_{V^\epsilon+}(x;k) - e^{ikx}$ is to be outgoing, we require $U_0-e^{ikx}$ and each $U_i$ ($i=1,\dots,3$) to satisfy the outgoing condition. We now proceed with the implementation.
\bigskip

\noindent{\bf Caveat lector!} {\it The formal expansion presented in the remainder of this section yields terms involving spatial derivatives 
 of $e_{V_0+}(x;k)$ and $q_j(x)$ of arbitrarily high order. Now $\partial_x e_{V_0+}(x;k)$ has jump discontinuities on $\supp V_{sing}$ and $q_j(x)$ has jump discontinuities. Hence, the expansion must viewed in a distributional sense, {\it e.g.} involving terms, such as $\partial_x^\alpha\delta(x-x_j)$ {\it etc.} Furthermore, when we impose the jump conditions~\eqref{JumpCondition} to the expansion, order by order in $\epsilon$, we shall throughout assign $[\partial_x^\alpha\delta(x-x_j)]_{x=x_j}=0$. Although seemingly risky, in section~\ref{sec:rigorous-theory} we give a complete rigorous proof of the expansion with error bounds.}
 \bigskip 
 
Beginning at $\mathcal{O}(\epsilon^{-2})$, one has from~\eqref{H1}
\begin{equation}\label{U0x} r_{-2}=0 \Longrightarrow \partial^2_y U_0=0 \Longrightarrow U_0(x,y)=U_0(x). \end{equation}
Consequently, one has from~\eqref{H2} 
\begin{equation}\label{U1x} r_{-1}=0 \Longrightarrow \partial^2_y U_1=- 2\partial_y \partial_x U_0=0 \Longrightarrow U_1(x,y)=U_1(x).\end{equation}

Recall that $y\mapsto q(x,y)$ is 1-periodic and $\int_0^1 q(x,y)\ dy=0$. Integration
 of the equation~\eqref{H3} with respect to $y$ yields:
\begin{equation}\int_0^1 r_0(x,y)\ dy=0 \Longrightarrow -\frac{d^2}{d x^2} U_0(x)+ V_0(x)U_0(x) -k^2 U_0(x)=0.\label{U0h}\end{equation}
Furthermore, since $U_0-e^{ikx}$ is outgoing, one has by Proposition~\ref{prop:outgoing}
 \begin{equation}
U_0(x)\equiv e_{V_0+}(x;k).\label{U0}\end{equation}

By~\eqref{U0h} and~\eqref{H3} leads to 
\begin{equation}r_{0}=0 \iff -\partial_y^2 U_2(x,y)+q(x,y) e_{V_0+}(x)=0.
\label{U2p}\end{equation}
Thus, we decompose $U_2$ as:
\[U_2=U_2^{(h)}(x)+U_2^{(p)}(x,y),\]
with $U_2^{(p)}(x,y)$ a particular solution, and $U_2^{(h)}(x)$ an homogeneous solution to be determined.

Again, since $y\mapsto q(x,y)$ is 1-periodic and $\int_0^1 q(x,y)\ dy=0$, when by 
~\eqref{H4},
\begin{equation}\int_0^1 r_1(x,y)\ dy=0 \Longrightarrow -\frac{d^2}{d x^2} U_1(x)+ V_0(x)U_1(x) -k^2 U_1(x)=0.\label{U1h}\end{equation}
Since $U_1$ is outgoing, we claim
\begin{equation}
\label{U1}
U_1 \equiv 0.
\end{equation}
Indeed, in this case $k^2$ is a scattering resonance energy and $U_1$ its corresponding mode. Scattering resonances necessarily satisfy $\Im k^2 <0$
~\cite{Tang-Zworski}. However, $k^2\in\RR$ and hence $U_1\equiv0$.

Consequently, 
 \begin{equation}r_{1}=0 \Longrightarrow -\partial_y^2 U_3(x,y)-2\partial_y\partial_xU_2=0.\label{U3p}\end{equation}
In the same way as for $U_2$, we decompose $U_3$ as \[U_3=U_3^{(h)}(x)+U_3^{(p)}(x,y),\] with $U_3^{(p)}(x,y)$ a particular solution, and $U_3^{(h)}(x)$ an homogeneous solution to be determined. 

Integration of the equations~\eqref{H5} and~\eqref{H6} with respect to $y$, respectively, yields:
\begin{align}\label{U2h}
 &-\frac{d^2}{d x^2} U_2^{(h)}(x)+ V_0(x)U_2^{(h)}(x) +\int_0^1 q(x,y)U_2^{(p)}(x,y)dy -k^2 U_2^{(h)}(x) = 0. \\
 \label{U3h}
 &-\frac{d^2}{d x^2} U_3^{(h)}(x)+ V_0(x)U_3^{(h)}(x) -k^2 U_3^{(h)}(x) = -\int_0^1 U_3^{(p)}(x,y)q(x,y)\ dy .
\end{align}

We now solve~\eqref{U2p},~\eqref{U3p},~\eqref{U2h} and~\eqref{U3h} to obtain a unique (approximate) solution satisfying both outgoing and jump conditions, as we see in the following.
First, we use the decomposition in Fourier series of $q(x,y)$ in $y$ :
\[q(x,y)=\sum_{j\neq 0} q_j(x) e^{2i\pi jy}.\]
Consequently, equation~\eqref{U2p} leads immediately to
\begin{equation}
\label{solU2p}
U_2^{(p)}(x,y)=\ -\frac{ e_{V_0+}(x;k)}{4\pi^2} \sum_{|j|\geq 1} \frac{q_j(x)}{j^2} e^{2i\pi jy}.
\end{equation}

From~\eqref{U3p}, one deduces 
\[\partial_y^2 U_3(x,y)=\frac{i}{\pi} \sum_{|j|\geq 1}\frac{\partial_x( e_{V_0+}(x;k)q_j(x))}{j}e^{2i\pi jy}.\]
A particular solution $U_3^{(p)}(x,y)$ is therefore given by 
\begin{equation}
\label{solU3p}
U_3^{(p)}(x,y)=\ - \frac{i}{4\pi^3}\sum_{|j|\geq 1}\frac{\partial_x( e_{V_0+}(x;k)q_j(x))}{j^3}e^{2i\pi jy}.
\end{equation}
Then, using the Fourier series of $q$ and $U_2^{(p)}$, we obtain the following equations from~\eqref{U2h} and~\eqref{U3h}:
\begin{align}
\label{solU2h}
&-\frac{d^2}{d x^2} U_2^{(h)}(x)+ V_0(x)U_2^{(h)}(x) -k^2 U_2^{(h)}(x) = \frac{ e_{V_0+}(x;k)}{4\pi ^2}\sum_{|j|\geq 1} \frac{|q_j(x)|^2}{j^2},\ \ {\rm and}\\
\label{solU3h}
&-\frac{d^2}{d x^2} U_3^{(h)}(x)+ V_0(x)U_3^{(h)}(x) -k^2 U_3^{(h)}(x) \nn\\
& \qquad  = -\int_0^1 U_3^{(p)}(x,y)q(x,y)\ dy \ = \ \frac{i}{4\pi ^3}\sum_{|j|\geq 1} \frac{\partial_x(e_{V_0+}(x;k)q_j(x))q_{-j}(x)}{j^3}.
\end{align}
By Proposition~\ref{inhomog-solve}, equations~\eqref{solU2h} and~\eqref{solU3h} have unique outgoing solutions.
We refer to the expansion of $U^\epsilon$ obtained in this way as the \bigskip

\noindent {\bf Bulk (homogenization) expansion:}
\begin{equation}
U^\epsilon(x,y)\ =\ e_{V_0+}(x;k)\ +\ \epsilon^2\left( U_2^{(p)}(x,y) + U_2^{(h)}(x) \right)\ +\ 
 \epsilon^3\left( U_3^{(p)}(x,y) + U_3^{(h)}(x) \right)\ +\ \dots
 \label{bulkexpansion}\end{equation}
 
 It consists of a leading order average term (homogenization) plus correctors at each order in $\epsilon$ due to microstructure. \bigskip
 
\noindent\textbf{Failure of Jump conditions at interfaces:}

 Recall that we seek a solution which satisfies the jump conditions~\eqref{JumpCondition} on $U^\epsilon(x,y)$ 
 for all $(x,y)=(x,x/\epsilon)$ at each order in $\epsilon$. The leading order term, $e_{V_0+}$ satisfies all jump conditions. Now consider the terms $U_j^{(p)}(x,y) + U_j^{(h)}(x)$, arising at order $\mathcal{O}(\epsilon^j)$. 
  By construction, $U_j^{(h)}$ satisfies~\eqref{JumpCondition}. However $U_j^{(p)}(x,x/\epsilon)$ does not. Indeed, for the cases $j=2,3$, referring to expressions~\eqref{solU2p}
  and~\eqref{solU3p} we observe violation of~\eqref{JumpCondition} in $U_j^{(p)}(x,x/\epsilon)$ at discontinuities of $q_j(x)$ and $e_{V_0+}(x;k)$, and their derivatives.
  
  More precisely, the jump conditions for $U_2^{(p)}$ fail at $a_l$ ($l=1,\ldots,M$) each point of discontinuity of $q(x,x/\epsilon)$, since one has 
  \begin{align}
 \left[U_2^{(p)}(x,\frac{x}{\epsilon})\right]_a &= F^\epsilon_{2,a} , \label{JumpU2}\\
 \left[\frac{d}{dx}U_2^{(p)}(x,\frac{x}{\epsilon})\right]_a &=\frac1\epsilon G^\epsilon_{2,a}\ +\ H^\epsilon_{2,a} .\label{JumpU2prime}
\end{align}
with
\begin{equation}\label{jumpconditionU2p}
 \begin{array}{rl}
 F^\epsilon_{2,a} & \equiv \displaystyle\frac{-1}{4\pi^2}\sum_{|j|\geq 1} e_{V_0+}(a;k) \left[q_j\right]_a \frac{ e^{2i\pi ja/\epsilon}}{j^2} ,\\
 G^\epsilon_{2,a} & \equiv \displaystyle \frac{-i}{2\pi}\sum_{|j|\geq 1} e_{V_0+}(a;k)\left[q_j\right]_a \frac{ e^{2i\pi j a/\epsilon}}{j} ,\\
 H^\epsilon_{2,a} & \equiv \displaystyle \frac{-1}{4\pi^2}\sum_{|j|\geq 1} \left[\partial_x( e_{V_0+}(x;k)q_j(x))\right]_a \frac{ e^{2i\pi ja/\epsilon}}{j^2}.\end{array}
 \end{equation}
 
 In the same way, the jump conditions for $U_3^{(p)}$ fail at points of discontinuity of the functions $q(x,x/\epsilon)$ and $\partial_x q(x,x/\epsilon)$, and for
 ${x\in\{x_0,\cdots,x_{N-1}\}}$ the support of $V_{sing}$\ (\ recall: $V_{sing}=\sum_{j=0}^{N-1} c_j \ \delta (x-x_j)$\ ): 
\begin{align}
 \left[U_3^{(p)}(x,\frac{x}{\epsilon})\right]_a &= F^\epsilon_{3,a} , \label{JumpU3}\\
 \left[\frac{d}{dx}U_3^{(p)}(x,\frac{x}{\epsilon})\right]_a &=\frac1\epsilon G^\epsilon_{3,a}\ +\ H^\epsilon_{3,a}, \label{JumpU3prime}
\end{align} 
with $F^\epsilon_{3,a}$ and $H^\epsilon_{3,a}$ bounded highly oscillating functions and
\begin{equation}\label{jumpconditionU3p}
 G^\epsilon_{3,a} \equiv \frac{1}{2\pi^2}\sum_{|j|\geq 1}\left[\partial_x( e_{V_0+}(x;k)q_j(x))\right]_a \frac{e^{2i\pi ja/\epsilon} }{j^2}.
 \end{equation}
 $F_{3,a}^\epsilon$ and $H_{3,a}^\epsilon$ can be made explicit, but we omit these expressions as they contribute only at $\mathcal{O}(\epsilon^3)$.
 \bigskip
 
 \noindent\textbf{Restoring the Jump Conditions at interfaces:}
 
 In order to restore the jump conditions~\eqref{JumpCondition}, we must add to the expansion, at each point where the jump conditions are not satisfied, an appropriate {\it corrector}. These correctors each solve a non-homogeneous equation, driven 
 by the jumps in the bulk expansion~\eqref{bulkexpansion}.
 
 To see this, first note that $\left[\ \frac{d}{dx}U_j^{(p)}\ \right]_a=\mathcal{O}(\epsilon^{-1}),\ j=2,3$. Since $U_j^{(p)}$ contributes at order $\epsilon^j$, this suggests adding a corrector at order $\epsilon^{j-1}$. Thus, we introduce the
 
 \noindent{\bf Bulk expansion with corrector terms:}
\begin{align}
U^\epsilon(x,y)\ =\ e_{V_0+}(x;k)\ +\ \epsilon\ \mathcal{U}_1^\epsilon(x)\ +\ \epsilon^2\left(\ U_2^{(p)}(x,y) + U_2^{(h)}(x) + \mathcal{U}_2^\epsilon(x)\right) &\nn\\
  +\ \epsilon^3\left(\ U_3^{(p)}(x,y) + U_3^{(h)}(x) + \mathcal{U}_3^\epsilon(x)\right)\ +\ \dots&
\label{correctedexpansion}\end{align}
{\it The interface correctors $\mathcal{U}_j^\epsilon(x)$ are to be determined so that, at each order in $\epsilon$, 
the expansion~\eqref{correctedexpansion} satisfies the jump conditions~\eqref{JumpCondition}, the differential equation 
\eqref{eV+eqn} and outgoing radiation condition.}\medskip
\\
We construct 
$\mathcal{U}_j^\epsilon(x),\ j=1,2$
 below. The general construction uses the following 
   %
  %
  \begin{lemma} \label{lem:piecewiseU} Let $F_1,F_2\in \RR$ and $V_0=V_{sing}+V_{reg}$ as in~\eqref{Vdecomp}. Then there exists $\mathcal{U}(x)$, an outgoing piecewise $C^2$ solution of:
 \begin{equation}
\left(-\frac{d^2}{d x^2}+V_0(x)-k^2\right)\mathcal{U}=0 \mbox{, for } x<a \mbox{ and } x>a, 
\label{Ubl-eqn}\end{equation}
        which also satisfies the following jump conditions at the point $x=a$:
		\begin{align*}
& \left[ \mathcal{U}(x) \right]_a\ =\ F_1, \nn\\ 
 &\left[ \frac{d}{dx}\mathcal{U}(x)\ \right]_a\ -\ c\ \mathcal{U}(a-)\ =\ F_2. \nn
 \end{align*}
Here, $\mathcal{U}(a-)=\lim\limits_{x\uparrow a}{\mathcal{U}(x)}$, and the constant 
\begin{equation}\label{constant-c}
c=\left\{ \begin{array}{lr}
           0 &{\rm if }\ a\notin\supp V_{sing}, \\
      c_{j_0} &{\rm if }\ a=x_{j_0}\in\supp V_{sing}\ ; \\
      \end{array}\right.
\end{equation}
recall $V_{sing}(x)=\sum_{j=0}^{N-1} c_j \ \delta (x-x_j)$. 
 \medskip

$\mathcal{U}(x)$ has the form 
\begin{equation}\label{piecewiseU}\mathcal{U}(x)=\left\{ \begin{array}{lr}
       \alpha e_{V_0-}(x;k) &{\rm if }\ x<a, \\
	   \beta e_{V_0+}(x;k) &{\rm if }\ x>a, \\
      \end{array}\right.
\end{equation}
for appropriate choice of $\alpha$ and $\beta$, namely
\begin{equation}
\alpha=\frac{F_2 e_{V_0+}(a;k)\ - \ F_1 \partial_x e_{V_0+}(a+;k)}{2ik\ t_0^{hom}(k)} \ \ {\rm and} \ \ \beta=\frac{F_2 e_{V_0-}(a;k)\ - \ F_1 \partial_x e_{V_0-}(a+;k)}{2ik\ t_0^{hom}(k)}. \label{alpha-beta}
\end{equation}
\end{lemma}

 Before giving the proof, we explain why choosing $\mathcal{U}_j^\epsilon$ as in Lemma~\ref{lem:piecewiseU} does not change the bulk expansion~\eqref{bulkexpansion} constructed above. Therefore, our approach which first computes the bulk-expansion and then the correctors is consistent.
 
 As pointed out the expressions in the bulk expansion~\eqref{bulkexpansion}
\begin{equation}
U_j(x,y)\ =\ U_j^{bulk}(x,y)\ \equiv\ U_j^{(p)}(x,x/\epsilon) + U_j^{(h)}(x)
\label{Ujbulk}
\end{equation}
do not satisfy jump conditions~\eqref{JumpCondition}. Suppose now that we replace the functions 
$U_j(x,y)=U_j^{bulk}(x,y)$ by $U_j(x,y)=U_j^{bulk}(x,y)+ \mathcal{U}_a^\epsilon(x)$
and we seek $\mathcal{U}_a^\epsilon(x)$ so as to ensure jump conditions~\eqref{JumpCondition}. (Assume only one corrector is required).
Note that since $\mathcal{U}_a^\epsilon(x)$ lies in the kernel of $\partial_y$, adding such a term has no effect on the equations determining $U_j^{(p)}(x,y)$. Further, 
we want to preserve the form of $U_j^{(h)}(x)$, which 
has previously been constructed. Thus,
\begin{align}
r_j\left[U_{j+2},U_{j+1},U_j^{bulk} +\mathcal{U}_a^\epsilon\right]\ =\ & r_j\left[U_{j+2},U_{j+1},U_j^{bulk}\right]\nn \\
 &+ \left(-\partial_x^2+V_0(x)-k^2\right)\mathcal{U}_a^\epsilon(x)\ +\ q(x,y)\mathcal{U}_a^\epsilon(x).
 \label{rj-perturbed}
 \end{align}
 The equation for $U_j^{(h)}(x)$ is obtained by averaging~\eqref{rj-perturbed} with respect to $y$. Since $q(x,y)$ has mean zero with respect to $y$, this gives
 \begin{equation}
 \int_0^1 r_j\left[U_{j+2},U_{j+1},U_j^{bulk}\right](x,y)\ dy\ +\ \left(-\partial_x^2+V_0(x)-k^2\right)\mathcal{U}_a^\epsilon(x)\ =\ 0.
 \label{rj-integrated}\end{equation}
 Thus, if we choose $\mathcal{U}_a^\epsilon(x)$ to satisfy~\eqref{Ubl-eqn}, then the second term in~\eqref{rj-integrated} vanishes and 
 the equation for $U_j^{(h)}(x)$ is preserved. Therefore, if Lemma 
~\ref{lem:piecewiseU} is used to determine the jump-driven correctors at each order in 
 $\epsilon$, then the corrected bulk expansion~\eqref{correctedexpansion} is the solution we seek.

\bigskip

\noindent{\it Proof of Lemma~\ref{lem:piecewiseU}:}
 The piecewise form of $\mathcal{U}$~\eqref{piecewiseU} satisfies the outgoing radiation condition, by construction. The constants $\alpha$ and $\beta$ are determined by the jump conditions. 
 
Using the fact that $ e_{V_0+}(x;k)$ and $e_{V_0-}(x;k)$ satisfy the jump conditions~\eqref{JumpCondition}, one has 
 \[ \left\{ \begin{array}{rl}
 \left[ \mathcal{U}(x) \right]_a &=\beta e_{V_0+}(a;k) -\alpha e_{V_0-}(a;k), \\
\left[ \frac{d}{dx}\mathcal{U}(x) \right]_a\ - \ c\mathcal{U}(a-) &=\beta \partial_x e_{V_0+}(a+;k) -\alpha \partial_x e_{V_0-}(a-;k) - c \alpha\ e_{V_0-}(a-;k) \\
&=\beta \partial_x e_{V_0+}(a+;k) -\alpha \partial_x e_{V_0-}(a+;k).
\end{array} \right. \]
 Solving this inhomogeneous system, using the value of the Wronskian, given in~\eqref{wronskian}, leads immediately to~\eqref{alpha-beta}. This completes the proof of Lemma~\ref{lem:piecewiseU}.\ \endproof \medskip
 
 We now proceed to apply Lemma~\ref{lem:piecewiseU} to determine the correctors associated with $U_2^{(p)}$ and $U_3^{(p)}$. Using~\eqref{JumpU2}-\eqref{JumpU2prime} 
 and~\eqref{JumpU3}-\eqref{JumpU3prime}, the jump conditions~\eqref{JumpCondition} applied to $\epsilon \mathcal{U}_1^\epsilon\ + \ \epsilon^2 \mathcal{U}_2^\epsilon\ + \ \epsilon^2 U_2^{(p)}\ +\ \epsilon^3 U_3^{(p)}$ read:
 \begin{align}
\epsilon\ \left[\ \mathcal{U}_1^\epsilon\ \right]_a\ +\ 
\epsilon^2 \left(\ \ F_{2,a}^\epsilon\ +\ \left[\ \mathcal{U}_2^\epsilon\ \right]_a\ \right)\ =\ \mathcal{O}(\epsilon^3)&,\label{jumpUepsilon}\\
\epsilon \Big( G_{2,a}^\epsilon\ +\ \Big[\ \frac{d}{dx}\mathcal{U}_1^\epsilon\ \Big]_a -c \mathcal{U}_1^\epsilon(a-)\Big)\ +\ 
\epsilon^2 \Big( H_{2,a}^\epsilon -c\ U_2^{(p)}(a-) + G_{3,a}^\epsilon &\nn\\ + \Big[\ \frac{d}{dx}\mathcal{U}_2^\epsilon\ \Big]_a -c\ \mathcal{U}_1^\epsilon(a-)\Big) = \mathcal{O}(\epsilon^3)&.\label{jumpDUepsilon}
 \end{align}
 
Equations~\eqref{jumpUepsilon} and~\eqref{jumpDUepsilon} imply jump conditions at order $\epsilon$ and order $\epsilon^2$. Therefore, we construct $\mathcal{U}_{j,a}^\epsilon,\ j=1,2$ solving the two inhomogeneous problems at each point, $a$, of non-smoothness.

\noindent{\bf System for corrector $\mathcal{U}_{1,a}^\epsilon$:}
\begin{align}
&\left(-\frac{d^2}{d x^2}+V_0(x)-k^2\right)\mathcal{U}_{1,a}^\epsilon=0, \ \ \ x\ne a, 
\label{U1a-eqn}\\
& \left[\ \mathcal{U}_{1,a}^\epsilon\ \right]_a\ =\ 0,\ \ \ \left[ \frac{d}{dx}\mathcal{U}_{1,a}^\epsilon\ \right]_a\ -\ 
 c\ \mathcal{U}_{1,a}^\epsilon(a-)\ =\ -G_{2,a}^\epsilon. 
\label{U1a-jumps}
\end{align}

\noindent{\bf System for corrector $\mathcal{U}_{2,a}^\epsilon$:}
\begin{align}
&\left(-\frac{d^2}{d x^2}+V_0(x)-k^2\right)\mathcal{U}_{2,a}^\epsilon=0,\ \ \ x\ne a, 
\label{U2a-eqn}\\
& \left[\ \mathcal{U}_{2,a}^\epsilon\ \right]_a\ =\ -F_{2,a}^\epsilon, \ \ \ \ \left[ \frac{d}{dx}\mathcal{U}_{2,a}^\epsilon(x)\ \right]_a - c\ \mathcal{U}_{2,a}^\epsilon(a-)\ =\ -H_{2,a}^\epsilon-G_{3,a}^\epsilon+c\ U_2^{(p)}(a-). 
\label{U2a-jumps}
\end{align}

Lemma~\ref{lem:piecewiseU}, applied to~\eqref{U1a-eqn}-\eqref{U1a-jumps}
 and~\eqref{U2a-eqn}-\eqref{U2a-jumps} defines the unique correctors $\mathcal{U}_{1,a}^\epsilon$ and $\mathcal{U}_{2,a}^\epsilon$:
 $ \mathcal{U}_{1,a}^\epsilon$ is given by~\eqref{piecewiseU}, {\it i.e.}
 \begin{equation}\label{piecewiseU.1}
 \mathcal{U}_{1,a}^\epsilon(x)=\left\{ \begin{array}{lr}
       \alpha_{1,a}^\epsilon\ e_{V_0-}(x;k) &{\rm if }\ x<a, \\
       &\\
	   \beta_{1,a}^\epsilon\ e_{V_0+}(x;k) &{\rm if }\ x>a, \\
      \end{array}\right.
\end{equation}
with $\alpha_{1,a}^\epsilon$ and $\beta_{1,a}^\epsilon$ given by 
\begin{equation}
\alpha_{1,a}^\epsilon\ =\ -\frac{G_{2,a}^\epsilon}{2ik t_0^{hom}(k)}\ e_{V_0+}(a;k),\ \ \ 
\beta_{1,a}^\epsilon\ =\ -\frac{G_{2,a}^\epsilon}{2ik t_0^{hom}(k)}\ e_{V_0-}(a;k),
\label{alpha-beta-1a}\
\end{equation}
where $G_{2,a}^\epsilon$ is given in~\eqref{jumpconditionU2p}. Then, 
$ \mathcal{U}_{2,a}^\epsilon$ is given by~\eqref{piecewiseU} with $\alpha_{2,a}^\epsilon$ and $\beta_{2,a}^\epsilon$ 
$\alpha$ and $\beta$ given by
\begin{align}
\alpha_{2,a}^\epsilon &= \frac{1}{2ik t_0^{hom}(k)}\left( \left(-H_{2,a}^\epsilon-G_{3,a}^\epsilon +c\ U_2^{(p)}(a-)\right) e_{V_0+}(a;k) + F_{2,a}^\epsilon \partial_xe_{V_0+}(a+;k) \right), \nn\\
\beta_{2,a}^\epsilon &= \frac{1}{2ik t_0^{hom}(k)}\left( \left(-H_{2,a}^\epsilon-G_{3,a}^\epsilon+c\ U_2^{(p)}(a-)\right) e_{V_0-}(a;k) + F_{2,a}^\epsilon \partial_xe_{V_0-}(a+;k) \right), 
\label{alpha-beta-2a}\
\end{align}
where $H_{2,a}^\epsilon,\ F_{2,a}^\epsilon$ and $G_{3,a}^\epsilon$ are given in~\eqref{jumpconditionU2p} and~\eqref{jumpconditionU3p}.

Therefore at $\mathcal{O}(\epsilon)$, we define the corrector $\mathcal{U}^\epsilon_1$ as 
 \begin{equation}\label{defU1BL}
 \mathcal{U}^\epsilon_1=\sum_{j=1}^{M} \mathcal{U}^\epsilon_{1,a_j}
 \end{equation}
where $a_j,\ \ j=1,\dots,M$ denote the points of discontinuity of $q\left(x,x/\epsilon\right)$ .

At order $\mathcal{O}(\epsilon^2)$, we have a violation of the jump conditions~\eqref{JumpCondition} due to 
\begin{itemize}
\item[(i)] points of ``discontinuity'' of $q_j(x)$, {\it i.e.} $a_j$ ($j=1,\dots,M$) for which $\left[q_j\right]_a\ne0$ or $[\partial_x q_j]_a\ne0$, and 
\item[(ii)] the singular set $\supp V_{sing}=\{x_0,\dots,x_{N-1}\}$.
\end{itemize} 

Thus we construct, $\mathcal{U}_{2,a}^\epsilon$, for all $a$ in the set, $\Omega$, of non-smooth points of $V^\epsilon(x)$:
\begin{equation}
\Omega=\{x_0,\dots,x_{N-1}\}\cup\{-\infty=a_0,a_1,\dots,a_M=\infty\}
\label{Omega-def}
\end{equation}
 and define the corrector $\mathcal{U}^\epsilon_2$ by:
 \begin{equation}\label{defU2BL}
 \mathcal{U}^\epsilon_2=\sum_{a\in\Omega} \mathcal{U}^\epsilon_{2,a}.
 \end{equation}
 
 We summarize the preceding calculation in the following
 \begin{proposition}\label{expansion-summary}
\begin{align}
\label{devhomo}
e_{V^\epsilon+}(x;k)\ =\ & U^\epsilon\left(x,x/\epsilon\right)= e_{V_0+}(x;k)\ +\ \epsilon\ \mathcal{U}_1^\epsilon(x)\nn \\
& \ + \ \epsilon ^2\ \left(\ U_2^{(h)}(x)\ +\ U_2^{(p)}(x,x/\epsilon)\ +\ \mathcal{U}_2^\epsilon(x)\ \right) +\mathcal{O}(\epsilon ^3) 
\end{align}
gives a formal construction of the distorted plane wave $e_{V^\epsilon+}(x;k)$, through $\mathcal{O}(\epsilon^2)$ with error of size $\mathcal{O}(\epsilon^3)$.
 The correctors $\mathcal{U}_1^\epsilon(x)$ and $\mathcal{U}_2^\epsilon(x)$ are given by~\eqref{U1a-eqn}-\eqref{U1a-jumps} and~\eqref{U2a-eqn}-\eqref{U2a-jumps}.
 
 Finally, $U_j^{(p)}(x,y)$ and $U_j^{(h)}(x)$ are given by
\begin{align}
&U_2^{(p)}(x,y) = -\frac{e_{V_0+}(x;k)}{4\pi^2} \sum_{|j|\geq 1} q_j(x) \frac{e^{2i\pi jy}}{j^2},  \nn\\
&\left(-\frac{d^2}{d x^2}+V_0(x)-k^2\right)U_2^{(h)}(x) = 
\frac{e_{V_0+}(x;k)}{4\pi ^2}\sum_{|j|\geq 1} \frac{|q_j(x)|^2}{j^2},\ \ U_2^{(h)}\ {\rm outgoing}, \nn\\
&U_3^{(p)}(x,y)=-\frac{i}{4\pi^3}\sum_{|j|\geq 1}\partial_x(e_{V_0+}(x;k)q_j(x))\frac{e^{2i\pi jy}}{j^3},\nn\\
&\left(-\frac{d^2}{d x^2}+V_0(x)-k^2\right)U_3^{(h)}(x) = i\ \sum_{|j|\geq 1} \frac{\partial_x(e_{V_0+}(x;k)q_j(x))q_{-j}(x)}{4\pi ^3\ j^3},\ U_3^{(h)} \ {\rm outgoing}\nn.
\end{align}

\end{proposition}

\subsection{Expansion of the transmission coefficient, $t^\epsilon(k)$}\label{secThomo}

The results of the previous section can now be used to derive expansion~\eqref{eqn:Texpansion} for the transmission coefficient, $t^\epsilon(k)$ associated with the potential $V^\epsilon(x)$. $t^\epsilon(k)$, through order $\epsilon^2$ is derived by isolating appropriate terms in the expansion~\eqref{devhomo}. The sense in which the remainder is small is proved, by entirely different means, in section~\ref{sec:rigorous-theory}.
\begin{itemize}
\item[$\mathcal{O}(\epsilon^0)$:] The only term at order one is $ e_{V_0+}$, which gives the leading order transmission coefficient, $t_0^{hom}(k)$, corresponding to the average potential $V_0$.
\item[$\mathcal{O}(\epsilon^1)$:] At order $\epsilon$, we seek the contribution to $t_\epsilon(k)$ from $\mathcal{U}^\epsilon_1$. From~\eqref{piecewiseU} we have, since $e_{V_0+}(x;k)\sim t_0^{hom}(k)\ e^{ikx}$ as $x\to +\infty$, that 
 the contribution of $\mathcal{U}^\epsilon_{1,a}$ to the transmission coefficient is given by 
\[
 t^\epsilon_{1,a} =\ \beta_{1,a}^\epsilon\ t_0^{hom}(k)\ =\ \frac{1}{4\pi k} e_{V_0+}(a;k)e_{V_0-}(a;k) \sum_{|j|\geq 1} \left[q_j\right]_a \frac{ e^{2i\pi j a/\epsilon}}{j}.
 \]
Finally, summing over all contributions from points of discontinuity of $q_j$, one obtains the complete first order contribution from $\mathcal{U}^\epsilon_1$:
\begin{equation}\label{T1BL}
 t^\epsilon_1 = \sum_jt^\epsilon_{1,a_j}\ =\ \sum_{j=1}^{M} \frac{e_{V_0+}(a_j;k)e_{V_0-}(a_j;k)}{4\pi k} \sum_{|l|\geq 1} \left[q_l\right]_{a_j} \frac{e^{2i\pi l\frac{a_j}{\epsilon}}}{l} .
\end{equation}

\item[$\mathcal{O}(\epsilon^2)$:]

\subitem\ (a)\ {\it No contribution to $t^\epsilon(k)$ from $\epsilon^2 U_2^{(p)}$:}\  We estimate $U_2^{(p)}$ pointwise.
\begin{align*}
 |U_2^{(p)}\left(x,\frac{x}{\epsilon}\right)| \leq \frac1{4\pi^2} |e_{V_0+}(x;k)|\sum_{|j|\geq 1} \frac{|q_j(x)|}{j^2} 
 &\le C\ \Big(\sum_{|j|\ge1} |q_j(x)|^2\Big)^{1/2}\ \to\ 0,\ |x|\to\infty.
\end{align*}
Here we have used the uniform bound~\eqref{eqn:m-bound1} on $e_{V_0+}$ for $x\ge0$ 
and the hypothesis~\eqref{qL2to0}.
Since, $U_2^{(p)}(x,\frac{x}{\epsilon})\to0$ as $x\to\infty$, it does not contribute to the transmission coefficient.

\bigskip

\subitem\ (b)\ {\it Contribution of $\epsilon^2 U_2^{(h)}(x)$ to $t^\epsilon(k)$:}\medskip

\noindent From~\eqref{solU2h}, one has 
\[\left(-\frac{d^2}{d x^2}+V_0(x)-k^2\right)U_2^{(h)}(x) = \frac{e_{V_0+}(x;k)}{4\pi ^2}\sum_{|j|\geq 1} \frac{|q_j(x)|^2}{j^2}.\]
Using expression~\eqref{eqn:RV0} for the outgoing resolvent we have
\begin{align*}
U_2^{(h)}(x) = & R_{V_0}(k)\left(\frac{e_{V_0+}(\cdot;k)}{4\pi ^2}\sum_{|j|\geq 1} \frac{|q_j(\cdot)|^2}{j^2}\right) \\
 = & \frac{-1}{2ik\ t_0^{hom}(k)}\int_{-\infty}^{x} \frac{e_{V_0+}(\zeta;k)}{4\pi^2}\sum_{|j|\geq 1}\frac{|q_j(\zeta)|^2}{j^2}e_{V_0-}(\zeta;k)e_{V_0+}(x;k) \ d\zeta 
 \\& +\frac{-1}{2ik\ t_0^{hom}(k)}\int_{x}^{+\infty} \frac{e_{V_0+}(\zeta;k)}{4\pi^2}\sum_{|j|\geq 1}\frac{|q_j(\zeta)|^2}{j^2}e_{V_0+}(\zeta;k)e_{V_0-}(x;k) \ d\zeta.
\end{align*}
Therefore, since $q_j\in L^2$ for all $j\in\ZZ$, and $e_{V_0\pm}\in L^\infty$, one has when $x\to\infty$:
\[
\lim\limits_{x\to\infty}{U_2^{(h)}(x)- \left(\frac{-e_{V_0+}(x;k)}{8ik\pi^2\ t_0^{hom}(k)}\int_\RR \sum_{|j|\geq 1}\frac{|q_j(\zeta)|^2}{j^2}e_{V_0-}(\zeta;k)e_{V_0+}(\zeta;k) d\zeta \right)}=0.\]
It follows that the contribution of $U_2^{h}(x)$ to the transmission coefficient is
\begin{equation}\label{T2hom}
 t_2^{hom}(k) \equiv -\frac{1}{8ik\pi^2}\int_\RR \sum_{|j|\geq 1}\frac{|q_j(\zeta)|^2}{j^2}e_{V_0-}(\zeta;k)e_{V_0+}(\zeta;k) d\zeta .
\end{equation}

\subitem\ (c)\ {\it Contribution of $\epsilon^2\mathcal{U}^\epsilon_2$ to $t^\epsilon(k)$:} \medskip

\noindent We study $\mathcal{U}^\epsilon_2$ as above.
 From~\eqref{piecewiseU} we have, since $e_{V_0+}(x;k)\sim t_0^{hom}(k)\ e^{ikx}$ as $x\to +\infty$, that 
 the contribution of $\mathcal{U}^\epsilon_{2,a}$ to the transmission coefficient is given by $t_2^\epsilon\equiv t_0^{hom}(k)\ \beta_{2,a}^\epsilon$.

From~\eqref{alpha-beta-2a},~\eqref{jumpconditionU2p} and~\eqref{JumpU3prime} we have
\begin{align*}
t_{2,a}^\epsilon(k)\ &=\ \beta_{2,a}^\epsilon\ t_0^{hom}(k) \\
 &=\frac{-e_{V_0-}(a;k)}{8\pi^2ik}\ \sum_{|j|\geq 1}\left( c\ e_{V_0+}(a;k) q_j(a-) +
 \left[ \partial_x( e_{V_0+}(x;k)q_j(x)) \right]_a\ \right) \frac{ e^{2i\pi ja/\epsilon}}{j^2}\\ 
 &\ \ \ \ \ -\ \frac{1}{8\pi^2ik}\ \sum_{|j|\geq 1} e_{V_0+}(a;k) \left[q_j\right]_a \partial_x e_{V_0-}(a+;k) \frac{ e^{2i\pi ja/\epsilon}}{j^2}.
 \end{align*}
 Using the easily verified identity
 \begin{equation}
 \left[\ q_j(x)\partial_x e_{V_0-}(x;k)\ \right]_a\ 
 = \partial_x e_{V_0-}(a+;k) \left[\ q_j(x)\ \right]_a\ +\ c\ e_{V_0-}(a;k)\ q_j(a-),
 \label{identity}
 \end{equation}
we obtain
 \begin{align*}
t_{2,a}^\epsilon(k)\ 
 &=\ \frac{-1}{8\pi^2ik}\sum_{|j|\geq 1} \Big(\left[\partial_x( e_{V_0+}(x;k)q_j(x))\right]_a e_{V_0-}(a;k)\\
 &\ \ \ \ + \left[q_j(x) \partial_x e_{V_0-}(x;k)\right]_a e_{V_0+}(a;k) \Big)\frac{ e^{2i\pi ja/\epsilon}}{j^2} \\
 &=\frac{i}{8\pi^2k}\sum_{|j|\geq 1} \left[\partial_x( e_{V_0-}(x;k)e_{V_0+}(x;k)q_j(x))\right]_a \frac{ e^{2i\pi ja/\epsilon}}{j^2} .
\end{align*}
Finally, summing over all contributions of all singular and / or discontinuity points of $V^\epsilon$, we obtain the simple expression:
\begin{equation}\label{T2BL}
 t^\epsilon_2(k) = \sum_{a\in\Omega} t_{2,a}^\epsilon(k) = \frac{i}{8\pi^2 k}\sum_{a\in\Omega} \sum_{|l|\geq 1} \left[\partial_x( e_{V_0-}(x;k)e_{V_0+}(x;k)q_l(x))\right]_a \frac{ e^{2i\pi l a/\epsilon}}{l^2}.
\end{equation}
\medskip

\item[$\mathcal{O}(\epsilon^3)$]:\ By similar considerations to the above discussion 
of $U_2^{(h)}$ and $U_2^{(p)}$, the terms $U_3^\epsilon= U_3^{(h)}+U_3^{(p)}$ in the expansion of $e_{V^\epsilon_+}$ give a correction to $t^\epsilon(k)$ of order $\epsilon^3$, and is therefore subsumed by the error term in the expansion~\eqref{eqn:Texpansion}.
\end{itemize}
\bigskip

In summary, we have an expansion of $t^\epsilon(k)$, agreeing with the expansion~\eqref{eqn:Texpansion} in Theorem~\ref{thm:Texpansion}:
\bigskip

\begin{proposition}\label{prop:formal-expansion-of-tepsilon}
{\bf Formal corrected homogenization expansion:}
\begin{equation}\label{Texp-hom}
t^\epsilon(k) =t^{hom}_0(k)\ +\ \epsilon\ t^\epsilon_1(k) + \epsilon^2\ \left(\ t^{hom}_2(k) + t^\epsilon_2(k)\ \right)\ +\
 \mathcal{O}(\epsilon^3),
 \end{equation}
 where the leading order term, $t_0^{hom}(k)$, is the transmission coefficient associated with the homogenized (average with respect to the fast scale) potential 
$V_0$, $t_2^{hom}(k)$ is a classical homogenization theory corrector given by~\eqref{T2hom}, and $t_j^\epsilon,\ j=1,2$ are interface correctors given by
 ~\eqref{T1BL} and~\eqref{T2BL}.
 
 Note that if $V_0$ is generic, then since using that $t_{V_0}(k)$ and $e_{V_0\pm}(x,k)$ are $ \mathcal{O}(k)$ as $k\to0$, we see the expansion is formally valid for any $k\in\RR$. However, if $V_0$ is not generic then we must exclude $k=0$; see the discussion of Remark~\ref{rem:k=0}.
 \end{proposition}


\section{Rigorous analysis of the scattering problem}\label{sec:rigorous-theory}
In the preceding section, we applied the classical method of multiple scales to derive a formal expansion for the distorted plane wave $e_{V^\epsilon+}(x;k)$ and transmission coefficient $t^\epsilon(k)$; see section~\ref{sec:background&statement}.  
For sufficiently smooth potentials, this expansion satisfies, at each order in $\epsilon$, all necessary 
continuity conditions as well as the radiation condition at infinity; see Definitions~\ref{def:solution} and~\ref{def:outgoing}.

 We found, however, that if 
the potential is non-smooth this expansion, while valid {\it in the bulk}, violates continuity conditions at 
(i) discontinuities, and (ii) at strong singularities of 
the background, unperturbed potential, $V_0(x)=V_{reg}(x)+V_{sing}(x)$. We found, in Section~\ref{sechomo} that we can, ``by hand'', construct interface correctors for each point of non-smoothness, thereby giving a corrected expansion (bulk expansion plus interface correctors) which is a valid solution to any finite order in $\epsilon$. The expansion of Proposition~\ref{prop:formal-expansion-of-tepsilon}
is explicit through order $\epsilon^2$ with order $\epsilon^3$ correctors.
 
\noindent{\bf Question:} Does the procedure of section~\ref{sechomo} yield a valid expansion with an error terms satisfying an appropriate higher order error bound? 
 
It turns out that the formal expansion is correct with an appropriate error estimate. However, we obtain this result, not by expansion in scalar $\epsilon$ but rather in the the function $q_\epsilon(x)$, with respect to which there is an analytic perturbation theory in an appropriate function space . Smallness required for control of the perturbation expansion derives from $q_\epsilon(x,x/\epsilon)$ being supported at high frequencies if $\epsilon$ is small. The principle terms, displayed in the expansion of Proposition~\ref{prop:formal-expansion-of-tepsilon} (and indeed the terms at any finite order in the small parameter, $\epsilon$), are obtained via small $\epsilon$ asymptotics of the leading order terms in the $q_\epsilon$ expansion. The approach we use was introduced by Golowich and Weinstein in~\cite{GW:05}.

\subsection{Formulation of the problem}

We consider the general one-dimensional scattering problem 
\begin{align}
&\left(-\frac{d^2}{d x^2}+V_0(x)+Q(x)-k^2\right) e_{V+}=0\nn\\
&e_{V_0+Q,+}(x;k)\ -\ e^{ikx}\ \to0,\ \ x\to-\infty,\
\label{general-1d}
\end{align}
where $V_0(x)$ as hypothesized in section~\ref{sec:introduction} and $Q$ is a 
spatially localized perturbing potential, which we think of as being spectrally supported at high frequencies.
 $Q$ may be large in $L^\infty$. As a model, we have in mind $Q(x)=q_\epsilon(x)=q(x,x/\epsilon)$,
 with $\epsilon$ small.

We introduce the scattered field, $u_s$, via 
\begin{equation}
e_{V+}(x;k) = e_{V_0+}(x;k)+u_s(x;k)\label{us-def}
\end{equation}
 where $u_s$ is outgoing as $x\to\pm\infty$. 
Therefore, $u_s$ is the solution of
\begin{equation}
\label{scat}
\left(-\frac{d^2}{d x^2}+(V_0+Q)(x)-k^2\right)u_s(x;k)=-Q e_{V_0+}(x;k),
\end{equation}
with outgoing conditions : $\left(\ \partial_x\ \mp\ ik\ \right)u_s \to\ 0,\ x\to\pm\infty$.

Applying the outgoing resolvent, $R_{V_0}(k)$, to~\eqref{scat} and rearranging terms we obtain the Lippman-Schwinger equation
\begin{align}
u_s\ &=\ -\left(I+R_{V_0}(k)Q\right)^{-1}\ R_{V_0}(k)\ Q\ e_{V_0+}(\cdot;k)\ \ \implies\nn\\
e_{V+}(x;k) &=\ e_{V_0+}(x;k)\ -\left(I+R_{V_0}(k)Q\right)^{-1}\ R_{V_0}(k)\ Q\ e_{V_0+}(\cdot;k) .
\label{us-int-eqn}
\end{align}
Consider now the {\it formal} Neumann expansion, obtained from~\eqref{us-int-eqn}. 
\begin{align}
e_{V+}(x;k)\ 
&= e_{V_0+}(x;k) - R_{V_0}(k)(Q e_{V_0+}(x;k)) + R_{V_0}(k) Q R_{V_0}(k) (Qe_{V_0+}(x;k)) + \dots \nn \\
&= e_{V_0+}(x;k)\ +\ \sum_{m=0}^\infty \left[\left(-R_{V_0}(k)\ Q\right)^m\ e_{V_0+}\right](x;k).
\label{devscat}\end{align}
In this section we show for a class of 
 $Q$, which include high-contrast (pointwise large) microstructure (highly oscillatory) potentials that the expansion~\eqref{devscat} converges in an appropriate sense and that any truncation satisfies an error bound. 

\subsection{Reformulation of the Lippman-Schwinger equation and the norm $ |||Q|||$ }
\label{sec:Formulation-LS-equation}
We seek a reformulation of the Lippman-Schwinger equation~\eqref{us-int-eqn} in which it is explicitly clear that if $Q$ is  highly oscillatory, then the terms of the Neumann series are successively smaller.
 Introduce, via the Fourier transform, the operator $\D^{s}$ 
\begin{equation}
\D^{s} g \equiv (I-\Delta)^{s/2}g \equiv \frac{1}{2 \pi} \int_{-\infty}^{+\infty}e^{ix\xi}\ (1+\xi^2)^{s/2}\ \hat{g}(\xi)d\xi .
\label{D0inv}
\end{equation}
and the localized function $\chi$
\begin{equation}
\label{defChi}
\chi(x) = \langle x \rangle^{-\sigma}\ =\ \left(1+x^2\right)^{-\frac{\sigma}{2}},\ \ \sigma>4.
\end{equation}
Now introduce the {\it spatially and frequency weighted distorted plane wave}, $E_{V+}(x;k)$, given by:
\begin{equation}
 E_{V+}(x;k)\ \equiv\ \left(\ \D\chi e_{V+}\ \right)(x;k).
\label{EV+} \end{equation}
With the operator definitions 
\begin{align}
T_{R_{V_0}}(k)&\equiv\D\chi R_{V_0}(k)\chi\D , \label{defTR} \\
T_Q&\equiv\D^{-1}\chi^{-1}Q \chi^{-1}\D^{-1}\ =\ \D^{-1} \langle x \rangle^{\sigma}\cdot Q\cdot \langle x \rangle^{\sigma}\D^{-1}\ , \label{defTQ}
\end{align}
$E_{V+}(x;k)$ can be seen to satisfy 
\begin{align}
\left(I+T_{R_{V_0}} \ T_Q\right)\ \big(\ E_{V+}(\cdot;k) - E_{V_0+}(\cdot;k)\ \big)=-\D\chi R_{V_0}(k)\ Q\ e_{V_0+}(x;k).
\label{scat2}
 \end{align}
 
 Here's the motivation for our strategy. 
 Note that $T_Q$ has the operator $\D^{-1}$ as both a pre- and post- multiplier. This has the effect of a high frequency cutoff. Therefore, for highly oscillatory $Q$, $T_Q$ is expected to be of small operator norm.  
 If the norm of $T_{R_{V_0}} \circ T_Q$ is small then $I+T_{R_{V_0}} \circ T_Q$ is invertible and we have
 the {\it preconditioned} Lippman-Schwinger equation
 \begin{align}
E_{V+}(\cdot;k)\ =\ E_{V_0+}(\cdot;k)\ -\ \left(I+T_{R_{V_0}} \ T_Q\right)^{-1}\ \D\chi R_{V_0}(k)\ Qe_{V_0+}(x;k).
\label{precond-LS}
 \end{align}
 
 We proceed now to construct a norm, $|||Q|||$, such that if $|||Q|||$ is small then $T_{R_{V_0}} \ T_Q$ is bounded and of small norm as an operator norm from $L^2$ to $L^2$. 

 The norm we choose for the perturbing potential is defined as follows:
 \begin{equation}
 |||Q|||\ \equiv\ \left\|T_Q\right\|_{L^2\to L^2}\ =\ \left\| \D^{-1}\ \langle x\rangle^{\sigma} Q \ \langle x\rangle^{\sigma}\D^{-1}  \right\|_{L^2 \to L^2} , \ \ \sigma>4.
 \label{Qnorm-def}
 \end{equation}

The next result establishes the expansion of the distorted plane waves $e_{V+}(x;k)$ 
 in a $H^1(\RR;\chi(x)dx)$ and therefore, by the Sobolev inequality, a $L^\infty(\RR;\chi(x)dx)$ convergent expansion for $|||Q|||$ sufficiently small.\bigskip

\begin{theorem}\label{thm:Expand-LS}
 Let $V$ satisfy Hypotheses~\textbf{(V)}, and $k\in K$ a compact subset of $\RR$, satisfying Hypothesis~\textbf{(G)}. Define 
\begin{equation*}
\tau_0(K)\ \equiv\ \frac{1}{\max_{k\in K} \left\| T_{R_{V_0}}(k)\right\|_{L^2\to L^2} } \ >\ 0.
\end{equation*}
If $|||Q|||<\tau_0(K)$, then for all $k\in K$:
\begin{itemize}
\item The preconditioned Lippman Schwinger equation~\eqref{precond-LS} has a unique spatially and spectrally weighted distorted plane solution, $E_{V+}(x;k)$.
\item This solution can be expressed as a series, which converges in $L^2(\RR)$, uniformly in $k\in K$:
\begin{align*}
E_{V+}(x;k)\ &=\ E_{V_0+}(x;k)\ +\ \sum_{m=0}^\infty\ \left(- T_{R_{V_0}}(k) \ T_{Q} \right)^m\left[ \D\chi R_{V_0}(k)\ Q\ e_{V_0+}\right](x;k)\\
&=\ E_{V_0+}(x;k)\ -\ \D\chi\ R_{V_0}(k)\ Q\ e_{V_0+}(x;k)\\
&\ \ \ \ \ \ +\ T_{R_{V_0}}(k) \ T_{Q}\ \D\chi\ R_{V_0}(k)\ Q\ e_{V_0+}(x;k) -\dots\end{align*}
\item It follows that the distorted plane wave, $e_{V_0+Q,+}(x;k)$ satisfies the approximation
for any $M\ge1$ \end{itemize}
 \begin{equation}
 \left\| \D\chi\ \left( u_s(\cdot;k)\ +\ \sum_{m=0}^M \left[\left(-R_{V_0}(k) Q\right)^m e_{V_0+}\right](\cdot;k)\ \right) \right\|_{L^2(\RR)} \le\ C 
 |||Q|||^{M+1} \label{QtoM}
 \end{equation}
with $u_s(x;k) \ \equiv \ e_{V_0+Q,+}(x;k) \ - \ e_{V_0+}(x;k)$.
\end{theorem}
\bigskip

\begin{remark} In the proof of Theorem~\ref{thm:Expand-LS}, the distinction between generic and non-generic cases arises through the properties of the unperturbed resolvent, $R_{V_0}(k)$, as $k\to0$; see Proposition~\ref{inhomog-solve}.
\end{remark}
\medskip

In the following, we prove that both $T_{R_{V_0}}$ and $T_Q$ are well-defined operators, bounded in $L^2$. Then, Theorem~\ref{thm:Expand-LS} follows immediately if $Q$ satisfies the smallness condition
\begin{equation}
\big\|T_Q \big\|_{L^2\to L^2}\ <\ \min_{k\in K}\ \left(\big\|T_{R_{V_0}}(k) \big\|_{L^2\to L^2}\right)^{-1}\ \equiv\ \tau_0(k).
\label{TQ-smallness}
\end{equation}

\begin{proposition}\label{prop:TQ}
Let $\langle x\rangle^{2\sigma} Q(x)\in L^2(\RR)$.\ 
 Then $T_Q$, as defined in~\eqref{defTQ}, is a Hilbert-Schmidt operator and is therefore compact. 
\end{proposition}
\medskip

\begin{proposition}\label{prop:Tqeps-small}
Let $q_\epsilon(x) $ satisfy the conditions in Hypotheses~{\bf (V)}. Then, for $\epsilon$ small,
 \begin{equation}
 \big\|T_{q_\epsilon}\big\|_{L^2\to L^2}=\mathcal{O}(\epsilon).
 \label{Tqepsilon-small}
 \end{equation}
 \end{proposition}
 \medskip
 
 \begin{proposition}\label{prop:TR}
$T_{R_{V_0}}(k)$ is a bounded operator from $L^2$ to $L^2$.
\end{proposition}

Propositions~\ref{prop:TQ} and~\ref{prop:Tqeps-small} are proved below.
The proof of Proposition~\ref{prop:TR} is somewhat more technical proof and is 
found in Appendix~\ref{proof-TR}.
 
 We now prove Proposition~\ref{prop:TQ}. The proof of Proposition~\ref{prop:Tqeps-small} begins on page
 \pageref{proof:Tqeps-small}.

\begin{proof}\label{proof:TQ}
 {\bf Proof of Proposition~\ref{prop:TQ}:}
We begin by introducing the notation 
\begin{equation}
 Q^\sharp = \chi^{-1}Q \chi^{-1}.
\end{equation}

Then, one uses the following calculation:
\begin{align*}
\D^{-1} Q^\sharp \D^{-1}f (x) = & \D^{-1}Q^\sharp \D^{-1}\left(\frac{1}{2\pi}\int_\xi e^{ix\xi}\hat{f}(\xi)\ d\xi \right)\\
 = & \int_\xi \frac{\hat{f}(\xi)}{2\pi} \D^{-1}Q^\sharp \D^{-1} e^{ix\xi} d\xi \\
 = & \int_\xi \frac{\hat{f}(\xi)}{4\pi^2} \D^{-1}Q^\sharp e^{ix\xi}(1+\xi^2)^{-1/2} d\xi \\
 = & \int_\xi \frac{\hat{f}(\xi)}{4\pi^2}(1+\xi^2)^{-1/2} \frac{1}{2 \pi} \int_\eta e^{ix\eta}(1+\eta^2)^{-1/2} \widehat{Q^\sharp e^{i\eta\xi}}\ d\eta\ d\xi \\
 = & \int_\xi \frac{1}{8\pi^3} \left(\int_\zeta e^{-iy\xi}f(\zeta)\ d\zeta\right) \int_\eta (\langle \xi\rangle \langle\eta\rangle)^{-1} e^{ix\eta} \widehat{Q^\sharp}(\xi-\eta) \ d\eta\ d\xi \\
 = & \int_\zeta f(\zeta) K(x,\zeta)\ d\zeta.
\end{align*}
 with the kernel
 \begin{equation}
 K(x,\zeta)\equiv \frac{1}{8\pi^3}\int_\xi \int_\eta (1+\xi^2)^{-1/2}(1+\eta^2)^{-1/2} e^{i(x\eta-\zeta\xi)}\widehat{Q^\sharp}(\xi-\eta)\ d\eta\ d\xi
 \end{equation}

We want to prove that $\iint |K(x,\zeta)|^2\ dx\ d\zeta <+\infty$, {\it i.e.} $K\in L^2(\RR^2)$. One has 
\begin{align*}
\widehat{K}(s,z)  = & \iint_{\RR^2} K(x,\zeta)e^{-ixs} e^{-i\zeta z}\ dx\ d\zeta \\
 = & \frac{1}{8\pi^3}\iint_{x,\zeta} \iint_{\eta,\xi} \langle \xi\rangle^{-1}\ \langle\eta\rangle^{-1}e^{ix(\eta-s)}e^{-i\zeta(\xi+z)}\widehat{Q^\sharp}(\xi-\eta)\ dx \ d\zeta\ d\eta\ d\xi \\
 = & \frac{1}{8\pi^3} \frac{\widehat{Q^\sharp}(-s-z)}{(1+s^2)^{1/2}(1+z^2)^{1/2}}.
\end{align*}
Therefore, we deduce 
\[\iint_{x,\zeta} |K(x,\zeta)|^2\ dx\ d\zeta = \iint_{s,z} |\widehat{K}(s,z)|^2\ ds\ dz = \frac{1}{8\pi^3} \int_s \frac{1}{1+s^2} \int_z \frac{|\widehat{Q^\sharp}(s+z)|^2 }{1+z^2}\ dz\ ds.\]

Since $Q^\sharp \in L^{2}(\RR)$, one has immediately $\iint_{x,\zeta} |K(x,\zeta)|^2\ dx\ d\zeta <\infty$, and
\[ \big\| K(x,\zeta) \big\|_{L^2(\RR^2)} \leq C \big\| Q^\sharp \big\|_{L^2(\RR)}.\]
Therefore $T_Q$ is a Hilbert-Schmidt integral operator, and is therefore bounded, with
\[ \big\| T_Q \big\|_{L^2\to L^2} \leq C \big\| Q^\sharp \big\|_{L^2(\RR)}.\]
This completes the proof of Proposition~\ref{prop:TQ}.
\end{proof}
\medskip
\begin{proof} \label{proof:Tqeps-small}
\noindent{\bf Proof of Proposition~\ref{prop:Tqeps-small}:}
\medskip

 Consider $T_Q$, where $Q=q_\epsilon(x)=q\left(x,\frac{x}\epsilon\right)$ as in~\eqref{qeps}. From the proof below,
one has $T_{q_\epsilon} f(x)=\int_\zeta f(\zeta) K_\epsilon(x,\zeta)\ d\zeta,$
with the kernel $K_\epsilon(x,\zeta)$ satisfying
\[\widehat{K_\epsilon}(s,z) = \frac{1}{8\pi^3}\frac{\widehat{q_\epsilon^\sharp}(-s-z)}{(1+s^2)^{1/2}(1+z^2)^{1/2}}.\]
Using the decomposition in Fourier series of $q_\epsilon(x)=q\left(x,\frac{x}\epsilon\right)$, one has
\[q_\epsilon^\sharp(x)\equiv\chi^{-1}q_\epsilon\chi^{-1}(x)=\sum_{|j|\geq 1} q_j^\sharp(x)e^{2i\pi j(x/\epsilon)},\] and therefore
\[\widehat{q_\epsilon^\sharp}(\xi)=\sum_{|j|\geq 1}\int_x q_j^\sharp(x)e^{2i\pi j(x/\epsilon)}e^{-ix\xi}dx=\sum_{|j|\geq 1}\widehat{q_j^\sharp}\left(\frac{2\pi j}{\epsilon}-\xi\right).\]
One deduces then
\begin{align*}
 \iint_{s,z} |\widehat{K_\epsilon}(s,z)|^2\ ds\ dz &=\frac{1}{8\pi^3} \sum_{|j|\geq 1} \iint_{s,z} \frac{|\widehat{q_j^\sharp}(s+z+(2\pi j/\epsilon))|^2}{(1+s^2)(1+z^2)}\ ds\ dz \\
&=\ \frac{1}{8\pi^3} \sum_{|j|\geq 1} \int_\RR dz \int_\RR d\eta \frac{|\widehat{q_j^\sharp}(\eta+2\pi j/\epsilon))|^2}{(1+(\eta-z)^2)(1+z^2)}\nn\\ 
 &= \sum_{|j|\geq 1}  \int_\RR dz \int_{|\eta|\geq\frac{\pi j}{\epsilon}} d\eta \frac{|\widehat{q_j^\sharp}(\eta+(2\pi j/\epsilon))|^2}{(1+(\eta-z)^2)(1+z^2)}\\
 &\ \ \ + \sum_{|j|\geq 1} \int_\RR dz \int_{|\eta|\leq\frac{\pi j}{\epsilon}} d\eta \frac{|\widehat{q_j^\sharp}(\eta+(2\pi j/\epsilon))|^2}{(1+(\eta-z)^2)(1+z^2)} \equiv\ I_1\ +\ I_2.
  \end{align*}
 {\bf Estimation of $I_1$}:\medskip 
 
 \begin{align*}
 & \int_\RR dz \int_{|\eta|\geq\frac{\pi j}{\epsilon}} \frac{|\widehat{q_j^\sharp}(\eta+\frac{2\pi j}{\epsilon})|^2}{(1+(\eta-z)^2)(1+z^2)}\ d\eta\ =\ \left(\ \int_{|z|\ge\frac{\pi j}{2\epsilon}} dz + 
 \int_{|z|\le\frac{\pi j}{2\epsilon}} dz\ \right)\ \int_{|\eta|\geq\frac{\pi j}{\epsilon}} d\eta\\
 & \ \ \ \ \ \ = \ I_{1,A} \ +\ I_{1,B}, \ \ {\rm with} \\
& I_{1,A} = \int_{|z|\ge\frac{\pi j}{2\epsilon}} dz \int_{|\eta|\geq\frac{\pi j}{\epsilon}} d\eta 
  \le C \frac{\epsilon^2}{j^2} \int_{|\eta|\geq\frac{\pi j}{\epsilon}} \left|\widehat{q_j^\sharp}\left(\eta+\frac{2\pi j}{\epsilon}\right)\right|^2 \int_{|z|\ge\frac{\pi j}{2\epsilon}} \frac{1}{1+(\eta-z)^2} dz d\eta\\
 & \ \ \ \ \ \le\ C'\ \frac{\epsilon^2}{j^2} \|q_j^\sharp\|_{L^2}^2; \ \ \text{ and as }\ |\eta-z|\ge \frac{\pi j}{2\epsilon} \ \text{ for }\ |z|\leq\frac{\pi j}{2\epsilon},\ |\eta|\ge\frac{\pi j}{\epsilon},\\
&I_{1,B} = \int_{|z|\le\frac{\pi j}{2\epsilon}} dz\ \int_{|\eta|\geq\frac{\pi j}{\epsilon} } d\eta
 \le C \frac{\epsilon^2}{j^2}\ \int_{|z|\le\frac{\pi j}{2\epsilon}} \frac{1}{1+z^2}\ dz\ 
 \int_{|\eta|\ge\frac{\pi j}{\epsilon}}\ \left|\widehat{q_j^\sharp}\left(\eta+\frac{2\pi j}{\epsilon}\right)\right|^2\ d\eta\\
 & \ \ \ \ \ \le\ C'\ \frac{\epsilon^2}{j^2} \|q_j^\sharp\|_{L^2}^2.\end{align*}

Now, summing on $j$, one obtains $I_1(\epsilon)={\mathcal O}(\epsilon^2).$

\medskip

{\bf Estimation of $I_2$:}\ We first show that if we assume only that $\sum_{|j|\ge1} \|q_j^\sharp\|^2_{L^2}<\infty$, then
 $I_2(\epsilon)\to0$ as $\epsilon\to0$, and therefore $\|T_{q_\epsilon}\|_{L^2\to L^2}^2=o(1)+{\mathcal O}(\epsilon^2)=o(1)$ as $\epsilon\to$ with no specified rate.
 
  We then show that if $q_\eps$ is as in hypotheses {\bf (V)} then $\|T_{q_\epsilon}\|_{L^2\to L^2}= {\mathcal O}(\epsilon)$ as $\epsilon\to0$.
 \medskip
 
 Assume $\sum_{|j|\ge1} \|q_j^\sharp\|^2_{L^2}<\infty$. Then,
 \begin{align*}
 I_2\ &\equiv\
 \sum_{|j|\geq 1} \int_\RR dz \int_{|\eta|\leq\frac{\pi j}{\epsilon}} d\eta \frac{|\widehat{q_j^\sharp}(\eta+(2\pi j/\epsilon))|^2}{(1+(\eta-z)^2)(1+z^2)}\nn\\
 &\le\  \int_\RR\ \frac{1}{(1+z^2)} \ dz\ \sum_{|j|\geq 1} \int_{|\eta|\leq\frac{\pi j}{\epsilon}}\ |\widehat{q_j^\sharp}(\eta+(2\pi j/\epsilon))|^2\ d\eta\ \\
 &\le\ C \sum_{|j|\geq 1} \int_{\frac{\pi j}{\epsilon}\le\tau\leq\frac{3\pi j}{\epsilon}}\ |\widehat{q_j^\sharp}(\tau)|^2\ d\tau.
  \end{align*}
  Note that
$
 \sum_{|j|\geq 1} \int_\RR |\widehat{q_j^\sharp}(\tau)|^2\ d\tau\ =\ \sum_{|j|\ge1} \|q_j^\sharp\|^2_{L^2}<\infty,
$
 implying $I_2=o(1)$ as $\epsilon\to0$. 
 \bigskip
 
 We now turn to the case where $q_\epsilon$ satisfies the condition in Hypotheses~{\bf (V)}
 in order to establish that $\|T_{q_\epsilon}\|_{L^2\to L^2}= {\mathcal O}(\epsilon)$ as $\epsilon\to0$. The estimate for $I_1(\epsilon)$ is as above:
 $
 I_1(\epsilon) \ =\ {\mathcal O}(\epsilon^2).\nn
$
 
 Now, we estimate $I_2(\epsilon)$ using the fact that since $q_j^\sharp \in L^2$ and $(q_j^\sharp)'\in L^2$:
 \begin{align*}
 \left| \widehat{q_j^\sharp}(\tau)\right| &= \left|\sum_{l=0}^{M}\int_{a_l}^{a_{l+1}}q_j^\sharp(x)e^{-i\tau x}\ dx\right| \\
  &=\left|\frac1{i\tau} \sum_{l=0}^{M}\left( \int_{a_l}^{a_{l+1}}(q_j^\sharp)'(x)e^{-i\tau x}\ dx -q_j^\sharp(a_{l+1}^-)e^{-i\tau a_{l+1}}+q_j^\sharp(a_{l}+)e^{-i\tau a_l} \right) \right| \\
  &\leq C \frac1\tau \left(\big\|(q_j^\sharp)' \big\|_{L^2} + \sum_{l=1}^{M} \left[q_j^\sharp(x)\right]_{a_j}\right) = \mathcal{O}\left(\frac1\tau\right).
 \end{align*}

 Therefore, one has
    \begin{align*}
 I_2\ &\equiv\
 \sum_{|j|\geq 1} \int_\RR dz \int_{|\eta|\leq\frac{\pi j}{\epsilon}} d\eta \frac{|\widehat{q_j^\sharp}(\eta+(2\pi j/\epsilon))|^2}{(1+(\eta-z)^2)(1+z^2)}\nn\\
 &\le\ C \int_\RR\ \frac{1}{(1+z^2)} \ dz\ \sum_{|j|\geq 1} \int_{|\eta|\leq\frac{\pi j}{\epsilon}}\ \left(\frac{\epsilon}{\pi j}\right)^2 \frac{1}{1+(\eta-z)^2} d\eta\ \le\ C' \epsilon^2 \sum_{|j|\geq 1} \frac{1}{j^2}.
  \end{align*}
 One deduces finally that
$
\big\|T_{q_\epsilon}\|_{L^2 \to L^2}=\big\|K_\epsilon\big\|_{L^2(\RR^2)}=I_1+I_2=\mathcal{O}(\epsilon)$.
 This completes the proof of Proposition~\ref{prop:Tqeps-small}.
\end{proof}

\subsection{Application to the transmission coefficient, $t(k)=t[k;Q]$}
\label{sec:Texpansion-LS}
This section is devoted to the proof of Theorem~\ref{thm:Texpansion-LS}. The heart of the matter is to view $t(k)$ as a functional of the perturbing microstructure potential, $Q(x)$ 
 \begin{equation}
 Q\ \mapsto\ t[Q]
 \label{qto-tq}
 \end{equation}
and to use the Lippman-Schwinger expansion of Theorem~\ref{thm:Expand-LS} to expand $t[Q]$ for small $|||Q|||$:
\begin{equation}
t[Q]\ =\ t_0^{hom}\ +\ t_1[Q]\ +\ t_2[Q,Q]\ +\ t_3[Q,Q,Q]\ +\ \dots,
\label{tQ-expand}\end{equation}
where $t_j[Q,Q,\dots,Q]$ is $j-$ linear in $Q$.
 The transmission coefficient expansion of Theorem~\ref{thm:Texpansion-LS} is recovered from the small $|||Q|||$ asymptotics of the first several terms of the expansion of $t[q_\epsilon]$. Finally, the error terms are estimated.
\medskip

Recall that from~\eqref{eqn:RT} the transmission coefficient, $t_W(k)$, associated with the distorted plane wave $e_{W+}(x;k)$, is given by
\begin{equation}
t_W(k)\ =\ \lim_{x\to+\infty}\ e^{-ikx}\ e_{W+}(x;k).
\nn\end{equation}
We denote the transmission coefficients of $e_{V_0+}(x;k)$ and $e_{V_0+Q,+}(x;k)$, respectively,
 \begin{align*}
 t_{V_0}(k)\ &\equiv\ t_0(k)\ \equiv\ t_0^{hom}(k),\nn\\
 t_{V}(k)\ &\equiv\ t(k)\ =\ \lim_{x\to+\infty}\ e^{-ikx}\ e_{V_0+Q,+}(x;k)\ =\ t_0^{hom}(k)\ +\ \lim_{x\to+\infty} e^{-ikx}\ u_s(x;k).
 \end{align*}
 To obtain the desired leading order expansion of $t(k)$ of Theorem~\ref{thm:Texpansion-LS} we now derive the small $|||Q|||$ asymptotics of the {\it linear and quadratic terms in $Q$} of~\eqref{QtoM}.\medskip
 
\noindent{\bf Calculation of $t_1[Q]$:}\medskip

One has from~\eqref{outgoing-resolvent} that
\begin{align*}
-R_{V_0}(k)\ Q\ e_{V_0+}(x;k) &= \ \int_{-\infty}^{x} Q(\zeta)\ e_{V_0+}(\zeta;k) e_{V_0-}(\zeta;k) \ d\zeta\ \frac{e_{V_0+}(x;k)}{2ik\ t_0^{hom}}\\ 
&\ \ \ \ +\int_{x}^{+\infty} Q(\zeta)\ e_{V_0+}(\zeta;k) e_{V_0+}(\zeta;k) \ d\zeta\ 
 \frac{e_{V_0-}(x;k)}{2ik\ t_0^{hom}}\nn\\
 &\sim\ t_1[Q]\ e^{ikx},\ \ {\rm as}\ x\to\infty,
 \end{align*}
where
\begin{equation}\label{t1Q}
t_1[Q]\ \equiv \ \frac{1}{2ik}\int_{-\infty}^{\infty} Q(\zeta)\ e_{V_0+}(\zeta;k) e_{V_0-}(\zeta;k) \ d\zeta.
\end{equation}

\noindent{\bf Calculation of $t_2[Q,Q]$:}\medskip

 One has from~\eqref{outgoing-resolvent} that
\begin{align*}
&R_{V_0}(k)\ Q\ R_{V_0}(k)\ Q(\zeta)\ e_{V_0+}(x;k) \\
& \qquad = 
\int_{-\infty}^{x} Q(\zeta) R_{V_0}(k)(Q(\zeta)\ e_{V_0+}(\zeta;k)) e_{V_0-}(\zeta;k) \ d\zeta\ 
 \frac{e_{V_0+}(x;k)}{2ik\ t_0^{hom}}\\&\qquad \qquad +\ \int_{x}^{+\infty} Q(\zeta) R_{V_0}(k)(Q(\zeta)\ e_{V_0+}(\zeta;k)) e_{V_0+}(\zeta;k)
 \ d\zeta\ \frac{e_{V_0-}(x;k)}{2ik\ t_0^{hom}}\nn\\
 &\qquad \sim\ t_2[Q,Q]\ e^{ikx}, \end{align*}
 where 
 \begin{align}\label{t2QQ}
 t_2[Q,Q]&\equiv\frac{1}{2ik}\int_{-\infty}^{\infty} Q(\zeta) R_{V_0}(k)(Q(\zeta)\ e_{V_0+}(\zeta;k)) e_{V_0-}(\zeta;k) \ d\zeta \nn \\
 &=\frac{1}{2ik}\frac1{-2ik\ t_0^{hom}}\int_{-\infty}^{+\infty} Q(\zeta) e_{V_0-}(\zeta;k) \left(\ I_l(\zeta) + I_r(\zeta) \ \right)\ d\zeta, \ \ {\rm with} \\
 \label{t2qq}
& \ \ \ \nn
\begin{array}{c}
 I_l(\zeta) \ = \ \int_{-\infty}^{\zeta} Q(z)\ e_{V_0+}(z;k)\ e_{V_0-}(z;k)\ e_{V_0+}(\zeta;k) \ dz,\\
 I_r(\zeta) \ = \ \int_{\zeta}^{+\infty} Q(z)\ e_{V_0+}(z;k)\ e_{V_0+}(z;k)\ e_{V_0-}(\zeta;k) \ dz.
\end{array}\end{align}

\noindent{\bf Estimation of the error terms:}\medskip

The final step for the proof of Theorem~\ref{thm:Texpansion-LS} consists in a bound on the contribution to the transmission coefficient from the remainder term in expansion~\eqref{tQ-expand}. This is given by the following Theorem: 
\begin{theorem}\label{prop:convergence-asymptotic-expansion-LS}
Let $K$ denote a compact subset of $\RR$, satisfying Hypothesis~\textbf{(G)}.
Introduce for $k\in K$
\begin{equation}
t_{rem}(k;Q) \ \equiv \ t(k;Q)\ -\ t^{hom}_0(k)\ -\ t_1[Q]\ -\ t_2[Q,Q].
\label{trem-def-LS}
\end{equation}
 Then we have, uniformly in $k\in K$: 
\begin{enumerate}
 \item If $V$ has compact support, then $t_{rem}(k)\ =\ \mathcal{O}(|||Q|||^3)$.
 \item If $V$ is exponentially decreasing, then $t_{rem}(k)\ =\ \mathcal{O}(|||Q|||^{3-})$.
 \item If $\langle x \rangle^{\rho+1}V_0(x)\in L^{1}(\RR)$ and $\langle x \rangle^\rho Q(x)\in L^{2}(\RR)$, $\rho>8$, then there exists $2<\beta<3$ such that $t_{rem}(k)\ =\ \mathcal{O}(|||Q|||^{\beta})$.
\end{enumerate}
\end{theorem}
\medskip

\begin{proof}
It is convenient to first introduce 
\begin{align}
\label{form1-LS} f_{rem} \ & \equiv \ -\left(I+T_{R_{V_0}}T_{Q} \right)^{-1}(T_{R_{V_0}}T_{Q})^3 \D\chi e_{V_0+}(x;k) \\
\label{form2-LS} & \equiv \ \D\chi u_s \ + \ \D\chi R_{V_0}(k)\ Q(x)\ e_{V_0+}(x;k) \nn \\
& \qquad \ - \ \D\chi R_{V_0}(k)\ Q\ R_{V_0}(k)\ Q(x)\ e_{V_0+}(x;k).
\end{align}
Using~\eqref{eqn:m-bounds}, one deduces that $\D\chi e_{V_0+}(x;k)\in L^2_x$, with
\begin{align*}\left\| \D\chi e_{V_0+} \right\|_{L^2} &= \left\| \langle\eta\rangle\widehat{\chi e_{V_0+}}(\eta;k) \right\|_{L^2_\eta} \\
 &\leq \left\| \chi(x) e_{V_0+}(x;k) \right\|_{L^2_x} \ + \ \left\| \partial_x\left(\chi(x) e_{V_0+}(x;k)\right)\right\|_{L^2_x}\ \le\ \left\| \chi(x) \langle x \rangle \right\|_{L^2_x}.
\end{align*}
Therefore, thanks to Propositions~\ref{prop:TQ} and~\ref{prop:TR}, and using~\eqref{form1-LS}, one has for $|||Q|||$ small enough,
\begin{equation}\label{eqn:finL2-LS} f_{rem}\in L^2 \ \ \ \mbox{ and }\ \ \ \|f_{rem}\|_{L^2}\ \leq\ C_\chi\ |||Q|||^3.\end{equation}
The following pointwise bound can be also be deduced 
\begin{align*}
 \left|\D^{-1} f_{rem} \right |\ &\leq \left|\int_\eta \langle\eta\rangle^{-1}\widehat{f_{rem}}(\eta) e^{i\eta x} \right |\ \leq \| \langle\eta\rangle^{-1}\|_{L^2_\eta}\ \left\|f_{rem}\right\|_{L^2},
\end{align*}
which implies
\begin{equation}\label{eqn:tr-bound-LS} \left|\chi^{-1}(x)\D^{-1} f_{rem} (x) \right| \leq C_\chi\ \chi^{-1}(x)\ |||Q|||^3 .\end{equation}

From~\eqref{form2-LS} we have that $t_{rem}(k)$ is the complex number for which 
\[ \lim_{x\to\infty}\left(\ \chi^{-1}(x)\D^{-1} f_{rem} \ - \ t_{rem}(k)e^{ikx}\ \right)\ = \ 0.\]

We now use the decay properties of the potential $V$ to estimate the magnitude of $t_{rem}(k)$ for $|||Q|||$ small.

\medskip

{\bf Case 1:\ $V$ has compact support}\\ Assume $\supp V \subset [-M,M],\ M>0$. Using the explicit representation of $R_{V_0}$,~\eqref{eqn:RV0}, for $x>M$ we have:
\begin{align*}
 -R_{V_0}(k)Q\ e_{V_0+}(x;k)\ &=\  \frac{1}{2ik\ t_0^{hom}}\ \int_{-\infty}^{x} Q(\zeta)\ e_{V_0+}(\zeta;k)\ e_{V_0-}(\zeta;k)e_{V_0+}(x;k)\ d\zeta\ \\
 &\ \ \ \ \ +\ 
 \frac{1}{2ik\ t_0^{hom}}\int_{x}^{+\infty} Q(\zeta)\ e_{V_0+}(\zeta;k)\ e_{V_0+}(\zeta;k)e_{V_0-}(x;k)\ d\zeta \\
 &= \ \frac{1}{2ik\ t_0^{hom}}\ \int_{-\infty}^{+\infty} Q(\zeta)\ e_{V_0+}(\zeta;k)\ e_{V_0-}(\zeta;k)e_{V_0+}(x;k)\ d\zeta \\
 &= \ \frac{ t_1[Q]}{t_0^{hom}} e_{V_0+}(x;k)\ =\ t_1[Q]\ e^{ikx}.
\end{align*}
Similarly, for the quadratic in $Q$-term we have
 \begin{equation}
 R_{V_0}(k)\ Q\ R_{V_0}(k)\ Q\ e_{V_0+}(x;k)\ =\ t_2[Q,Q]\ e^{ikx}.
 \nn\end{equation}
 Therefore,
 \begin{equation}
 u_s\ =\ t_1[Q]\ e^{ikx}\ +\ t_2[Q,Q]\ e^{ikx}\ +\ t_{rem}\ e^{ikx},
 \nn\end{equation}
 where for $x>M$
\[\chi^{-1}\D^{-1} f_{rem} (x) \ = \ t_{rem}(k)e^{ikx}.\]
Therefore, using the pointwise bound~\eqref{eqn:tr-bound-LS} we have
\[ |t_{rem}(k)|\ \leq \ C_\chi \chi^{-1}(M) \ |||Q|||^3\ =\ \mathcal{O}(|||Q|||^3). \]

\medskip

{\bf Case 2:\ $V$ is exponentially decreasing}\\
Assume $|V_0(x)|+|Q(x)|\leq C e^{-\alpha |x|}$ for some $C,\ \alpha>0$ and $x>M$.
As in the first case, the formula for the resolvent~\eqref{eqn:RV0} leads to 
\begin{align*}
 &-R_{V_0}(k)Q\ e_{V_0+}(x;k)\ =\ t_1[Q]\ e^{ikx} \ +\ \frac{t_1[Q]}{t_0^{hom}}
 \left( e_{V_0+}(x;k) - t_0^{hom} e^{ikx} \right)\\
 & \quad +\frac{1}{2ik\ t_0^{hom}}\int_{x}^{+\infty} Q(\zeta) e_{V_0+}(\zeta;k) \Big( e_{V_0+}(\zeta;k)e_{V_0-}(x;k) -e_{V_0+}(x;k)e_{V_0-}(\zeta;k)\Big)\ d\zeta. 
\end{align*}
Using~\eqref{eqn:m-bounds}, one can easily bound for $x>M$
\begin{align*}
 \left|\frac{1}{2ik\ t_0^{hom}}\int_{x}^{+\infty} Q(\zeta)\ e_{V_0+}(\zeta;k)\ ( e_{V_0+}(\zeta;k)e_{V_0-}(x;k)-e_{V_0+}(x;k)e_{V_0-}(\zeta;k))\ d\zeta \right| &\\
  \ \ \ \ \leq C\ \int_{x}^{+\infty} C e^{-\alpha |\zeta|} \langle \zeta \rangle \ d\zeta\ \leq \ C' e^{-\alpha/2 |x|}.&
\end{align*}
Now, we use the estimate~\eqref{eqn:m-bound1}
\[|m_+(x;k)-1|\leq \frac{1+\max(-x,0)}{1+|k|} \int_x^\infty (1+|s|)|V_0(s)| ds,\]
so that $|e_{V_0+}(x;k)-t_0^{hom}e^{ikx}|\leq Ce^{-\alpha/2 x}$ for $x>M$. Finally, one obtains 
\begin{align*}
\left| R_{V_0}(k)Q\ e_{V_0+}(x;k)\ -t_1[Q]\ e^{ikx} \right| \ \leq \ Ce^{-\alpha/2 x}.
\end{align*}

A similar estimate holds for $u_s \ = \ e_{V+}(x;k)-e_{V_0+}(x;k)$, and for the quadratic term $R_{V_0}(k)\ Q\ R_{V_0}(k)\ Q\ e_{V_0+}(x;k)$. Therefore, for $x>M$ we have
\[\left| \chi^{-1}\D^{-1} f_{rem} (x) \ - \ t_{rem}(k)e^{ikx}\right| \leq \ C e^{-\alpha/2 x}. \]

 Again the pointwise bound~\eqref{eqn:tr-bound-LS} implies, for $x>M$, that
\[ |t_{rem}(k)|\ \leq \ C_\chi \chi^{-1}(x)|||Q|||^3 \ + \ C e^{-\alpha/2 x}. \]
Finally, choosing $x=-\frac6\alpha\ln|||Q|||$, one has for $|||Q|||$ small enough,
\[|t_{rem}(k)| \leq\ C|||Q|||^3(1+\langle \ln |||Q||| \rangle ) \leq C|||Q|||^{3-}. \]

\medskip

{\bf Case 3:\ $\langle x \rangle^{\rho+1}V_0(x)\in L^{1}(\RR)$ and $\langle x \rangle^\rho Q(x)\in L^{2}(\RR)$, with $\rho>8$}\\
We use again the formula of the resolvent~\eqref{eqn:RV0}:
\begin{align*}
 &-R_{V_0}(k)Q\ e_{V_0+}(x;k)\ =\ t_1[Q]\ e^{ikx} \ +\ \frac{t_1[Q]}{t_0^{hom}}
 \left( e_{V_0+}(x;k) - t_0^{hom} e^{ikx} \right)\\
 & \quad +\frac{1}{2ik\ t_0^{hom}}\int_{x}^{+\infty} Q(\zeta) e_{V_0+}(\zeta;k) \Big( e_{V_0+}(\zeta;k)e_{V_0-}(x;k) -e_{V_0+}(x;k)e_{V_0-}(\zeta;k)\Big)\ d\zeta. 
\end{align*}
Using the estimate~\eqref{eqn:m-bound1} leads to
\[|e_{V_0+}(x;k) - t_0^{hom}(k) e^{ikx}|\leq C \int_x^\infty \frac1{(1+|s|)^{\rho}}(1+|s|)^{\rho+1}|V_0(s)|\ ds \leq C \frac1{\langle x \rangle^{\rho}} \|V_0\|_{L^{1,\rho+1}}.\]
Therefore, one has 
\begin{align*}
 \left|\frac{1}{2ik\ t_0^{hom}}\int_{x}^{+\infty} Q(\zeta)\ e_{V_0+}(\zeta;k)\ ( e_{V_0+}(\zeta;k)e_{V_0-}(x;k)-e_{V_0+}(x;k)e_{V_0-}(\zeta;k))\ d\zeta \right| &\\
 \leq C \int_{x}^{+\infty} |Q(\zeta)| \frac1{\langle \zeta \rangle^{\rho}} \left(\frac{\langle x \rangle}{\langle \zeta \rangle^{\rho}}+\frac{\langle \zeta \rangle}{\langle x \rangle^{\rho}}\right) \ d\zeta \leq \frac{C}{\langle x \rangle^{\rho}} \left\|\frac1{\langle \zeta \rangle^{\rho-1}}\right\|_{L^2_\zeta} \left\|Q(\zeta)\right\|_{L^2_\zeta},&
\end{align*}
from which we deduce 
\begin{align*}
\left| R_{V_0}(k)Q\ e_{V_0+}(x;k)\ -t_1[Q]\ e^{ikx} \right| \leq \frac{C}{\langle x \rangle^{\rho}} \left(\|V_0\|_{L^{1,\rho+1}} +\left\|\frac1{\langle \zeta \rangle^{\rho-1}}\right\|_{L^2_\zeta} \left\| Q\right\|_{L^2}\right).
\end{align*}
Similar estimates hold for $u_s \ = \ e_{V+}(x;k)-e_{V_0+}(x;k)$, and for the quadratic term $R_{V_0}(k)\ Q\ R_{V_0}(k)\ Q\ e_{V_0+}(x;k)$. Therefore, for $x>M$, one has
\[\left| \chi^{-1}\D^{-1} f_{rem} (x) \ - \ t_{rem}(k)e^{ikx}\right| \leq \ \frac{C}{\langle x \rangle^{\rho}} \|V_0\|_{L^{1,\rho+1}}. \]

Since $\chi(x)=\langle x \rangle^{-\alpha}$ with $\alpha > 4$, the pointwise bound~\eqref{eqn:tr-bound-LS} yields
\[t_{rem}(k)\ \leq \ C_\chi \langle x \rangle^{\alpha}|||Q|||^3 \ + \ C \frac1{\langle x \rangle^{\rho}}, \]
so that choosing $x=|||Q|||^{-3/(\rho+\alpha)}$, which tends to infinity as $|||Q|||$ tends to $0$, one has 
\[ |t_{rem}(k)|\ \leq\ C\ |||Q|||^{\frac{3\rho}{\rho+\alpha}}.\]
It follows that with $\alpha>4$ and $\rho>2\alpha$, one has
\[ |t_{rem}(k)|\ =\ \mathcal{O}(|||Q|||^\beta),\ \ 2<\beta\equiv\frac{3\rho}{\rho+\alpha}.\]
This completes the proof.
\end{proof}

\subsection{Completion of the proof of Theorem~\ref{thm:Texpansion}}
\label{sec:Texpansion}
In this section, we show how to derive the corrected multi-scale / homogenization expansion of section~\ref{sechomo} from the rigorous results of the previous section with a potential $V=V_0+q_\epsilon$ satisfying Hypotheses~{\bf (V)}, and using Proposition~\ref{prop:Tqeps-small}. Theorem~\ref{thm:Texpansion} follows then as a direct consequence. 

\noindent{\bf The small $\epsilon$ asymptotics of $t_1[q_\epsilon]$:}\medskip

We use the decomposition of $q_\epsilon$ in Fourier series in $y$ 
\[q_\epsilon(x)=q\left(x,\frac{x}\epsilon\right)=\sum_{|j|\geq 1} q_j(x) e^{2i\pi j(x/\epsilon)},\] that we plug into $t_1[q_\epsilon]$, given in~\eqref{t1Q}:
\[t_1[q_\epsilon]\ =\ \frac{1}{2ik} \sum_{|j|\geq 1} t_1[q_\epsilon]^j \mbox{, with } t_1[q_\epsilon]^j =\int_{-\infty}^{+\infty} q_j(\zeta)e_{V_0+}(\zeta;k) e_{V_0-}(\zeta;k) e^{2i\pi j(\zeta/\epsilon)}\ d\zeta.\]
We assume that $q_j$ is piecewise $C^3$, so that there exists $-\infty=a_0<a_1<\dots<a_M<a_{M+1}=\infty$, such that $q_j\in C^3(a_l,a_{l+1})$. Then, one has
\begin{align*}
t_1^{j,l}[q_\epsilon] &= \frac{1}{2ik}\int_{a_l}^{a_{l+1}} q_j(\zeta) e_{V_0+}(\zeta;k) e_{V_0-}(\zeta;k) e^{2i\pi j(\zeta/\epsilon)} \ d\zeta \\
&= \frac{-1}{2ik} \int_{a_l}^{a_{l+1}} \partial_\zeta(q_j(\zeta) e_{V_0+}(\zeta;k) e_{V_0-}(\zeta;k) ) \frac{\epsilon}{2i\pi j } e^{2i\pi j(\zeta/\epsilon)} \ d\zeta +b_1^{j,l} \\
&= \frac{1}{2ik}\int_{a_l}^{a_{l+1}} \partial_\zeta^2 (q_j(\zeta) e_{V_0+}(\zeta;k) e_{V_0-}(\zeta;k) ) \left(\frac{\epsilon}{2i\pi j }\right)^2 e^{2i\pi j(\zeta/\epsilon)} \ d\zeta+b_1^{j,l}+b_2^{j,l},
\end{align*}
with the following boundary terms
\[\begin{array}{rl}
 b_1^{j,l}&=\frac{-\epsilon}{4k \ \pi j }\left(q_j(a_{l+1}^-) e_{V_0+}(a_{l+1};k) e_{V_0-}(a_{l+1};k) ) e^{2i\pi j(a_{l+1}/\epsilon)}\right.\\
 & \ \ \ \ \left.- q_j(a_l^+) e_{V_0+}(a_l;k) e_{V_0-}(a_l;k) ) e^{2i\pi j(a_l/\epsilon)}\right),\\
 b_2^{j,l}&=\frac{-i\ \epsilon^2}{8k\ \pi^2 j^2 } \left(\left.\partial_\zeta (q_j(\zeta) e_{V_0+}(\zeta;k) e_{V_0-}(\zeta;k) )\right|_{\zeta={a_{l+1}^-}} e^{2i\pi j(a_{l+1}/\epsilon)}\right. \\ &\ \ \ \ \left. - \left.\partial_\zeta (q_j(\zeta) e_{V_0+}(\zeta;k) e_{V_0-}(\zeta;k) )\right|_{\zeta={a_{l}^+}} e^{2i\pi j(a_{l}/\epsilon)} \right).
\end{array}\]
Now, one has
\[\begin{array}{r}
\partial_x^2\big(q_j(x)e_{V_0+}(x;k)e_{V_0-}(x;k)\big) = \dfrac{d^2q_j}{dx^2}(x)e_{V_0+}(x;k) e_{V_0-}(x;k) \hfill \\
 \qquad\qquad+\ 2 \dfrac{dq_j}{dx}(x)\partial_x(e_{V_0+}(x;k)e_{V_0-}(x;k) )\ +\ 2 q_j(x)\partial_x e_{V_0+}(x;k)\partial_x e_{V_0-}(x;k) \\
 +\ q_j(x)\big((\partial_x^2 e_{V_0+}(x;k)) e_{V_0-}(x;k)\ +\ e_{V_0+}(x;k)\partial_x^2 e_{V_0-}(x;k)\big).
\end{array}\]
The first three terms are piecewise-$C^1$, so that oscillatory integrals predict that
\begin{align} \int_{a_l}^{a_{l+1}} \Big(\frac{d^2q_j}{d\zeta^2}(\zeta)e_{V_0+}(\zeta;k) e_{V_0-}(\zeta;k) +\frac{dq_j}{d\zeta}(\zeta)\partial_\zeta(e_{V_0+}(\zeta;k) e_{V_0-}(\zeta;k)) & \nn \\
 \ \ \ +2q_j(\zeta)\partial_\zeta e_{V_0+}(\zeta;k)\partial_\zeta e_{V_0-}(\zeta;k)\Big) e^{2i\pi j(\zeta/\epsilon)} \ d\zeta & = \mathcal{O}(\epsilon).
\label{first3}\end{align}
For the fourth term, we use the fact that $e_{V_0+}$ and $e_{V_0-}$ satisfy $\big(-\frac{d^2}{d x^2}+V_0-k^2\big)u=0$, so that one has, with $\Omega_j \ = \ \{x_0,\dots,x_{N-1}\}\cap (a_j,a_{j+1})$,
\begin{align*}
t_1^{j,l}[q_\epsilon] 
& =\frac{i\ \epsilon^2}{8k \pi^2j^2} \sum_{x_i\in\Omega_j} 2 c_i q_j(x_i)e_{V_0+}(x_i;k)e_{V_0-}(x_i;k)e^{\frac{2i\pi j\ x_i}\epsilon} +b_1^{j,l}+b_2^{j,l}+\mathcal{O}\left(\epsilon^3/j^2\right) \\
& = \frac{i\ \epsilon^2}{8k\ \pi^2j^2} \sum_{x_i\in\Omega_j} \left[\partial_x (q_j(x) e_{V_0+}(x;k) e_{V_0-}(x;k) )\right]_{a_j} e^{\frac{2i\pi j\ x_i}\epsilon} +b_1^{j,l}+\mathcal{O}\left(\epsilon^3/j^2\right).
\end{align*}
Finally, we have $t_1[q_\epsilon]\ =\ \sum_{l=0}^{M-1}\ \sum_{|j|\geq 1}\ t_1^{j,l}[q_\epsilon]\ +\ \mathcal{O}(\epsilon^3)$, and
one recovers immediately terms of the expansion of Theorem~\ref{thm:Texpansion}:
\[ \begin{array}{l}
\displaystyle   \sum_{l=0}^{M-1} \sum_{|j|\geq 1} b_1^{j,l} = \epsilon t_1^\epsilon \ \ \ {\rm and}\\
\displaystyle \sum_{l=0}^{M-1} \sum_{|j|\geq 1} \frac{i\ \epsilon^2}{8k\ \pi^2j^2} \sum_{x_i\in\Omega_j} \Big[\partial_x (q_j(x) e_{V_0+}(x;k) e_{V_0-}(x;k) )\Big]_{a_j} e^{2i\pi j(x_i/\epsilon)}= \epsilon^2 t_2^\epsilon +\mathcal{O}(\epsilon^3), \end{array}
\]
so that\ $ t_1[q_\epsilon] = \epsilon t_1^\epsilon(k) +\epsilon^2 t_2^\epsilon(k) + \mathcal{O}(\epsilon^3)$.
\medskip

\noindent{\bf The small $\epsilon$ asymptotics of $t_2[q_\epsilon,q_\epsilon]$:}\medskip

Let us assume that $\zeta$ is fixed outside $\supp V_{sing}$, and outside the discontinuities of $q_j,\partial_xq_j$ (this particular case arises for a finite number of values of $\zeta$, and therefore brings no contribution to the transmission coefficient, when integrated).
Then integrating by part leads to the following expansion for $\epsilon$ small:
\begin{align*}
I_l^j(\zeta) & \ \equiv\ -e_{V_0+}(\zeta;k)\ \int_{-\infty}^{\zeta} \partial_z \big(q_j(z) e_{V_0+}(z;k)e_{V_0-}(z;k)\big)\ \frac{\epsilon}{2i\pi j}e^{2i\pi j(z/\epsilon)}\ dz\\
& \ =\ \frac{\epsilon ^2}{4 \pi^2j^2} \ e_{V_0+}(\zeta;k)\ \left( \ -\int_{-\infty}^{\zeta} \partial_z^2 \big(q_j(z) e_{V_0+}(z;k)e_{V_0-}(z;k)\big)\ e^{2i\pi j(z/\epsilon)} \ dz \right.\\
& \ \ \ \ \left. + \ \left[\partial_z \big(q_j(z) e_{V_0+}(z;k)e_{V_0-}(z;k)\big)\ e^{2i\pi j(./\epsilon)}\right]_{-\infty}^{\zeta} \right).
 \end{align*}
The first term, treated as previously and using the fact that the functions $q_j$, $e_{V_0+}$, and $e_{V_0-}$ are piecewise-$C^3$, brings a contribution of order $\mathcal{O}(\epsilon^3)$.

Now, using the same analysis on $I_r^j(\zeta)$ and the Wronskian identity~\eqref{wronskian}, one obtains the following expansion for the $-d\zeta$ integrand of~\eqref{t2QQ}:
\begin{align*}
 I_l^j(\zeta) + I_r^j(\zeta) &\ = \ \Big( e_{V_0+}(\zeta;k)\partial_\zeta (q_j(\zeta) e_{V_0+}(\zeta;k)e_{V_0-}(\zeta;k)) \\
 & \ \ \ \ \ - \ e_{V_0-}(\zeta;k)\partial_\zeta (q_j(\zeta) e_{V_0+}(\zeta;k)e_{V_0-}(\zeta;k)) \Big)\ +\ \mathcal{O}(\epsilon^3)\\
 & \ = \ -2ik \ t_0^{hom}\ q_j(\zeta)\ e_{V_0+}(\zeta;k) \ e^{2i\pi j(\zeta/\epsilon)}\ +\ \mathcal{O}(\epsilon^3).
\end{align*}

Therefore, one has
\begin{align}
t_2[q_\epsilon ,q_\epsilon ]\ & = \frac{1}{2ik}\frac1{-2ik\ t_0^{hom}} \sum_{|j|\ge1}\frac{\epsilon ^2}{4 \pi^2j^2} \int_{-\infty}^{+\infty} \sum_{|m|\ge1} q_m(\zeta) e^{2i\pi m(\zeta/\epsilon)}e_{V_0-}(\zeta;k) \nn \\
& \ \ \ \ \left( -2ik \ t_0^{hom} q_j(\zeta)e_{V_0+}(\zeta;k) e^{2i\pi j(\zeta/\epsilon)} \right) \ d\zeta\nn \ +\ \mathcal{O}(\epsilon^3) \\
&=\epsilon^2\frac{i }{8k\pi^2} \int_{-\infty}^{+\infty} \sum_{|j|\ge1} \frac{q_{-j}(\zeta)q_{j}(\zeta)}{j^2} e_{V_0-}(\zeta;k) e_{V_0+}(\zeta;k) \ d\zeta \ +\ \mathcal{O}(\epsilon^3).
\end{align}

One recovers finally:\ $t_2[q_\epsilon,q_\epsilon]  = \epsilon^2 t_2^{hom} +\mathcal{O}(\epsilon ^3).$

\medskip

\noindent{\bf Estimate of $t_{rem}^\epsilon$:}\medskip

Using Proposition~\ref{prop:Tqeps-small} with Theorem~\ref{prop:convergence-asymptotic-expansion-LS} yields:
\begin{proposition}\label{prop:convergence-asymptotic-expansion}
Let $K$ denote a compact subset of $\RR$, satisfying Hypothesis~\textbf{(G)}. Introduce for $k\in K$
\begin{equation}
t_{rem}^\epsilon(k) \ \equiv \ t^\epsilon(k) - t^{hom}_0(k)-\epsilon t_1^\epsilon(k)-\epsilon^2\left(\ t^{hom}_2(k)+t_2^\epsilon(k)\ \right).
\label{trem-def}
\end{equation}
 Then we have 
\begin{enumerate}
 \item If $V$ has compact support, then $t_{rem}^\epsilon(k)\ =\ \mathcal{O}(\epsilon^3)$.
 \item If $V$ is exponentially decreasing, then $t_{rem}^\epsilon(k)\ =\ \mathcal{O}(\epsilon^{3-})$.
 \item If $\langle x \rangle^\rho V_0\in L^{1}$, $\rho>9$, then there exists $2<\beta<3$ such that $t_{rem}^\epsilon(k)\ =\ \mathcal{O}(\epsilon^{\beta})$.
\end{enumerate}
\end{proposition}
The proof of Theorem~\ref{thm:Texpansion} is now complete.

\appendix
\section{The numerical computations}
In this section we outline the numerical method we used to obtain results displayed in figures~\ref{figT-T0} and~\ref{figSmoothDirac}. 

We approach the computation of $t(k)$, the transmission coefficient associated with the potential $V(x)$, by 
 numerical approximation of the function 
\[u(x;k)\ \equiv \ \frac1{t(k)}e_{V-}(x;k),\]
where $e_{V-}(x;k)$ denotes the distorted plane wave generate by an incoming wave from positive infinity; see~\eqref{eqn:RT}.
We rewrite the equation
\[
\left(-\frac{d^2}{d x^2}+V(x)-k^2\right)u(x;k)\ =\ 0,
\]
equivalently in terms of the variable $U(x;k) \equiv (u(x;k), \partial_x u(x;k))^T$ as the first order system
\begin{equation}\label{schrod-sys}
\frac{d}{dx}U \ =\ \begin{pmatrix} 0 & 1 \\ V(x)-k^2 & 0 \end{pmatrix} U.
\end{equation}
Note that if $V$ is assumed to have compact support ($\supp V \subset [-M,M]$ with $M>0$), then 
\begin{align}
 \label{U-left} U(x;k) \ &\equiv \ \begin{pmatrix} e^{-ikx} \\ -ike^{-ikx} \end{pmatrix}\ \ \ \mbox{for } x<-M, \\
 \label{U-right} U(x;k) \ &\equiv \ \begin{pmatrix} \frac{r_r(k)}{t(k)}e^{ikx}+\frac1{t(k)}e^{-ikx} \\ \frac{ik\ r_r(k)}{t(k)}e^{ikx}-\frac{ik}{t(k)}e^{-ikx} \end{pmatrix}\ \ \ \mbox{for } x>M.
\end{align}
Starting with the initial data given by~\eqref{U-left}, we numerically solve the system of first order ODEs defined by~\eqref{schrod-sys} up to $x>M$, and~\eqref{U-right} allows to recover the desired value of $t(k)$.
At the location of the singularities $x=x_j$, the jump conditions~\eqref{JumpCondition} allow to obtain $U(x+;k)$ from $U(x-;k)$ via a transfer matrix. Between the singularities, one approximatively solves~\eqref{schrod-sys} using for example Runge-Kutta formulae. We used the Matlab solver \verb+ode45+; see~\cite{SRmatlab:97} for more information about the Matlab ODE Suite. 

\medskip 

We conclude this section by stating the precise functions and parameters used to obtain the plots displayed in figures~\ref{figT-T0} and~\ref{figSmoothDirac}.

\label{sec:numerics}
For the case when $V_0$ has singularities, as in the left and center panels of figure~\ref{figT-T0}, we set 
\[V_0=V_{sing}(x)\equiv 40 \left(\delta(x) + \delta(x-0.5) + \delta(x-1) \right).\]
Otherwise, we set 
\[V_0=V_{reg}(x)\equiv 40 \left(\delta_\rho(x) + \delta_\rho(x-0.5)+ \delta_\rho(x-1) \right),\]
with $\delta_\rho(x)\equiv \frac{1}{\rho\sqrt{\pi}} \mathrm{e}^{-x^2/\rho^2}$ the smoothed out approximation. One has $\rho=0.1$ for the right panels of figures~\ref{figT-T0} and~\ref{figSmoothDirac}, and respectively $\rho=0.01$ and $\rho=0.001$ for the center and left panels of figure~\ref{figSmoothDirac}.

 We set $q_\epsilon(x)\ =\ f(x) \sin(2\pi x/\epsilon)$, with $f(x)\equiv 0$ for $x\in\RR\setminus[-2/3;2/3]$, and elsewhere 
\[\left\{\begin{array}{l}
     f(x)=40\ \ \ \mbox{in the discontinuous case (left panel of figure~\ref{figT-T0}), or} \\
	 f(x)=40 e^{-\frac{x^2}{(x-2/3)(x+2/3)}}\ \ \ \mbox{in the smooth cases (all other panels).}
     \end{array}\right.\]
Finally, we set $k=5.5$, since it corresponds to a case where $t_0^{hom}(k)$ approaches unity when $V_0=V_{sing}$. 

\section{The Jost solutions}\label{sec:Jost}
In this section, we provide a construction of the Jost solutions and a rigorous derivation of their properties, including bounds that are used in the proof of Proposition~\ref{prop:TR}, Appendix~\ref{proof-TR}.
We recall that by Definition~\ref{def:JostSolutions}, the Jost solutions are the unique solutions $f_\pm(x;k)$ of
 \begin{equation}
 \left(\ H_W\ -\ k^2\ \right)u\ \equiv\ \left(-\frac{d^2}{d x^2}+W(x)-k^2\right) u =0.
 \label{schrod-eqn-Jost}
 \end{equation}
such that $f_\pm(x;k)=e^{\pm ikx}\ m_\pm(x;k)$ and 
 \[ \lim_{x\to \pm\infty} m_\pm(x;k) \ = \ 1.\]
 The existence of Jost solutions for regular potentials $W\in L^{1,3/2+}(\RR)$ is established in~\cite{DT:79}. The generalization to potentials allowing a singular component
 \begin{align*}
 W&=W_{reg}+W_{sing}, \ \ {\rm with}\\
 & \ W_{reg}\ \in\ L^{1,3/2+}(\RR), \\
 & \ W_{sing}\ =\ \sum_{j=0}^{N-1} c_j \ \delta (x-x_j),\ \ {\rm where}\ \ 
 c_j, x_j \in \RR,\ \ x_j < x_{j+1}.
\end{align*}
 can be found in~\cite{DMW:10}. 
 
 \medskip
 
As an intermediate step of the proof, one introduces an equivalent definition of the Jost solution, as solutions of integral equations.
In the case where $W$ is regular, one has
\begin{align}\label{eqn:integralequation}
m_+ (x;k) & = 1 + \int_x^\infty D_k (\zeta-x) W(\zeta) m_+ (\zeta;k) d\zeta, \\
m_- (x;k) & = 1 + \int_{-\infty}^x D_k (x-\zeta) W(\zeta) m_- (\zeta;k) d\zeta, \qquad 
D_k (x) = \int_0^x e^{2ik\zeta} d\zeta. \nn
\end{align}
If $W$ has regular and singular components, we work with a variant of equations~\eqref{eqn:integralequation}:
\begin{align*}
m_+ (x;k) & = 1 + \int_x^\infty D_k (\zeta-x) W(\zeta) m_+ (\zeta;k) d\zeta \ + \ \sum_{x_j>x} D_k (x_j-x)c_j m_+ (x_j;k), \\
m_- (x;k) & = 1 + \int_{-\infty}^x D_k (x-\zeta) W(\zeta) m_- (\zeta;k) d\zeta \ + \ \sum_{x_j<x} D_k (x_j-x)c_j m_+ (x_j;k).
\end{align*}
 From these integral equations, one deduces 
 \begin{align}\label{eqn:m-bound1}
|m_+(x;k)-1|&\leq \frac{1+\max(-x,0)}{1+|k|} \int_x^\infty (1+|s|)|W(s)| ds,\nn \\
|m_-(x;k)-1|&\leq \frac{1+\max(-x,0)}{1+|k|} \int_{-\infty}^{-x} (1+|s|)|W(s)| ds. 
\end{align}
Then, since $m_+$ satisfies 
\begin{align*}
 \partial_x m_+(x;k)&=\int_x^\infty e^{2ik(t-x)}W(t)m_+(t;k), \mbox{ and}\\
 \partial_k m_+(x;k)&=\int_x^\infty D_k(t-x) W(t)\partial_k m_+(t;k)\ + \ \int_x^\infty \partial_k D_k(t-x) W(t) m_+(t;k),
\end{align*}
one obtains easily the following uniform bounds
\begin{align}\label{eqn:m-bounds}
&|m_+(x;k)|\leq C \langle x \rangle,\quad &|\partial_x m_+(x;k)|\leq C, \nn \\
&|\partial_k m_+(x;k)|\leq C \langle x \rangle^2, \quad &|\partial_x\partial_k m_+(x;k)|\leq C \langle x \rangle, 
\end{align}
where $C$ is independent of $k$. The same bounds clearly hold for $m_-(x;k)$.

\section{Proof of Proposition~\ref{prop:TR}}
\label{proof-TR}
This Section is dedicated to the proof of Proposition~\ref{prop:TR}, namely
\[
T_{R_{V_0}}(k)\ \equiv\ \D\chi R_{V_0}(k)\chi\D\mbox{ is a bounded operator from }L^2\mbox{ to }L^2.
\]

This result has been proved by in~\cite{GW:05}, for $V_0\equiv 0$, and spatial dimensions $n=1,2,3$. We generalize this result in the one dimensional case for $V_0=V_{reg}+V_{sing}$ as in~\eqref{Vdecomp}, so that singularities in the potential are allowed.

Our proof requires the use of the generalized Fourier transform, described in terms of the distorted plane waves. We introduce
\[
\Psi (x;\zeta) \ = \ \frac{1}{\sqrt{2 \pi}} \left\{ \begin{array}{cc}
e_{V_0+} (x;\zeta) & \zeta \geq 0, \\
e_{V_0-} (x;-\zeta) & \zeta < 0 ,
\end{array} \right. \equiv \ \frac{1}{\sqrt{2 \pi}} \left\{ \begin{array}{cc}
t(\zeta) m_+(x;\zeta)e^{ix\zeta} & \zeta \geq 0, \\
t(-\zeta) m_-(x;-\zeta) e^{ix\zeta} & \zeta < 0 ,
\end{array} \right. 
\]
with $m_+(x;\zeta)\to 0$ as $x\to\infty$ and $m_-(x;\zeta)\to0$ as $x\to-\infty$. 

Then $\mathcal{F}$ and $\mathcal{F}^*$ the distorted Fourier transform and its adjoint are defined by
\begin{align*}
\begin{array}{rl}
 \mathcal{F} & : \begin{array}{ll}
L^2 &\to L^2 \\
\displaystyle \phi &\mapsto \mathcal{F}[\phi](\xi) \equiv \int_{-\infty}^{+\infty} \phi (x) \overline{\Psi (x,\xi)}\ dx,\end{array} \\
 \mathcal{F}^* & : \begin{array}{ll}
L^2 &\to L^2 \\
\displaystyle \Phi &\mapsto \int_{-\infty}^{+\infty} \Phi(\xi) \Psi (x,\xi)\ d\xi.\end{array}
\end{array}\end{align*}

One has the following property:
\begin{align*}
P_c \phi = \mathcal{F}^* \mathcal{F} \phi,
\end{align*}
where $P_c$ denotes the spectral projection onto the continuous spectral subspace associated with the operator 
\begin{equation}
H\equiv-\partial_x^2+V_0.
\label{Hdef}
\end{equation}

To construct a smoothing operator which commutes with functions of $H$, it is convenient to introduce,
 using the distorted plane wave spectral representation of $H$:
\begin{equation}
\DV\ f\ =\ ( I-\Delta+V_0)^{1/2}\ f\ =\ \int_\RR \langle\eta\rangle\mathcal{F}[f](\eta)\ \Psi(x;\eta)\ d\eta
\label{DVf}
\end{equation}

Therefore, one has 
\begin{align}
T_{R_{V_0}}= &\D \DV^{-1} \DV\chi R_{V_0}(k)\chi\DV\DV^{-1}\D \\
\equiv & \D \DV^{-1}\circ \tilde{T}_{R_{V_0}}\circ \DV^{-1}\D. \label{decomp}
\end{align}

There are thus three terms to estimate. 
In order to deal with $\D \DV^{-1}$ and $\DV^{-1}\D$, we introduce the classical wave operator, $W$ and its adjoint $W^*$, defined by
\begin{align}
\label{eqn:wpm}
W &\equiv s - \lim_{t \to \infty} e^{it H} e^{-it H_0},\\
W^* &\equiv s - \lim_{t \to \infty} e^{i t H_0} e^{-it H} P_c,
\label{eqn:wpmstar}
\end{align}
with $H\equiv-\partial_x^2+V_0$ and $H_0\equiv-\partial_x^2$.
The wave operators have the property to intertwine between the continuous part of $H$ and $H_0$, so that for any Borel function $f$:
\[f(H)P_c=W f(H_0) W^\ast.\]
Especially, one has $\DV=W\D W^\ast$, so that 
\begin{equation}
\label{DDV}
\D \DV^{-1}=\D W \D^{-1} W^\ast.
\end{equation}

Let us state the following result, that has been introduced in~\cite{Weder} and extended in~\cite{DMW:10} to potentials $V_0=V_{reg}+V_{sing}$ as in~\eqref{Vdecomp}, thus allowing Dirac delta functions:
\begin{lemma}\label{lem:Wbounded}
$W$ and $W^\ast$ have extensions to bounded operators on $H^k$, for $k=-1,0,1$.
\end{lemma}
Using this last result and the known fact that $\D^s$ is bounded from $H^k$ to $H^{k-s}$, we obtain directly from~\eqref{DDV} that
\[\D \DV^{-1} \mbox{ is bounded from }L^2 \mbox{ to }L^2.\]
Similarly, 
\[\DV^{-1} \D \mbox{ is bounded from }L^2 \mbox{ to }L^2.\]
\medskip

In order to deal with the last term of~\eqref{decomp}, we decompose $\tilde{T}_{R_{V_0}}$ as a sum of four operators, commuting $\D$.
\begin{align*}
\tilde{T}_{R_{V_0}} \equiv & \DV\chi R_{V_0}(k)\chi\DV \\
= & \left(\chi\DV +[\DV,\chi]\right)R_{V_0}(k)\left(\DV\chi+[\chi,\DV]\right)\\
= & \chi\DV R_{V_0}(k) \DV\chi + ([\DV,\chi])(R_{V_0}(k)\DV\chi)\\&+(\chi\DV R_{V_0}(k))([\chi,\DV]) + ([\DV,\chi])( R_{V_0}(k))([\chi,\DV]) \\
= & A_I + A_{II}^{(a)}+A_{II}^{(b)}+A_{III}.
\end{align*}
Each of these terms is proved to be bounded from $L^2$ to $L^2$.
We treat each term separately in Propositions~\ref{prop:A1},~\ref{prop:A2a},~\ref{prop:A2b}, and~\ref{prop:A3}.

\bigskip

\begin{proposition}\label{prop:A1} $A_{I}\ \equiv\ \chi\DV R_{V_0}(k) \DV\chi$ is bounded $L^2 \to L^2$, {\it i.e.}
\begin{equation}
\left\|\ \chi\DV R_{V_0}(k) \DV\chi\ g\ \right\|_{L^2}\ \le\ C\ \|g\|_{L^2},\ \ g\in L^2(\RR)
\label{toprove1}\end{equation}
\end{proposition}
First we commute $\DV$ and $R_{V_0}$. It is obvious that $R_0$ and $\D$ commute, so that using the wave operators introduced above (so that $\DV=W\D W^\ast$ and $R_{V_0}(k)=W R_0(k) W^\ast$, with $W$ unitary).
\begin{align*}
 A_I \ &=\ \chi\DV R_{V_0}(k) \DV\chi \\ &=\ \chi W\D W^\ast W R_0(k) W^\ast W\D W^\ast \chi \\ &=\ \chi W R_0(k) \D^2 W^\ast \chi \\ &=\ \chi R_{V_0}(k) \DV^2 \chi.
\end{align*}
Then, applying the identity $\DV^2= I-\Delta+V_0$, one obtains
\[A_I\ = \ (1+k^2)\chi R_{V_0}(k) \chi\ +\ \chi^2 .\]
Finally, using~\eqref{outgoing-resolvent} together with~\eqref{eqn:m-bounds}, one has the pointwise bound
\[ \left| R_{V_0}(x,y;k) \right|\ \leq \ C\langle x \rangle \langle y \rangle\] 
with $C$ uniform in $k$. It follows that for $f\in L^2$,
\[\big\vert \chi R_{V_0}(k) \chi f \big\vert_{L^2}\ =\ \left\vert \chi(x) \int_\zeta R_{V_0}(x,\zeta;k)\chi(\zeta) f (\zeta) \ d\zeta \right\vert_{L^2_x}\ \leq \ C\left\vert \chi(x) \langle x \rangle\right\vert_{L^2_x}^2\big\vert f \big\vert_{L^2},\]
so that $A_I$ is bounded from $L^2$ to $L^2$ , with
\begin{equation}\label{bound:A1}
\big\| A_I \big\|_{L^2\to L^2} \leq C\left(\left\vert \chi(x) \langle x \rangle\right\vert_{L^2_x}^2+\big\vert \chi \big\vert_{L^\infty}\right).
\end{equation}

\bigskip

Before carrying on with estimating the term $A_{II}^{(a)}$, let us state the following Lemma.
\begin{lemma}\label{lem:Kbound}
Let $K$ be defined for $(\xi,\eta)\in \RR\times\RR$ by
\begin{equation}
K(\xi,\eta)\ \equiv\ \left(\langle \xi \rangle - \langle \eta \rangle \right)\
 \int_\zeta \overline{\Psi(\zeta;\xi)}\ \Psi(\zeta;\eta)\ \chi(\zeta)\ d\zeta\
 \label{Kdef}
\end{equation}
Then $K(\xi,\eta)$ satisfies the following upper bounds:
\begin{align}
\left|\ K(\xi,\eta)\ \right|\ & \leq \frac{C_\chi}{ 1+ |\xi-\eta|},\label{chi-bound} \\
\left|\ \partial_\eta K(\xi,\eta)\ \right|\ +\ \left|\ \partial_\xi K(\xi,\eta)\ \right|\ & \leq \frac{C_\chi'}{ 1+ |\xi-\eta|},\label{dchi-bound}
\end{align}
with the $C_\chi$ and $C_\chi'$ constants depending on the function $\chi$ with
\begin{align*}
C_\chi &\equiv C\ \left(\sum_{j=0}^2\ \left\|\langle \zeta\rangle^j\ \partial^j_\zeta\chi\right\|_{L^1_\zeta}+\left\|\langle \zeta\rangle^2\ \chi\right\|_{L^1_\zeta}+\left\|\langle \zeta\rangle^2\ \chi\right\|_{L^\infty_\zeta}\right) \\
C'_\chi &\equiv C\ \left(\sum_{j=0}^2\ \left\|\langle \zeta\rangle^{j+1}\ \partial^j_\zeta\chi\right\|_{L^1_\zeta}+\left\|\langle \zeta\rangle^3\ \chi\right\|_{L^1_\zeta}+\left\|\langle \zeta\rangle^3\ \chi\right\|_{L^\infty_\zeta}\right) \\
\end{align*}
\end{lemma}


\begin{proof}
We consider the case where $\xi\ge0$ and $\eta\ge0$. The other cases follow similarly. Therefore, one has 
\[ K(\xi,\eta)\ = \left(\langle \xi \rangle - \langle \eta \rangle \right)\ I(\xi,\eta), \ \ {\rm with}\]
\begin{equation}
I(\xi,\eta) \equiv \int_\zeta \overline{\Psi(\zeta;\xi)} \Psi(\zeta;\eta) \chi(\zeta)\ d\zeta = \int_\zeta e^{i\zeta(\eta-\xi)}\overline{t(\xi)m_+(\zeta;\xi)}t(\eta) m_+(\zeta;\eta) \chi(\zeta)\ d\zeta.
\label{I-def} 
\end{equation}

Throughout the proof, we will use extensively the uniform bounds on $m_+$ given in~\eqref{eqn:m-bounds}.

First, by the uniform boundedness of $t(\xi)$ and $\langle \zeta\rangle^{-1}\ m_+(\zeta,\xi)$ in $\zeta$ and $\xi$, one has
\begin{equation}
\left| I(\xi,\eta) \right| \le  |t(\xi)| |t(\eta)|\  \int_\zeta\ | \langle \zeta\rangle^{-1}m_+(\zeta;\eta)\ \langle \zeta\rangle^{-1}m_+(\zeta;\xi)|\ \langle \zeta\rangle^{2}|\chi(\zeta)|\ d\zeta\ \le C\ \|\langle \zeta\rangle^{2}\chi\|_{L^1_\zeta}.
\label{eta-xi-small}\end{equation}

For $|\eta-\xi|\ge1$ we write 
\begin{align}
I(\xi,\eta) &\equiv \frac{1}{(i(\eta-\xi))^2} \overline{t(\xi)}t(\eta)\ \int_\zeta \left(\ \frac{d^2}{d\zeta^2}\ e^{i\zeta(\eta-\xi)}\ \right)\ \overline{m_+(\zeta;\eta)}\ m_+(\zeta;\xi)\ \chi(\zeta)\ d\zeta
\label{I1}\\
& = \frac{1}{(i(\eta-\xi))^2}\ \overline{t(\xi)}t(\eta) \int_\zeta e^{i\zeta(\eta-\xi)}\ \frac{d^2}{d\zeta^2}\ \left(\ \overline{m_+(\zeta;\eta)}\ m_+(\zeta;\xi)\ \chi(\zeta)\right)\ d\zeta.\label{I2}
\end{align}
The most singular terms in the integrand of~\eqref{I2} are those containing $\partial_\zeta^2 m_+$. In particular, recall the relation $\partial_x^2 m_+=-2ik\partial_x m_++V_0 m_+$, where $V_0$ contains Dirac mass singularities. 
Thus, for $|\xi-\eta|\ge1$ we have
\begin{equation}
\left| I(\xi,\eta) \right|\ \le\ C\ \left(\sum_{j=0}^2\ \left\|\langle \zeta\rangle^j\ \partial^j_\zeta\chi\right\|_{L^1_\zeta}+\left\|\langle \zeta\rangle\ \chi\right\|_{L^1_\zeta}+\left\|\langle \zeta\rangle^2\ \chi\right\|_{L^\infty}\right)\ \cdot\ \frac{1}{ |\xi-\eta|^2}\label{I3}\end{equation}

Applying~\eqref{eta-xi-small} for $|\eta-\xi|\le1$ and~\eqref{I3} for $|\eta-\xi|\ge1$ yields
\begin{equation}
\left| I(\xi,\eta) \right|\ \le\ C_\chi \frac{1}{ 1+ |\xi-\eta|^2}.\nn\end{equation}
Finally, since $\left|K(\xi,\eta)\right| \ = \ \left|I(\xi,\eta)\right| \left|\langle\xi\rangle-\langle\eta\rangle\right|\ \leq\ C\left|I(\xi,\eta)\right|  \left|\xi-\eta\right|$, 
multiplication by $|\xi-\eta|$ implies~\eqref{chi-bound}. 

\medskip

Using the same method as previously, one obtains similarly
\begin{align*}
\left| \partial_\eta I(\xi,\eta) \right|\ \le\ C_\chi' \frac{1}{ 1+ |\xi-\eta|^2}. \end{align*}
Finally, one has $\left|\partial_\eta K(\xi,\eta)\right| \ \leq \ \left|\partial_\eta I(\xi,\eta)\right| \left|\langle\xi\rangle-\langle\eta\rangle\right|\ + \left|I(\xi,\eta)\right|$, 
so that we deduce the first part of~\eqref{dchi-bound}. By symmetry, one obtains the same estimate for $\partial_\xi K(\xi,\eta)$, which concludes the proof of Lemma~\ref{lem:Kbound}.
\end{proof}

\bigskip

\begin{proposition}\label{prop:A2a} $A_{II}^{(a)}\equiv[\DV,\chi]R_{V_0}\DV \chi$ is bounded $L^2 \to L^2$, {\it i.e.}
\begin{equation}
\left\|\ [\DV,\chi]R_{V_0}\DV \chi\ g\ \right\|_{L^2}\ \le\ C\ \|g\|_{L^2},\ \ g\in L^2(\RR).
\label{toprove2a}\end{equation}
\end{proposition}
\begin{proof}
Our strategy is as follows. We view the operator $A_{II}^{(a)}$ as a composition of two operators
\begin{equation}
 A_{II}^{(a)}\ =\ [\DV ,\chi]\ \circ\ R_{V_0}\DV \chi
 \nn\end{equation}
 and first find a representation of each operator with respect to the distorted Fourier basis. We then directly prove the boundedness of $A_{II}^{(a)}:L^2\mapsto L^2$ using this spectral representation and an appropriate frequency localization argument.\bigskip
 
In terms of the distorted Fourier transform, one has 
\begin{align}
 [\DV ,\chi]f(x) &= [\DV ,\chi]\left(\int_\eta \Psi(x;\eta)\mathcal{F}[f](\eta)\ d\eta \right)\nn\\
 &= \int_\eta \mathcal{F}[ f](\eta) \Big( \DV (\chi(x) \Psi(x;\eta))-\chi\DV \Psi(x;\eta) \Big)\ d\eta, \label{comm-rep}
 \end{align}
 Now, since $\DV \Psi(x;\eta)= \langle \eta \rangle \Psi(x;\eta)$, one has
 \begin{align*} \DV (\chi(x) \Psi(x;\eta)) &= \int_\xi \Psi(x;\xi)\int_\zeta \DV (\chi \Psi(\cdot;\eta))(\zeta)\overline{\Psi(\zeta;\xi)}\ d\zeta\ d\xi \\
 &= \int_\xi \Psi(x;\xi)\int_\zeta \chi(\zeta) \Psi(\zeta;\eta)\DV \overline{\Psi(\cdot;\xi)}(\zeta)\ d\zeta\ d\xi \\
 &= \int_\xi \Psi(x;\xi)\int_\zeta \chi(\zeta) \Psi(\zeta;\eta) \langle \xi \rangle \overline{\Psi(\zeta;\xi)}\ d\zeta\ d\xi.
 \end{align*}
 Therefore, we finally deduce
 \begin{align}
[\DV ,\chi] f(x) &= \int_\eta \mathcal{F}[f](\eta) \left(\int_\xi \Psi(x;\xi)\int_\zeta \Psi(\zeta;\eta)\overline{\Psi(\zeta;\xi)} \chi(\zeta) \left(\langle \xi \rangle - \langle \eta \rangle \right) \ d\zeta\ d\xi \ \right) d\eta \nn \\
 &= \int_\xi \Psi(x;\xi) \int_\eta\left(\langle \xi \rangle - \langle \eta \rangle \right) \int_\zeta \overline{\Psi(\zeta;\xi)} \Psi(\zeta;\eta) \chi(\zeta) \mathcal{F}[f](\eta) \ d\eta  d\xi. \label{defA}
\end{align}

To represent the operator $R_{V_0}\DV \chi $ in terms of the distorted Fourier basis we note:
 \begin{align}
  \mathcal{F}[R_{V_0}\DV \chi g](\eta) &= \int_\eta \overline{\Psi(z;\eta)}\ (R_{V_0}\DV \chi g)(z) \ dz\nn \\
  &= \int_z (R_{V_0} \DV \overline{\Psi(z;\eta)}) \chi(z) g(z) \ dz \nn\\
  &= \int_z \frac{\langle \eta\rangle}{\eta^2-k^2}\ \overline{\Psi(z;\eta)} \chi(z)g(z) \ dz\nn \\
  &= \frac{\langle \eta\rangle}{\eta^2-k^2}\ \mathcal{F}[\chi g](\eta). \label{factor2}
 \end{align}

Combining~\eqref{defA} and~\eqref{factor2}, one has 
 \begin{align}
[\DV,\chi]R_{V_0}\DV \chi g(x) &= \int_\xi \Psi(x;\xi) \int_\eta\left(\langle \xi \rangle - \langle \eta \rangle \right) \int_\zeta \overline{\Psi(\zeta;\xi)} \Psi(\zeta;\eta) \chi(\zeta)\ d\zeta\nn \\ &\qquad \frac{\langle \eta\rangle}{\eta^2-k^2} \mathcal{F}[\chi g](\eta) \ d\eta \ d\xi \nn \\
&=\int_\xi \Psi(x;\xi) T^{II}[g](\xi) d\xi.
\end{align}
By the Plancherel Theorem, the $L^2$ estimate of $A_{II}^{(a)}$ is equivalent to the bound
\begin{equation}
 \left\| T^{II}[g] \right\|_{L^2}=\left\|\int_\eta \int_\zeta \overline{\Psi(\zeta;\xi)}\ \Psi(\zeta;\eta)\ \chi(\zeta)\ d\zeta\frac{\left(\langle \xi \rangle - \langle \eta \rangle \right)\langle \eta\rangle}{\eta^2-k^2} \ \mathcal{F}[\chi g](\eta) \ d\eta\right\|_{L^2_\xi} \le C \|g\|_{L^2}\label{toprove2}
 \end{equation}

We now proceed with a proof of~\eqref{toprove2}. First we define $\varphi_{|\kappa|<\delta_0}$ to be the positive smooth function satisfying
\begin{equation}\label{def-varphi}
 \varphi_{|\kappa|<\delta_0}\mbox{ equal to one for $|\kappa|\le\delta_0/2$, zero for $|\kappa|>\delta_0$ and symmetric about $\kappa=0$.}
\end{equation}
We use $\varphi$ to localize at frequencies near $\eta=\pm k$ and frequencies away from $\eta=\pm k$.

\begin{equation}
T^{II}[g]\ \equiv\ T^{II}_{near}[g]+T^{II}_{far}[g]
\label{TIIdecomp}
\end{equation}
where
\begin{align}
T^{II}_{far}[g](\xi)\ &\equiv\ \int_\eta K(\xi,\eta)\ \frac{\langle\eta\rangle}{\eta^2-k^2} \ \left[1-\varphi_{||\eta|-|k||<\delta_0}(\eta)\right] \mathcal{F}[\chi g](\eta)\ d\eta, \label{TIIfar}\\
 T^{II}_{near}[g](\xi)\ &\equiv\ \int_\eta K(\xi,\eta)\ \frac{\langle\eta\rangle}{\eta^2-k^2} \ \varphi_{||\eta|-|k||<\delta_0}(\eta)\ \mathcal{F}[\chi g](\eta)\ d\eta, \label{TIInear}
\end{align}
with $K$ defined as in~\eqref{Kdef} by
\[
K(\xi,\eta)\ \equiv\ \left(\langle \xi \rangle - \langle \eta \rangle \right)\
 \int_\zeta \overline{\Psi(\zeta;\xi)}\ \Psi(\zeta;\eta)\ \chi(\zeta)\ d\zeta.
\]

 \noindent {\bf Bound on $T^{II}_{far}[g](\xi)$:} We bound the expression
\begin{equation}
T^{II}_{far}[g](\xi)\ \equiv\ \int_\eta K(\xi,\eta)\ \frac{\langle\eta\rangle}{\eta^2-k^2}\ \left[1-\varphi_{||\eta|-|k||<\delta_0}(\eta)\right] \mathcal{F}[\chi g](\eta)\ d\eta. 
\end{equation}

By Lemma~\ref{lem:Kbound}, $K(\xi,\eta)$ satisfies the following pointwise bound, which is valid for all $\xi,\eta\in\RR$: 
\begin{align}
\left|\ K(\xi,\eta)\ \right|\ &\le\ C_\chi\ \frac{1}{1+|\xi-\eta |}. 
\nn\end{align}

Recall now the special case of Young's inequality: 
\begin{equation}
\| h\star g\|_2\le \| h\|_2\ \| g\|_1 \ .
\nn\end{equation}
This, together with the pointwise bound of $K(\xi,\eta)$, yields:
\begin{align}
\left\|\ T^{II}_{far}[g]\ \right\|_2\ &=\ \left\|\ \int K(\xi,\eta)\ \frac{\langle \eta\rangle}{\eta^2-k^2}\ \left[1-\varphi_{||\eta|-|k||<\delta_0}(\eta)\right]\ \left|\ \mathcal{F}[\chi g](\eta)\ \right|\ d\eta\ \ \right\|_{L^2_\xi}\nn\\
&\le\ C_\chi\ \left\| \frac{1}{\langle\eta\rangle}\ \right\|_{L^2_\eta}\  \left\| \frac{\langle \eta\rangle}{\eta^2-k^2} \left[1-\varphi_{||\eta|-|k||<\delta_0}(\eta)\right] \mathcal{F}[\chi g](\eta)\ \right\|_{L^1_\eta}\nn\\
&\le\ C_\chi\ \left\|\frac{1}{\langle\eta\rangle}\ \right\|_{L^2_\eta}^2\ \left\| \mathcal{F}[\chi g]\ \right\|_{L^2_\eta}\ \leq\ C_\chi\ 
 \left\|\ \chi g\ \right\|_{L^2_\eta}\nn\\
 & \le\ C_\chi\ \|\chi\|_{L^\infty}\ \|g\|_{L^2}.
\nn\end{align}

 \bigskip
 
 \noindent {\bf Bound on $T^{II}_{near}[g](\xi)$:}\bigskip

 \begin{align}
T^{II,\varepsilon}_{near}[g](\xi)\ &\equiv\ \int_\eta K(\xi,\eta)\ \frac{\langle \eta\rangle}{\eta^2-k^2} \varphi_{\varepsilon\le ||\eta|-|k||<\delta_0}(\eta)\ \mathcal{F}[\chi g](\eta)\ d\eta \nn\\
&= \frac{1}{2k}\ \int_\eta K(\xi,\eta)\ \langle \eta\rangle\ \varphi_{\varepsilon\le||\eta|-|k||<\delta_0}(\eta)\ \left( \frac{\mathcal{F}[\chi g](\eta)}{\eta-k}\ -\ \frac{\mathcal{F}[\chi g](\eta)}{\eta+k}\ \right) d\eta\nn\\
&\equiv \int \Lambda^\varepsilon(\xi,\eta)\ \frac{\mathcal{F}[\chi g](\eta)}{\eta-k} d\eta\ 
+\ \int \Lambda^\varepsilon(\xi,\eta)\ \frac{\mathcal{F}[\chi g](\eta)}{\eta+k}\ d\eta,\nn
\end{align}
where 
\begin{align}
\Lambda^\varepsilon(\xi,\eta)\ &\equiv\
\frac{1}{2k}\ \langle \eta\rangle\ \varphi_{\varepsilon\le||\eta|-|k||<\delta_0}(\eta)\ K(\xi,\eta).
\label{Lambda-def}
\end{align}
and $K(\xi,\eta)$ is displayed in~\eqref{Kdef}.
Note that by Lemma~\ref{lem:Kbound},
\begin{equation}\label{Lambda-bound}
\left|\ \Lambda^\varepsilon(\xi,\eta)\ \right|\ \le\ C_\chi \frac{1}{ 1+ |\xi-\eta|} \varphi_{\varepsilon\le||\eta|-|k||<\delta_0}(\eta). 
\end{equation}

We bound the first term in the above expansion of $T^{II,\varepsilon}_{near}$. The second term is treated similarly. We have

\begin{align*}
 \int \Lambda^\varepsilon(\xi,\eta)\ \frac{\mathcal{F}[\chi g](\eta)}{\eta-k} d\eta\
 &=\ \mathcal{S}^\varepsilon(\xi)\ +\ \mathcal{E}^\varepsilon(\xi) +\ \mathcal{R}^\varepsilon(\xi), \ \ {\rm where} \\
& \mathcal{S}^\varepsilon(\xi)\equiv \Lambda^\varepsilon(\xi,k)\ \int \frac{\mathcal{F}[\chi g](\eta)}{\eta-k}\ 
 \mathbf{1}_{\varepsilon\le |\eta-k|\le \delta_0/4}\ d\eta, \\
& \mathcal{E}^\varepsilon(\xi)\equiv \ \int \left( \Lambda^\varepsilon(\xi,\eta)\ -\  \Lambda^\varepsilon(\xi,k) \right) \frac{\mathcal{F}[\chi g](\eta)}{\eta-k}\ \mathbf{1}_{\varepsilon\le |\eta-k|\le \delta_0/4}\ d\eta,\\
& \mathcal{R}^\varepsilon(\xi)\equiv \ \int \Lambda^\varepsilon(\xi,\eta)\ \frac{\mathcal{F}[\chi g](\eta)}{\eta-k} \mathbf{1}_{ |\eta-k|\ge \delta_0/4}\ d\eta.
 \end{align*}
 
 One bounds $\mathcal{R}^\varepsilon$ using~\eqref{Lambda-bound} by
 \begin{align}
\|\ \mathcal{R}^\varepsilon\ \|_{L^2}\ &\le\ \frac{4 C_\chi}{\delta_0}\left\| \frac{1}{1+|\eta|}\ \right\|_{L^2_\eta} \|\varphi_{\varepsilon\le||\eta|-|k||<\delta_0}(\eta)\mathcal{F}[\chi g](\eta)\|_{L^1} \nn \\ 
&\le\ \frac{4 C_\chi}{\delta_0} \|\varphi_{\varepsilon\le||\eta|-|k||<\delta_0} \|_{L^2_\eta} \|\chi\|_{L^\infty}\ \|g\|_{L^2}.
\label{cRest}
\end{align}
 
 Moreover, we have
\begin{align}
\left|\ \frac{ \Lambda^\varepsilon(\xi,\eta)- \Lambda^\varepsilon(\xi,k)}{\eta-k}\ \right|\ &\le\
\left|\ \left. \partial_\eta\ \Lambda^\varepsilon(\xi,\eta)\ \right|_{\eta=\tilde\eta\in \{\varepsilon\le |\eta-k|\le \delta_0/4\}}\ \right|\ \nn \\
&\leq \mathbf{1}_{ |\eta-k|\le \delta_0}\left(\ \left|\ \partial_\eta\ K(\xi,\eta)\ \right|\ \langle \eta\rangle\ \ + \ \left| K(\xi,\eta)\right|\ \right).\end{align}
From the estimates of Lemma~\ref{lem:Kbound}, and using Young's inequality, one deduces
\begin{equation}
\|\ \mathcal{E}^\varepsilon\ \|_{L^2}\ \le\ C_\chi'\ \|\chi\|_{L^\infty} \|g\|_{L^2}.
\label{cEest}
\end{equation}

We treat the singular integral $\mathcal{S}^\varepsilon$ as follows. By antisymmetry of the function
 $(\eta-k)^{-1}\ \mathbf{1}_{\varepsilon\le |\eta-k|\le \delta_0/4}(\eta)$ we have
 \begin{align}
\int \mathbf{1}_{\varepsilon\le |\eta-k|\le \delta_0/4}\ \frac{1}{\eta-k}\ \mathcal{F}[\chi g](\eta)\ d\eta =\ 
  \int \mathbf{1}_{\varepsilon\le |\eta-k|\le \delta_0/4}\ \frac{\mathcal{F}[\chi g](\eta)-\mathcal{F}[\chi g](k)}{\eta-k}\ \ d\eta.
\nn\end{align}

Moreover, we have
\begin{align}
\left|\ \frac{\mathcal{F}[\chi g](\eta)-\mathcal{F}[\chi g](k)}{\eta-k}\ \right|\ &\le\
\left|\ \left. \partial_\eta\ \mathcal{F}[\chi g](\eta)\ \right|_{\eta=\tilde\eta\in \{\varepsilon\le |\eta-k|\le \delta_0/4\}}\ \right|.
\end{align}
By the uniform boundedness of $\langle \zeta\rangle^{-2}\partial_\eta m_+(\zeta,\eta)$ and $\langle \zeta\rangle^{-2}\partial_\eta m_-(\zeta,\eta)$ in $\RR\times\RR$, we have that
\begin{align*}\left| \partial_\eta\left(\ \mathcal{F}[\chi g](\eta)\ \right) \right | &= \left| \int_\zeta \partial_\eta \overline{\Psi(\zeta;\eta)} \chi(\zeta)g(\zeta)\ d\zeta \right | \\
 &\leq \sup_{(\zeta;\eta)\in\RR\times\RR}|\langle \zeta\rangle^{-2}\partial_\eta\ \Psi(\zeta;\eta)| \|\langle \zeta\rangle^{2}\chi\|_{L^2_\zeta}\|g\|_{L^2}
 \ \le\ C\ \|\langle \zeta\rangle^{2}\chi\|_{L^2_\zeta}\|g\|_{L^2}.
\end{align*}
Therefore,
\begin{equation}
\left|\ \int \mathbf{1}_{\varepsilon\le |\eta-k|\le \delta_0/4}(\eta)\ \frac{1}{\eta-k}\ \mathcal{F}[\chi g](\eta)\ d\eta \ \right|\ \le\ C\ \|\langle \zeta\rangle^{2}\chi\|_{L^2_\zeta}\|g\|_{L^2},
\nn\end{equation}
from which it follows that 
\begin{equation}
\left|\ \mathcal{S}^\varepsilon(\xi)\ \right|\ \le\ C\ \|\langle \zeta\rangle^{2}\chi\|_{L^2_\zeta}\|g\|_{L^2}\ \left| \Lambda^\varepsilon(\xi,k)\ \right| \ \le\ C_\chi \frac{1}{ 1+ |\xi-k|}\ \|g\|_{L^2}.\label{cSest}
\end{equation}

Thus we have from~\eqref{cRest},~\eqref{cEest} and~\eqref{cSest},
 \begin{equation}
 \left\|\ T^{II}_{near}[g]\ \right\|_2\ \le\ C_\chi'\ \|g\|_2 .
 \nn\end{equation}

Using the estimates of $T^{II}_{far}[g]$ and $T^{II}_{near}[g]$ yields~\eqref{toprove2}. Therefore, $A_{II}^{(a)}$ is bounded from $L^2$ to $L^2$. This completes the proof of Proposition~\ref{prop:A2b}.
\end{proof}

\bigskip

\begin{proposition}\label{prop:A2b} $A_{II}^{(b)}\equiv\chi\DV R_{V_0}(k))([\chi,\DV])$ is bounded from $L^2$ to $L^2$.
\end{proposition}
This follows from Proposition~\ref{prop:A2a} and duality. 

\bigskip

Finally, we consider the operator $A_{III}\equiv [\DV,\chi]\circ R_{V_0}(k)\circ [\chi,\DV]$.

\begin{proposition}\label{prop:A3}
The operator $A_{III}$ is bounded from $L^2$ to $L^2$.
\end{proposition}
\begin{proof}
By~\eqref{comm-rep}, one has
\begin{align}
A_{III}[g](x) &=\ \int_\xi \Psi(x;\xi) \int_\eta \left(\langle \xi \rangle - \langle \eta \rangle \right) \int_\zeta \overline{\Psi(\zeta;\xi)} \Psi(\zeta;\eta) \chi(\zeta)\ d\zeta \nn \\
&\qquad \frac{1}{\eta^2-k^2} \mathcal{F}\left[\ [\DV,\chi]g\ \right](\eta) \ d\eta \ d\xi \nn \\
&= \int_\xi \Psi(x;\xi) \int_\eta K(\xi,\eta)
\frac{1}{\eta^2-k^2} \mathcal{F}\left[\ [\DV,\chi]g\ \right](\eta) \ d\eta \ d\xi \nn \\
&= \int_\xi \Psi(x;\xi)\ \int_\eta K(\xi,\eta)\ \frac1{\eta^2-k^2}\ \int_\theta K(\eta,\theta) \mathcal{F}[g](\theta) \ d\theta \ d\eta \ d\xi
\end{align}
\bigskip

By the Plancherel Theorem, the $L^2$ estimate of $A_{III}$ is equivalent to the bound
\begin{equation}
 \left\|\ T^{III}[g]\ \right\|_{L^2}=\left\|\int_\eta K(\xi,\eta)\ \frac1{\eta^2-k^2}\ \int_\theta K(\eta,\theta) \mathcal{F}[g](\theta) \ d\theta \ d\eta \right\|_{L^2_\xi}\ \le\ C\ \|g\|_{L^2}.\label{toproveIII2}
 \end{equation}

We now proceed with a proof of~\eqref{toproveIII2}. 
We use $\varphi_{|\kappa|<\delta_0}$, defined as in~\eqref{def-varphi}, to localize at frequencies near $\eta=\pm k$ and frequencies away from $\eta=\pm k$.

\begin{equation}
T^{III}[g]\ \equiv\ T^{III}_{near}[g]+T^{III}_{far}[g]
\label{TIIIdecomp}
\end{equation}
where
\begin{align}
T^{III}_{far}[g](\xi)\ &\equiv\ \int_\eta K(\xi,\eta)\ \left(1-\varphi_{||\eta|-|k||<\delta_0}(\eta) \right) \ \frac1{\eta^2-k^2}\ \int_\theta K(\eta,\theta) \mathcal{F}[g](\theta) \ d\theta \ d\eta \label{TIIIfar}\\
 T^{III}_{near}[g](\xi)\ &\equiv\ \int_\eta K(\xi,\eta)\ \ \varphi_{||\eta|-|k||<\delta_0}(\eta)\ \ \frac1{\eta^2-k^2}\ \int_\theta K(\eta,\theta) \mathcal{F}[g](\theta) \ d\theta \ d\eta \label{TIIInear}
\end{align}

 \noindent {\bf Bound on $T^{III}_{far}[g](\xi)$:}\ \ We recall Lemma~\ref{lem:Kbound}, stating that
$K(\xi,\eta)$ satisfies the following upper bound:
\[
\left|\ K(\xi,\eta)\ \right|\ \leq \ C_\chi \frac{1}{ 1+ |\xi-\eta|},
\]
with $C_\chi\equiv C\ \left(\sum_{j=0}^2\ \left\|\langle \zeta\rangle^j\ \partial^j_\zeta\chi\right\|_{L^1_\zeta}+\left\|\langle \zeta\rangle^2\ \chi\right\|_{L^1_\zeta}+\left\|\langle \zeta\rangle^2\ \chi\right\|_{L^\infty}\right)$. Therefore, one has the pointwise estimate
\[ \left|\int_\theta K(\eta,\theta) \mathcal{F}[g](\theta) \ d\theta \right| \leq \left\|\frac{1}{ 1+ |\eta-\theta|} \right\|_{L^2_\theta}\ \left\|\mathcal{F}[g] \right\|_{L^2} \leq \left\|\frac{1}{ 1+ |\cdot|} \right\|_{L^2_\theta}\ \left\|g \right\|_{L^2}. \]
Moreover, for $||\eta|-|k||>\delta_0$, one has\ 
 $\left(1-\varphi_{||\eta|-|k||<\delta_0}(\eta)\right)\ \ |\eta^2-k^2|^{-1}\ \in\ L^1.$
 Therefore, by Young's inequality,
 \[\left\|T^{III}_{far}[g] \right\|_{L^2} \leq C \left\|\frac{1}{ 1+ |\cdot|} \right\|_{L^2}^2 \left\|\left(1-\varphi_{||\eta|-|k||<\delta_0}(\eta)\right) \frac1{|\eta^2-k^2|} \right\|_{L^1}\left\|g \right\|_{L^2} \leq C_\chi \left\|g \right\|_{L^2}. \]
 
 \bigskip
 
 \noindent {\bf Bound on $T^{III}_{near}[g](\xi)$:}
 
 \begin{align}
T^{III,\varepsilon}_{near}[g](\xi)\ &\equiv\ \int_\eta K(\xi,\eta)\ \ \varphi_{||\eta|-|k||<\delta_0}(\eta)\ \ \frac1{\eta^2-k^2}\ \int_\theta K(\eta,\theta) \mathcal{F}[g](\theta) \ d\theta \ d\eta \nn\\
&\equiv \int \Lambda^\varepsilon(\xi,\eta)\ \frac{1}{\eta-k}\int_\theta K(\eta,\theta)\mathcal{F}[ g](\theta)\ d\theta d\eta\ \nn\\
&\qquad +\ \Lambda^\varepsilon(\xi,\eta)\ \frac{1}{\eta+k}\int_\theta K(\eta,\theta)\mathcal{F}[ g](\theta)\ d\theta d\eta\ ,\nn
\end{align}
with $\Lambda^\varepsilon(\xi,\eta)\ \equiv\
\frac{1}{2k}\ K(\xi,\eta)\ \varphi_{\varepsilon\le||\eta|-|k||<\delta_0}(\eta)$.

Note that by Lemma~\ref{lem:Kbound},
\begin{equation}\label{Lambda-boundIII}
\left|\ \Lambda^\varepsilon(\xi,\eta)\ \right|\ \le\ C_\chi \frac{1}{ 1+ |\xi-\eta|} \varphi_{\varepsilon\le||\eta|-|k||<\delta_0}(\eta).
\end{equation}

We bound the first term in the above expansion of $T^{III,\varepsilon}_{near}$. The second term is treated similarly. We have
\begin{align*}
 \int \Lambda^\varepsilon(\xi,\eta)\ \frac{\mathcal{F}[\chi g](\eta)}{\eta-k} d\eta\
 &= \mathcal{S}^\varepsilon(\xi)\ +\ \mathcal{E}^\varepsilon(\xi) +\ \mathcal{R}^\varepsilon(\xi), \ \ {\rm where}\ \\
\mathcal{S}^\varepsilon(\xi)\equiv & \Lambda^\varepsilon(\xi,k)\ \int\ \frac{1}{\eta-k}\int_\theta K(\eta,\theta)\mathcal{F}[ g](\theta)\ d\theta \ 
 \mathbf{1}_{\varepsilon\le |\eta-k|\le \delta_0/4}\ d\eta, \\
 \mathcal{E}^\varepsilon(\xi)\equiv &\int \frac{\Lambda^\varepsilon(\xi,\eta) - \Lambda^\varepsilon(\xi,k)}{\eta-k}\int_\theta K(\eta,\theta)\mathcal{F}[ g](\theta)\ d\theta \ \mathbf{1}_{\varepsilon\le |\eta-k|\le \delta_0/4}\ d\eta,\\
 \mathcal{R}^\varepsilon(\xi)\equiv & \int \Lambda^\varepsilon(\xi,\eta)\ \frac{1}{\eta-k}\int_\theta K(\eta,\theta)\mathcal{F}[ g](\theta)\ d\theta\ \mathbf{1}_{ |\eta-k|\ge \delta_0/4}\ d\eta.
\end{align*}

As in the proof of Proposition~\ref{prop:A2a}, the kernel of the integral operators defining $\mathcal{E}^\varepsilon$ and $\mathcal{R}^\varepsilon$ are non-singular, and we have uniformly in $\epsilon$:
\[\|\mathcal{E}^\varepsilon\|_{L^2} + \|\mathcal{R}^\varepsilon\|_{L^2}\ \leq\ C_\chi'\|g\|_{L^2}. \]

We treat the singular integral $\mathcal{S}^\varepsilon$ as follows. By antisymmetry of the function
 $(\eta-k)^{-1}\ \mathbf{1}_{\varepsilon\le |\eta-k|\le \delta_0/4}(\eta)$ we have
 \begin{align}
\mathcal{S}(\xi)&=\ \Lambda(\xi,k)\int \mathbf{1}_{\varepsilon\le |\eta-k|\le \delta_0/4}(\eta)\ \int_\theta \frac{K(\eta,\theta)}{\eta-k}\ \mathcal{F}[ g](\theta)\ d\theta \ d\eta \nn\\
& =\ 
  \int \mathbf{1}_{\varepsilon\le |\eta-k|\le \delta_0/4}(\eta)\ \int_\theta \frac{K(\eta,\theta)-K(k,\theta)}{\eta-k}\ \mathcal{F}[ g](\theta)\ d\theta  \ d\eta.
\nn\end{align}

Moreover, Lemma~\ref{lem:Kbound} leads to
\begin{align}
\left|\ \frac{K(\eta,\theta)-K(k,\theta)}{\eta-k}\ \right|\ &\le\
\left|\ \left. \partial_\eta\ K(\eta,\theta)\ \right|_{\eta=\tilde\eta\in \{\varepsilon\le |\eta-k|\le \delta_0/4\}}\ \right|\
\leq \ C_\chi' \frac{1}{1+|\eta-\theta|}.\end{align}

Therefore, by Cauchy-Schwarz inequality,
\begin{equation}
\left| \int_\theta \frac{K(\eta,\theta)-K(k,\theta)}{\eta-k}\ \mathcal{F}[ g](\theta)\ d\theta \ \right| \le\ C_\chi' \|g\|_{L^2},
\nn\end{equation}
from which it follows that 
\begin{equation}
\left|\ \mathcal{S}^\varepsilon(\xi)\ \right|\ \le\ |\Lambda^\varepsilon(\xi,k)| \ \|\mathbf{1}_{\varepsilon\le |\eta-k|\le \delta_0/4}\|_{L^1_\eta} \ \|g\|_{L^2}\ \le\ C_\chi' \frac{1}{ 1+ |\xi-k|}\ \|g\|_{L^2}.
\end{equation}

Thus we have\ 
$
 \left\|\ T^{III}_{near}[g]\ \right\|_2\ \le\ C_\chi'\ \|g\|_2 .
$
\ Using the estimates of $T^{III}_{far}[g]$ and $T^{III}_{near}[g]$ yields~\eqref{toproveIII2}. Therefore, $A_{III}$ is bounded from $L^2$ to $L^2$. This completes the proof of Proposition~\ref{prop:A3}, and hence the proof of Proposition~\ref{prop:TR}.
\end{proof}

%

\bibliographystyle{siam} 
\bibliography{dw}

\end{document}